\documentclass[prx,twocolumn,floatfix,superscriptaddress,longbibliography]{revtex4-2}

\usepackage[dvips]{graphicx}
\usepackage{amsmath,amssymb,amsthm,amstext,mathrsfs,amsfonts,dsfont,mathtools,physics}
\usepackage{epsfig}
\usepackage{braket}
\usepackage[colorlinks=true,citecolor=blue,linkcolor=magenta]{hyperref}
\usepackage{bm}
\usepackage{enumerate}
\usepackage{color}
\usepackage{graphicx}
\usepackage{multirow}
\usepackage{tikz}
\usetikzlibrary{quantikz2}
\usepackage{appendix}
\usepackage{algpseudocode}
\usepackage{booktabs}
\usepackage{makecell}
\setcellgapes{1.3pt}
\usepackage{orcidlink}

\usepackage{enumitem}
\usepackage[capitalize]{cleveref}
\usepackage[normalem]{ulem}
\usepackage{comment}
\usepackage{xcolor}

\usepackage{amsthm}

\newtheorem{remark}{Remark}
\newtheorem{definition}{Definition}

\newtheorem{theorem}{Theorem}
\newtheorem{proposition}{Proposition}
\newtheorem{lemma}{Lemma}
\newtheorem{corollary}{Corollary}
\newtheorem{example}{Example}
\newtheorem{conjecture}{Conjecture}
\Crefname{theorem}{Theorem}{Theorems}
\theoremstyle{remark}

\newcounter{algorithm}
\renewcommand{\thealgorithm}{\arabic{algorithm}}

\newenvironment{myalgorithm}[1][]{%
\par\vspace{1.0em}
  \refstepcounter{algorithm}%
  \vspace{0.35em}
  \par\medskip
  \hrule
  \vspace{0.4em}
  \noindent\textbf{Algorithm~\thealgorithm. #1}\par
  \noindent\rule{\linewidth}{1.1pt}\par
  \vspace{0.4em}
}{%
  \vspace{0.8em}
  \hrule
  \par\vspace{1.0em}
}

\algnewcommand{\algorithmiccontinue}{\textbf{continue}}
\algnewcommand{\Continue}{\State \algorithmiccontinue}

\newcommand{\sgn}{\mathrm{sgn}}
\renewcommand{\tr}{\mathrm{Tr}}

\newcommand{\Ad}{\operatorname{Ad}}
\newcommand{\ad}{\operatorname{ad}}

\newcommand{\btheta}{\boldsymbol{\theta}}

\newcommand{\g}{\mathfrak{g}}
\def\gsim{\g{\text - }{\rm sim}} 
\newcommand{\poly}{\mathrm{poly}}

\newcommand{\qmaddress}{\affiliation{Quantum Motion Technologies Pty Ltd, Eveleigh NSW 2015, Australia}}
\newcommand{\melbaddress}{\affiliation{School of Physics, The University of Melbourne, Parkville VIC 3010, Australia}}

\begin{document}

\title{Enabling Lie-Algebraic Classical Simulation beyond Free Fermions}

\author{Adelina Bärligea\,\orcidlink{0009-0008-5497-1941}}
\email{adelina.baerligea@uni-a.de}
\affiliation{Institute for Computer Science, University of Augsburg, 86159 Augsburg, Germany}

\author{Matthew~L.~Sims-Goh\,\orcidlink{0000-0002-7478-4026}}
\qmaddress
\melbaddress

\author{Jakob~S.~Kottmann\,\orcidlink{0000-0002-4156-2048}}
\affiliation{Institute for Computer Science, University of Augsburg, 86159 Augsburg, Germany}
\affiliation{Center for Advanced Analytics and Predictive Sciences, University of Augsburg, 86159 Augsburg, Germany }

\begin{abstract}
Efficient classical simulation has matured to a critical component of the quantum computing stack, driving hardware validation, algorithm design, benchmarking, and the study of structured quantum dynamics. Lie-algebraic simulation ($\gsim$) offers a compelling approach: 
it represents Heisenberg-picture dynamics in the adjoint space whose dimension is set by the dynamical Lie algebra (DLA) governing the circuit, enabling efficient simulation of expectation values whenever the DLA grows only polynomially with system size. 
Despite this promise, existing applications of $\gsim$ have been confined to free-fermionic settings. It has therefore remained unclear if the method can be applied to other structured circuit families, especially when their generators have large Pauli expansions, and hence whether Lie-algebraic simulability presents a genuinely broader simulability paradigm than free fermions.
In this work, we resolve this question by identifying additional non-trivial families of polynomial-dimensional DLAs and introducing symmetry- and subspace-adapted bases that make the required adjoint-space preprocessing tractable. In particular, we develop an explicit Pauli orbit representation for permutation-equivariant dynamics, enabling efficient processing of cubic-dimensional algebras despite exponential Pauli support, and a modified generalized Gell--Mann representation for bounded Hamming-weight ($U(1)$-equivariant) dynamics, yielding polynomial simulation costs on fixed excitation sectors. Together with streamlined routines for free-fermionic Pauli algebras and translation-invariant subalgebras, these constructions significantly broaden the practical scope of $\gsim$ as a unifying simulation tool for structured quantum circuits. Numerical benchmarks confirm favorable preprocessing scaling and validate large-scale proof-of-concept simulations beyond the reach of state-vector simulation.
\end{abstract}

\maketitle

\section{Introduction}
\label{sec:introduction}
Efficient classical simulation has long served as a lens for interpreting quantum advantage: identifying restrictions under which quantum dynamics become tractable directly informs our understanding of the computational power of quantum devices. Today, this role is no longer only conceptual: classical simulation has become a load-bearing component of the quantum computing stack, including benchmarking, validation, and verification of near-term hardware~\cite{knill2008randomized,helsen2022matchgate,arute2019quantum}, designing meaningful quantum algorithms~\cite{tang2019quantum,gilyen2018quantum,chia2018quantum,chen2023quantum,Larocca.2022-GroupInvariantQML,Larocca.2022-DiagnosingBarrenPlateaus,Larocca.2023-TheoryOverparametrizationQuantum,Larocca.2025-BarrenPlateausVariational}, studying structured quantum many-body dynamics \cite{white1992density,vidal2003efficient,schollwock2005density,schollwock2011density,lieb1961two,baxter1985exactly}, data-loading and state preparation~\cite{kivlichan2018quantum,google2020hartree,fomichev2024initial,ran2020encoding,holmes2020efficient,gonzalez2024efficient,rudolph2024decomposition,berry2025rapid,huggins2025efficient,rupprecht2026sparse}, efficient measurement and read-out~\cite{Huang.2020-PredictingManyProperties,elben2023randomized,zhao2021fermionic,Wan.2023-MatchgateShadowsFermionic,yen2021cartan}, circuit synthesis and quantum compilation~\cite{aaronson2004improved,kissinger2020reducing,gibbs2025deep}, quantum error mitigation~\cite{strikis2021learning,czarnik2021error,arute2020observation,montanaro2021error}, and quantum error correction~\cite{Gottesman.1997-Stabilizer,fowler2012surface,ferris2014tensor,farrelly2021tensor,gidney2021stim}.
As classical simulation matures from a theoretical complexity guide to a practical tool in actual end-to-end quantum computing workflows, the criteria for its success have fundamentally shifted. For practitioners relying on these tools, the most practically relevant boundary is not ``which circuits are simulable in principle'' under asymptotic complexity theory, but rather ``which structured computations admit simulation with \emph{useful constant factors} and with \emph{faithful access} to expectation values that underlie many near-term routines?'' Consequently, a pressing need has emerged to make classical simulation fast, broadly applicable, and accessible in terms of reusable software. In this work, we pursue these three goals for quantum computations strongly constrained by symmetries.

A broad ecosystem of classical simulation paradigms has emerged, each exploiting a different notion of structure: the discrete-group structure of the Clifford group enables stabilizer simulations for low-magic circuits~\cite{Gottesman.1997-Stabilizer}, fermionic Gaussian structure enables the simulation of matchgate circuits~\cite{Valiant.2001-QuantumComputersThat,Knill.2001-FermionicLinearOptics,Terhal.2002-ClassicalSimulationNoninteractingfermion}, limited entanglement leads to efficient tensor-network simulations~\cite{Perez-Garcia.2007-MPSrepresentations,Markov.2008-SimulatingQCbyContractingTNs}, and Hamiltonian stoquasticity enables scalable Monte-Carlo methods~\cite{Jerrum.1993-PolyTimeApproximationIsingModel,Bravyi.2015-MonteCarloSimStoquasticHamiltonians}, among others~\cite{bartlett2002efficient,mari2012positive,Anschuetz.2023-EfficientClassicalAlgorithms,miller2025simulation,fontana2025classical,nakhl2025stabilizer,chang2026practicalframeworkpermutationequivariant}. 
These approaches do not form a strict hierarchy: they rely on structurally distinct resources such as discrete stabilizer structure, Gaussianity, low entanglement, or problem-specific symmetries. Lie-algebraic simulation adds a complementary compression principle, namely that the induced operator dynamics may close within a low-dimensional invariant algebra even when the Hilbert space itself is exponentially large.

Early work in Hamiltonian complexity and quantum control showed that if state preparation, gates, and observables are restricted to a compact Lie algebra of dimension polynomial in system size, then relevant expectation values can be computed efficiently by working in a faithful low-dimensional representation of that algebra, irrespective of the exponentially large Hilbert space~\cite{Somma.2005-QuantumComputationComplexity,Somma.2006-EfficientSolvabilityHamiltonians}. 
In modern language, the computationally relevant degrees of freedom live in an invariant operator subspace whose dimension is governed by the associated dynamical Lie algebra (DLA)~\cite{Goh.2025-LiealgebraicClassicalSimulations}. Thus, Lie-algebraic simulability should be viewed as complementary to stabilizer, tensor-network, Monte-Carlo, and free-fermionic methods: these regimes may overlap, but g-sim exploits closure of Heisenberg-picture dynamics in a polynomial-dimensional operator representation rather than low magic, low entanglement, positivity, or fermionic Gaussianity alone.

This viewpoint has gained renewed relevance through variational quantum algorithms (VQAs)~\cite{McClean.2016-TheoryVariationalHybrid,Cerezo.2021-VariationalQuantumAlgorithms,Bharti.2022-NoisyIntermediatescaleQuantum}, which prepare parametrized states $\rho(\boldsymbol{\theta})$ and optimize expectation value objectives of the form $\tr[O\,\rho(\boldsymbol{\theta})]$, with applications in quantum chemistry~\cite{Peruzzo.2014-VariationalEigenvalueSolver, Tilly.2022-VariationalQuantumEigensolver}, optimization~\cite{Farhi.2014-QuantumApproximateOptimization,Abbas.2024-ChallengesOpportunitiesQuantum}, and quantum machine learning (QML)~\cite{Biamonte.2017-QML,Cerezo.2022-QMLreview}. Over the past years, however, it has become clear that the large-scale trainability of many parametrized circuits is limited by phenomena such as barren plateaus~\cite{McClean.2018-BarrenPlateausQuantum,Larocca.2025-BarrenPlateausVariational} and spurious local minima~\cite{You.2021-ExponentiallyManyLocal,Anschuetz.2022-QuantumVariationalAlgorithms}. A key conceptual advance is that both expressivity and gradient statistics can be characterized in terms of the DLA associated with the circuit generators, linking VQA trainability to tools from quantum control theory~\cite{DAlessandro.2021-IntroductionQuantumControl}. In particular, building on a conjecture~\cite{Larocca.2022-DiagnosingBarrenPlateaus}, independent works~\cite{Ragone.2024-LieAlgebraicTheory,Fontana.2024-CharacterizingBarrenPlateaus} established DLA-based variance formulas for broad classes of cost functions in suitable large-depth regimes.

This convergence of ideas suggests a unified simulability narrative. In many analyses of VQAs, the same algebraic restrictions that prevent barren plateaus also identify polynomially-sized operator subspaces in which expectation values admit efficient classical evaluation. This intuition underlies a prominent recent conjecture relating provable trainability and classical simulability~\cite{Cerezo.2025-DoesProvableAbsence} and motivates Lie-algebraic simulation as both a diagnostic and a computational primitive.

A concrete instantiation of Lie-algebraic simulation for quantum mean values is the $\gsim$ framework introduced in Ref.~\cite{Goh.2025-LiealgebraicClassicalSimulations}. At a high level, $\gsim$ replaces evolution in a $2^n$-dimensional Hilbert space by evolution in the $\dim(\mathfrak{g})$-dimensional adjoint space induced by the DLA $\mathfrak{g}$, enabling exact expectation values and gradients for supported observables, provided the relevant invariant operator subspace is polynomial-dimensional and its coordinate representation can be constructed efficiently. While the evaluation stage is then efficient by construction, the practical bottleneck is often \emph{preprocessing}: constructing a suitable basis representation of the DLA (or an invariant subspace) and computing the adjoint data (structure constants) required for adjoint-space propagation.

A broad range of practical applications of the $\gsim$ framework have been proposed and explored, including warm-starting of variational training~\cite{zering2025benchmarking,Goh.2025-LiealgebraicClassicalSimulations}, privacy vulnerabilities in QML~\cite{Heredge.2025-CharacterizingPrivacyQuantum}, reducing gradient estimation costs~\cite{bhowmick2025enhancing}, quantum metrology~\cite{Lecamwasam.2024-QuantumMetrologyLinear, goh2024protocols}, and quantum circuit synthesis~\cite{Goh.2025-LiealgebraicClassicalSimulations}. However, at present there are two intertwined gaps that limit the perceived scope of Lie-algebraic simulation in practice:
First, essentially all explicit, large-scale demonstrations of quantum circuits simulated with $\gsim$ to date have been in settings that are free-fermion equivalents~\cite{Goh.2025-LiealgebraicClassicalSimulations,zering2025benchmarking,bhowmick2025enhancing}. As a result, the practical evidence for $\gsim$ has so far largely coincided with regimes already known to be classically tractable through fermionic linear optics or other matchgate techniques. This has contributed to the perception that Lie-algebraic simulation is, in practice, mainly a reformulation of free-fermion simulability, which was further reinforced by classification results for Pauli Lie algebras showing that the only polynomial-dimensional families in this setting are essentially free-fermionic~\cite{Aguilar.2024-FullClassificationPauli}. Thus, while the formal $\gsim$ framework is more general, it has remained unclear which other structured generator families give rise to polynomial-dimensional Lie algebras that can be exploited for efficient classical simulation.
Second, even when such a polynomial-dimensional DLA is identified, it is unclear a priori whether preprocessing remains efficient. Generators and basis elements may expand into exponentially many Pauli strings, making naive commutator and inner-product computations in the Pauli basis intractable. This concern has recently been articulated as a conjectured additional requirement beyond $\dim(\mathfrak{g})=\poly(n)$, namely that basis elements admit ``slow Pauli expansions'' meaning Pauli basis decompositions with only polynomially many nonzero coefficients~\cite{Heredge.2025-CharacterizingPrivacyQuantum}. If true, such a condition would exclude highly symmetric sum generators from practical Lie-algebraic simulation despite their polynomial DLA dimension.

The purpose of this work is to show that Lie-algebraic simulability is meaningfully broader than the free-fermion/matchgate regime, and that the correct perspective is representation-theoretic: polynomial DLA dimension is the structural starting point, but efficient simulation hinges on choosing a basis adapted to the algebraic symmetries of the system. Thus, Pauli sparsity is not the fundamental resource as previously conjectured~\cite{Heredge.2025-CharacterizingPrivacyQuantum}. The relevant resource is the existence of a representation in which commutators, inner products, and coordinate maps can be evaluated efficiently. Concretely, we develop symmetry- and subspace-adapted operator bases and efficient primitives that make $\gsim$ preprocessing practical for several families of polynomial-dimensional DLAs that go well beyond free fermions.

\begin{figure*}[htbp]
\includegraphics[width=1\linewidth]{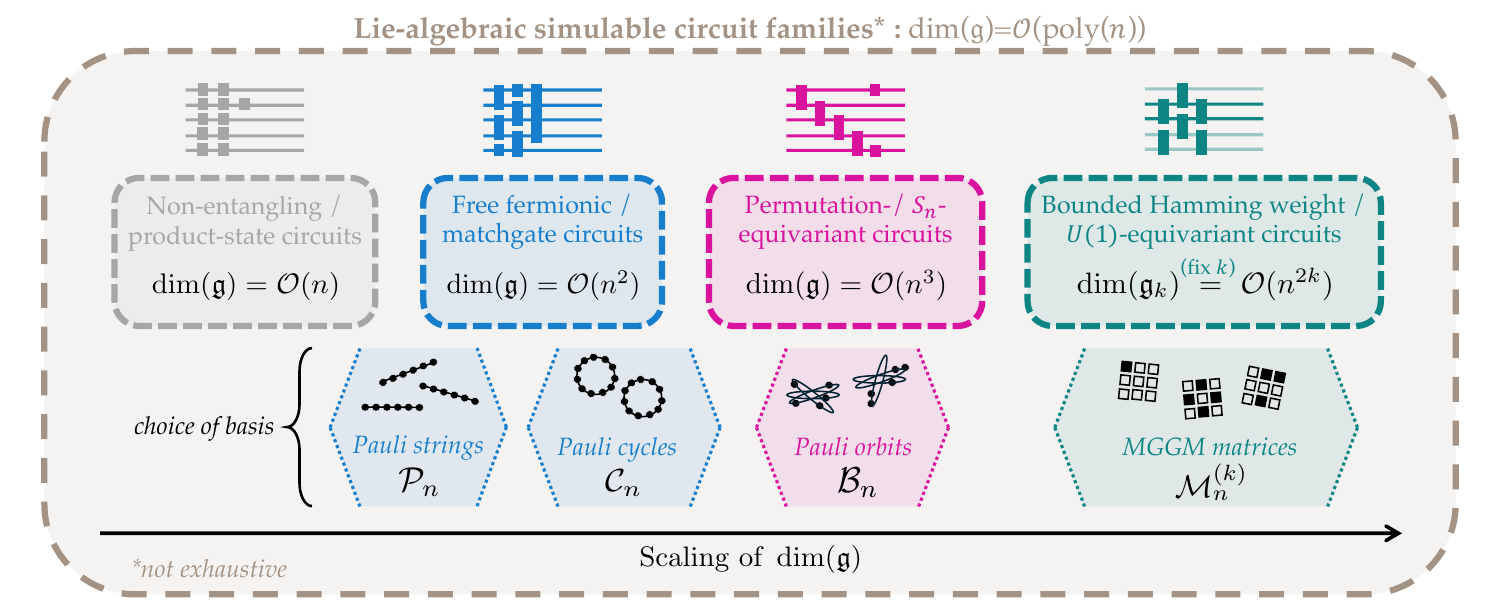}
\caption{\label{fig:summary}Summary of Lie-algebraically simulable circuit families and the symmetry-adapted bases developed in this work. The key structural condition exploited by $\mathfrak{g}$-sim is that the dynamical Lie algebra $\mathfrak{g}$ generated by the circuit has dimension $\dim(\mathfrak{g})=\mathrm{poly}(n)$, so that expectation values can be propagated in the $\dim(\mathfrak{g})$-dimensional adjoint representation instead of the $2^n$-dimensional Hilbert space. The top row highlights representative systems with polynomial scaling DLAs, as investigated in this paper: (i) non-entangling/product-state circuits with $\dim(\mathfrak{g})=\mathcal{O}(n)$ (included for completeness), (ii) free-fermionic/matchgate circuits with $\dim(\mathfrak{g})=\mathcal{O}(n^2)$ (\cref{sec:FFalgebras}), (iii) permutation-/$S_n$- equivariant (collective) circuits with $\dim(\mathfrak{g})=\mathcal{O}(n^3)$ (\cref{sec:Pauliorbitalgebra}), and (iv) bounded Hamming-weight / $U(1)$-equivariant circuits whose effective algebra on a fixed excitation sector $\mathcal{H}_k$ has $\dim(\mathfrak{g}_k)=\mathcal{O}(n^{2k})$ for fixed $k$ (\cref{sec:restrictedsubspace-algebras}). The bottom row emphasizes another important contribution: polynomial dimension alone does not specify an efficient simulation pipeline unless the algebra is represented in a basis that makes the adjoint space mapping (via commutators and inner products) tractable. Concretely, we provide symmetry-adapted bases and primitives enabling efficient $\mathfrak{g}$-sim preprocessing in the respective regimes: Pauli strings $\mathcal{P}_n$ for free-fermion Pauli algebras; Pauli cycles $\mathcal{C}_n$ for translation-invariant generator settings; Pauli orbits $\mathcal{B}_n$ for permutation-invariant algebras, and a subspace-adapted modified generalized Gell–Mann (MGGM) basis $\mathcal{M}_n^{(k)}$ for bounded-HW dynamics with constant-time commutator evaluations. Together, these constructions enable Lie-algebraic classical simulation beyond free fermions and illustrate how the `correct' basis turns $\dim(\mathfrak{g})=\mathrm{poly}(n)$ into an actionable, practically efficient simulation paradigm.}
\end{figure*}

Our contributions can be summarized as follows:

\paragraph{Beyond free fermions: two further polynomial regimes with explicit bases.}
Building on the sharp dichotomy for Pauli Lie algebras generated by individual Pauli strings~\cite{Aguilar.2024-FullClassificationPauli,Wiersema.2024-ClassificationDynamicalLie,Kokcu.2024-ClassificationDynamicalLie}, we emphasize that this dichotomy is not a dichotomy for Lie-algebraic simulability itself. It applies to a specific class of generator sets, but not to general Hermitian generators whose structure may be encoded in correlated sums of Pauli strings. We identify and treat two broad classes of non-free-fermionic polynomial DLAs arising from structured \emph{sum generators}:
(i) permutation-equivariant (collective) dynamics, whose ambient commutant algebra has cubic dimension scaling, and
(ii) Hamming-weight preserving ($U(1)$-equivariant) dynamics restricted to fixed-weight sectors, whose effective algebra on the relevant subspace has polynomial dimension for bounded Hamming weights and can exceed the quadratic and cubic regimes.
Thus, their simulability is not inherited from fermionic Gaussian structure, but arises from different algebraic resources: collective symmetry and fixed-subspace structure.

\paragraph{Symmetry-adapted representations that remove exponential Pauli support as an obstacle.}
For translation-invariant free-fermion equivalents, we introduce Pauli cycles as a symmetry-adapted representation that reduces commutator costs by exploiting cyclic twirling.
For permutation-equivariant systems, we develop the Pauli orbit basis and derive efficient inner-product and commutator primitives that close within the orbit representation with polynomial resources, despite exponential Pauli expansions of individual orbit elements. This directly refutes the necessity of the slow Pauli expansion assumption~\cite{Heredge.2025-CharacterizingPrivacyQuantum} in this setting.
For bounded Hamming-weight sectors, we propose a modified generalized Gell--Mann basis for the full operator algebra on the effective polynomial-dimensional subspace and provide the analytic commutation relations enabling efficient structure constant generation.

\paragraph{Efficient preprocessing with explicit complexity and benchmark evidence.}
Across these regimes, we focus on the preprocessing bottleneck of $\gsim$: Lie closure (basis construction) and adjoint representation (structure constant computation). We give explicit primitives and worst-case scaling bounds in the relevant representations and complement them with runtime benchmarks demonstrating that constant factors and average-case behavior render preprocessing practical in regimes far beyond what naive Pauli expansions would suggest.

\paragraph{Large-scale numerical demonstrations as proof of concept.}
We provide numerical experiments that reproduce and scale up representative tasks previously limited to state-vector simulation or small hardware demonstrations, including free-fermionic TFIM-based optimization and permutation-equivariant quantum neural networks with symmetry-adapted preprocessing. These experiments are intended as feasibility demonstrations of the theory and the representations introduced here.

\paragraph{Open-source implementation.}
To facilitate reuse and independent verification, we provide a highly optimized, open-source implementation of all proposed bases and primitives introduced in this work, together with code to reproduce the numerical experiments. These can be found in the GitHub repository of Ref.~\cite{github_repo}.

Taken together, these results sharpen the role of Lie-algebraic simulation as a unifying simulability paradigm that extends beyond the free-fermion setting. The central lesson is that $\dim(\mathfrak{g})=\poly(n)$ is not, by itself, an implementation; it is a structural promise that becomes practically useful only after one finds a representation in which basis construction, commutators, inner products, and adjoint propagation can all be performed efficiently. The symmetry-adapted bases developed here show how this promise can be realized in non-free-fermionic regimes, including collective permutation-equivariant dynamics and bounded Hamming-weight sectors (cf.~\cref{fig:summary}). In this sense, the work reframes $\gsim$ from a reformulation of known free-fermionic simulability into a representation-theoretic scheme for uncovering new efficient classical simulations.

After recalling Lie-algebraic characterizations of parametrized circuits and the $\gsim$ framework in \cref{sec:preliminaries}, we discuss the landscape of polynomial-dimensional DLAs and the representation bottleneck beyond Pauli string generators in \cref{sec:polyDLAsystems}. We then treat, in increasing algebraic complexity, free-fermion algebras and symmetry-restricted free-fermion equivalents (\cref{sec:FFalgebras}), permutation-invariant algebras (\cref{sec:Pauliorbitalgebra}), and bounded Hamming-weight algebras (\cref{sec:restrictedsubspace-algebras}). Numerical experiments are presented in \cref{sec:NumericalExamples}, and supporting proofs and additional benchmarks are deferred to the Appendices.

\section{Preliminaries}
\label{sec:preliminaries}
Lie-algebraic methods have long provided a structural route to classical simulation and exact solvability of restricted quantum dynamics. Early work by Somma \emph{et~al.}\ showed that quantum computational models whose gates and observables are confined to a compact, polynomial-dimensional Lie algebra can be exactly simulated classically with resources polynomial in the Lie-algebra dimension, even when the underlying Hilbert space is exponentially large~\cite{Somma.2005-QuantumComputationComplexity,Somma.2006-EfficientSolvabilityHamiltonians}. More recently, the same paradigm has been reintroduced in the context of variational quantum algorithms (VQAs), where the central computational primitive is the repeated evaluation of expectation values of circuit-evolved states~\cite{Goh.2025-LiealgebraicClassicalSimulations}.

Variational quantum algorithms~\cite{McClean.2016-TheoryVariationalHybrid,Cerezo.2021-VariationalQuantumAlgorithms,Bharti.2022-NoisyIntermediatescaleQuantum} aim to approximately solve ground-state and related optimization problems using parameterized quantum circuits on near-term quantum devices. Given an initial state $\rho_{\mathrm{in}}$ and a parameterized ansatz circuit $U(\boldsymbol{\theta})$, VQAs estimate the objective value
\begin{equation}
\ell_{\rho_{\mathrm{in}},O}(\boldsymbol{\theta})
=
\langle O(\boldsymbol{\theta})\rangle
:=
\tr\!\big[\,O\,U(\boldsymbol{\theta})\,\rho_{\mathrm{in}}\,U^\dagger(\boldsymbol{\theta})\,\big],
\label{eq:expectationvalue}
\end{equation}
where the problem is encoded in the Hamiltonian $O$ (i.e., a Hermitian observable) whose ground-state energy is sought. By the variational principle, minimizing $\ell_{\rho_{\mathrm{in}},O}$ over $\boldsymbol{\theta}$ yields an upper bound on the true ground-state energy. More general hybrid objectives often reduce to repeated evaluations of this same primitive: for instance, in quantum machine learning~\cite{Cerezo.2022-QMLreview} one typically considers families of input states $\rho_{\mathrm{in}}(x)$ indexed by data and forms a loss by averaging over a dataset, while in other settings one classically post-processes a finite collection of measured expectation values, correlators, or probabilities into a nonlinear objective. Thus, the efficient computation of expectation values remains the core quantum subroutine underlying a broad class of hybrid variational models.

Despite their algorithmic promise, VQAs face well-known scalability barriers. Beyond the detrimental effects of hardware noise~\cite{StilckFranca.2021-LimitationsOptimizationAlgorithms,Wang.2021-NoiseinducedBarrenPlateaus,Gonzalez-Garcia.2022-ErrorPropagationNISQ,DePalma.2023-LimitationsVariationalQuantum,Schuster.2024-PolynomialtimeClassicalAlgorithm} and finite-shot estimation errors~\cite{Scriva.2024-ChallengesVariationalQuantum,Barligea.2025-ScalabilityChallengesVariational}, a major obstacle is \emph{trainability}: many ansätze and cost functions exhibit barren plateaus~\cite{McClean.2018-BarrenPlateausQuantum,Larocca.2025-BarrenPlateausVariational}, i.e.\ exponentially suppressed gradient variance and concentrated loss at scale, and their optimization landscapes are covered by suboptimal local minima~\cite{Bittel.2021-TrainingVariationalQuantum,You.2021-ExponentiallyManyLocal,Anschuetz.2022-QuantumVariationalAlgorithms}. 
A key development of recent Lie-algebraic treatments is that both expressivity and trainability of deep circuit families can be characterized in terms of the dynamical Lie algebra (DLA) generated by the circuit gates~\cite{Larocca.2022-DiagnosingBarrenPlateaus,Ragone.2024-LieAlgebraicTheory,Fontana.2024-CharacterizingBarrenPlateaus}. More specifically, these works tie barren plateaus to a curse of dimensionality in operator space: when the relevant dynamics spread over exponentially large Lie-algebraic sectors, gradients generically concentrate, whereas polynomial-dimensional invariant subspaces can avoid this source of concentration provided the input states and observables have sufficient support on those sectors, such that one may obtain regimes with provable absence of barren plateaus. Importantly, those same polynomially-sized subspaces are precisely the ones that can be exploited for classical simulation, providing a concrete instance of the broader conjectural link between trainability and simulability~\cite{Cerezo.2025-DoesProvableAbsence}.

The remainder of this section collects the Lie-algebraic notions and notation used throughout the paper. In \cref{ssec:LieCharVQAs} we introduce the DLA formalism and its role in characterizing variational circuits. In \cref{ssec:gsimTheory} we summarize Lie-algebraic simulation of expectation values and highlight the fundamental computational primitives (basis and structure constants) that determine practical efficiency. In \cref{ssec:notionSimulability} we clarify important notions of classical simulability relevant to this work.

\subsection{Lie-Algebraic characterization of variational quantum circuits}
\label{ssec:LieCharVQAs}
For a system of $n$ qubits with Hilbert space ${\mathcal{H}{=}(\mathbb{C}^2)^{\otimes n}}$, consider a layered (periodic) parameterized circuit of the form
\begin{equation}
     U(\btheta)=\prod_{l=1}^L U_l(\btheta_l)\quad \mathrm{with}\quad U_l(\btheta_l)= \prod_{k=1}^K e^{-i\:\theta_{l,k}\:H_k },
    \label{eq:standardansatz}
\end{equation}
where $G=\{H_k\}_{k=1}^K$ is a fixed set of Hermitian, traceless generators, and $\boldsymbol{\theta}=\{\theta_{l,k}\}$ collects all real parameters. The dynamical Lie algebra (DLA) associated with $G$ is defined as the smallest real Lie subalgebra of $\mathfrak{su}(2^n)$ containing $iH_k$, i.e.,
\begin{equation}
    \g= \operatorname{span}_\mathbb{R}\langle iH_1,\dots,iH_K\rangle_\text{Lie}\subseteq \mathfrak{su}(2^n)
\label{eq:DLAdefinition}
\end{equation}
where $\langle\cdot\rangle_\text{Lie}$ denotes closure under real linear combinations and commutators. Equivalently, $\g$ is the real span of all nested commutators of the elements in $iG$. 

The connected Lie subgroup associated with $\g$ is
\begin{equation}
e^\g\equiv\{e^{iA}\ :\ iA \in \g\}\subseteq \mathrm{SU}(2^n).
\end{equation}
In the language of quantum control, $\g$ captures the infinitesimal directions generated by the available Hamiltonians, and $e^\g$ describes the corresponding connected Lie group of unitaries reachable in principle when one concatenates evolutions under the generators with freely chosen durations and orderings; see, e.g.\ Ref.~\cite{DAlessandro.2021-IntroductionQuantumControl}. Equivalently, from a more geometric viewpoint, $\g$ is the Lie algebra of $e^\g$, i.e.\ the tangent space to the reachable Lie group at the identity, and therefore gives the local directions in which the circuit can move through unitary space. For a fixed ansatz architecture such as \eqref{eq:standardansatz}, the set of unitaries reachable at a given finite depth $L$ is generally a strict subset of $e^\g$ (which holds for $L\rightarrow\infty$), but $\g$ remains the natural algebraic object controlling expressivity and related phenomena.

For a continuously parametrized gate set of the form considered here, the generators are universal if the Lie algebra generated by their infinitesimal generators is $\mathfrak{su}(2^n)$, or equivalently if $e^\g=\mathrm{SU}(2^n)$~\cite{Lloyd.1996-UniversalQuantumSimulators}. Moreover, universality is generic in the sense that “almost any” sufficiently noncommuting choice of generators yields $\g=\mathfrak{su}(2^n)$~\cite{Lloyd.1995-AnyGateIsUniversal}. A natural basis of $\g=\mathfrak{su}(2^n)$ is given by the traceless Pauli strings. Let
\begin{equation}
\mathcal{P}_n
:=
\bigl\{
P^{(n)}=\bigotimes_{i=1}^n P_i \;:\; P_i\in\{I,X,Y,Z\},\ P^{(n)}{\neq} I^{\otimes n}
\bigr\}.
\label{eq:Paulistringbasis}
\end{equation}
Then $\{iP:\,P\in\mathcal{P}_n\}$ is a real basis of $\mathfrak{su}(2^n)$ and $|\mathcal{P}_n|=4^n-1$. In this work we are primarily interested in the restricted regime of proper subalgebras $\g\subsetneq\mathfrak{su}(2^n)$ whose dimension grows only polynomially with $n$. 

A central insight of the recent theory on the trainability of VQAs is that the DLA $\mathfrak{g}$ of a system provides a structural descriptor of both expressivity and loss variance statistics of parameterized quantum circuits. Building on a conjecture~\cite{Larocca.2022-DiagnosingBarrenPlateaus}, Refs.~\cite{Ragone.2024-LieAlgebraicTheory,Fontana.2024-CharacterizingBarrenPlateaus} proved that, under suitable large-depth assumptions, the variance of a wide class of cost functions admits the universal form
\begin{equation}
\operatorname{Var}_{\boldsymbol{\theta}}\!\left[\ell_{\rho,O}(\boldsymbol{\theta})\right]
=
\frac{\mathcal{P}_{\mathfrak{g}}(\rho)\,\mathcal{P}_{\mathfrak{g}}(O)}{\dim(\mathfrak{g})},
\label{eq:varianceformula}
\end{equation}
Here, $\ell_{\rho,O}(\boldsymbol{\theta})$ denotes the parameterized cost associated with an input state $\rho$ and an observable $O$ (cf.\ expectation value~\eqref{eq:expectationvalue}). The quantity $\mathcal{P}_{\mathfrak{g}}(\cdot)$ is the \emph{$\mathfrak{g}$-purity} of a Hermitian operator~\cite{Barnum.2003-GeneralizationsEntanglementBased,Barnum.2004-SubsystemIndependentGeneralizationEntanglement,Ragone.2024-LieAlgebraicTheory}. Concretely, let $\{B_\alpha\}_{\alpha=1}^{\dim(\g)}$ be a Hilbert--Schmidt orthonormal basis of $i\g$ with respect to the Hilbert--Schmidt inner product $\langle A,B\rangle_{\mathrm{HS}}=\tr(AB)$. Then,
\begin{equation}
\mathcal{P}_{\mathfrak{g}}(H)
:= 
\sum_{\gamma=1}^{\dim(\mathfrak{g})}\big|\tr(B_\gamma^\dagger H)\big|^2
\ \le\ \tr(H^2).
\label{eq:g-purity}
\end{equation}
Intuitively, $\mathcal{P}_{\mathfrak g}(H)$ measures how much of $H$ lies in the Hermitian operator space $i\mathfrak g$, equivalently in the tangent space to the dynamical Lie group generated by the circuit.

Equation~\eqref{eq:varianceformula} provides a unifying explanation of several known mechanisms behind barren plateaus: increasing expressivity~\cite{McClean.2018-BarrenPlateausQuantum} corresponds to increasing $\dim(\mathfrak{g})$ and hence suppresses the variance; strong entanglement or symmetry mismatch can reduce $\mathcal{P}_{\mathfrak{g}}(\rho)$; and observables with little support in $i\g$ have small $\mathcal{P}_{\mathfrak{g}}(O)$, consistent with cost-function locality effects~\cite{Cerezo.2021-CostFunctionDependent}.

Since expressivity of circuits of the form~\eqref{eq:standardansatz} depends strongly on the depth $L$, the validity of \cref{eq:varianceformula} likewise requires a sufficiently mixed, large-depth regime. In Refs.~\cite{Larocca.2022-DiagnosingBarrenPlateaus,Ragone.2024-LieAlgebraicTheory,Fontana.2024-CharacterizingBarrenPlateaus}, this is formalized by assuming that the circuit ensemble induced by random parameters forms an (approximate) unitary $2$-design on the relevant dynamical group (or, more generally, that it mixes sufficiently in the second Haar moment; see Ref.~\cite{Mele.2024-IntroductionHaarMeasure}). 
In the fully controllable case, $\mathfrak{g}\cong \mathfrak{su}(2^n)$, one obtains an explicit sufficient depth condition of the form
\begin{equation}
L \ \gtrsim\ \frac{\log(1/\varepsilon)}{\log\!\left(1/\big\|\mathcal{A}^{(2)}\big\|_{\infty}\right)},
\label{eq:depth_bound_design}
\end{equation}
where $\mathcal{A}^{(2)}$ captures the contraction of the second moment (twirling) channel associated with one circuit layer and $\|\cdot\|_\infty$ denotes the induced operator norm on the traceless subspace~\cite{Larocca.2022-DiagnosingBarrenPlateaus,Ragone.2024-LieAlgebraicTheory}. In practice, evaluating or tightly bounding such moment operators for a given ansatz is highly nontrivial~\cite{Harrow.2023-ApproximateUnitaryTDesigns,Braccia.2024-ComputingExactMoments}, so $\g$-based characterizations are often used either in analytically controlled regimes or as qualitative predictors complemented by numerics.

Finally, beyond barren plateaus, the DLA dimension also governs other aspects of the loss landscape geometry: underparameterization can induce spurious local minima, whereas avoiding such artifacts typically requires $\mathcal{O}(\dim(\g))$ parameters (overparameterization)~\cite{Larocca.2023-TheoryOverparametrizationQuantum}.

\subsection{Lie-algebraic simulation of expectation values}
\label{ssec:gsimTheory}
The Lie-algebraic structure captured by the DLA not only underpins recent analyses of expressivity and trainability of variational circuits, but also enables efficient classical simulation in regimes where the induced algebra remains ``small'' (i.e., polynomially growing in system size); e.g., the discussion around the trainability--simulability conjecture in \cite{Cerezo.2025-DoesProvableAbsence}. A general and practically-oriented framework is $\gsim$~\cite{Goh.2025-LiealgebraicClassicalSimulations}, which builds on earlier ideas in quantum control and Hamiltonian complexity~\cite{Somma.2005-QuantumComputationComplexity,Somma.2006-EfficientSolvabilityHamiltonians,Zeier.2011-Symmetryprinciplesquantumsystems}. We briefly review the basic principles here.

Given an initial state $\rho_{\mathrm{in}}$ and an observable $O$, the simulation task is to evaluate the expectation value after evolution under the circuit $U(\boldsymbol{\theta})$, which is exactly the VQA loss defined in~\cref{eq:expectationvalue}. The key assumption in Lie-algebraic simulation is that the observables to be propagated are supported on a low-dimensional operator subspace invariant under conjugation by the dynamical Lie group $e^\g$, and that the input state admits efficient evaluation of its overlaps with a basis of this subspace. For clarity of exposition we present $\gsim$ in the simplest setting here, where the relevant invariant subspace is taken as $i\g$, i.e.\ the Hermitian operators generated by the DLA itself. More generally, however, $\gsim$ applies to arbitrary (irreducible) invariant subspaces of the adjoint action of $e^\g$; see Ref.~\cite{Goh.2025-LiealgebraicClassicalSimulations}.

Fix a Hilbert--Schmidt--orthonormal basis $\{B_\alpha\}_{\alpha=1}^{d}$ of $i\g$, where $d{:=}\dim(\g)$. We expand the observable in this basis and encode the input state by its overlaps with the basis elements:
\begin{equation}
O=\sum_{\alpha} w_\alpha\, B_\alpha,
\qquad
\big(\mathbf{e}^{(\mathrm{in})}\big)_\alpha := \tr[B_\alpha\,\rho_{\mathrm{in}}],
\label{eq:O_and_ein}
\end{equation}
so that by linearity of the trace,
\begin{equation}
\begin{gathered}
\langle O(\boldsymbol{\theta})\rangle
=\sum_{\alpha} w_\alpha\,\big(\mathbf{e}^{(\mathrm{out})}\big)_\alpha,
\\ \ \text{with}\quad
\big(\mathbf{e}^{(\mathrm{out})}\big)_\alpha := \tr\!\big[U^\dagger(\boldsymbol{\theta})\,B_\alpha\,U(\boldsymbol{\theta})\,\rho_{\mathrm{in}}\big].
\label{eq:expval_dotproduct}
\end{gathered}
\end{equation}

For any $U\in e^{\mathfrak{g}}$, its adjoint action on $i\g$ is defined by conjugation $B\mapsto U^\dagger B U$. In the basis $\{B_\alpha\}_{\alpha=1}^d$ this defines a real matrix $\Phi^{\Ad}(U)\in\mathbb{R}^{d\times d}$ via
\begin{equation}
U^\dagger B_\alpha U
=\sum_{\beta}\big(\Phi^{\Ad}(U)\big)_{\alpha\beta}\,B_\beta.
\label{eq:adjointaction_correct}
\end{equation}
Taking the Hilbert-Schmidt inner product with $\rho_{\mathrm{in}}$ yields a linear evolution rule for the output vector:
\begin{equation}
\mathbf{e}^{(\mathrm{out})}=\Phi^{\Ad}\!\big(U(\boldsymbol{\theta})\big)\,\mathbf{e}^{(\mathrm{in})}.
\label{eq:eout_adjoint}
\end{equation}

The corresponding infinitesimal action is determined by the \emph{structure constants} of the basis. We define them by
\begin{equation}
i[B_\alpha,B_\beta]
=
\sum_{\gamma} f_{\alpha\beta}^{\ \ \gamma}\,B_\gamma,
\qquad
f_{\alpha\beta}^{\ \ \gamma}
=
\frac{\langle B_\gamma,\, i[B_\alpha,B_\beta]\rangle}
     {\langle B_\gamma,B_\gamma\rangle},
\label{eq:structureconstants}
\end{equation}
where $\langle A,B\rangle=\tr(A^\dagger B)$ denotes the Hilbert--Schmidt inner product. The coefficients $f_{\alpha\beta}^{\ \ \gamma}$, sometimes also termed \emph{structure factors}, are defined entirely by the chosen operator basis $\{B_\alpha\}$ and represent the commutation relations between its elements.

For a Hermitian generator \(H\in i\g\), expanded as
\begin{equation}
H=\sum_{\mu} h_\mu\,B_\mu,
\end{equation}
its infinitesimal adjoint action is the linear map \(\Phi^{\ad}(H)\in\mathbb{R}^{d\times d}\) defined by
\begin{equation}
i[H,B_\alpha]
=
\sum_{\beta}\big(\Phi^{\ad}(H)\big)_{\alpha\beta}\,B_\beta.
\label{eq:adjoint_generator_definition}
\end{equation}
Using \eqref{eq:structureconstants}, its matrix elements are therefore
\begin{equation}
\big(\Phi^{\ad}(H)\big)_{\alpha\beta}
=
\sum_{\mu} h_\mu\, f_{\mu\alpha}^{\ \ \beta}.
\label{eq:adjoint_generator_matrix}
\end{equation}

For a single parameterized gate \(U(\theta)=e^{-i\theta H}\), the finite adjoint action is obtained by exponentiating the infinitesimal generator:
\begin{equation}
\Phi^{\Ad}(e^{-i\theta H})
=
e^{\,\theta\,\Phi^{\ad}(H)}.
\label{eq:adjointidentity}
\end{equation}
Thus, once the matrices $\Phi^{\ad}(H_k)$ for the circuit generators are known, expectation values can be evaluated by adjoint-space propagation.

Note that the adjoint representation is natural for the present formulation of $\gsim$ because expectation values are propagated in an invariant operator subspace, and the required coupling matrices are precisely the structure constants of $\g$ determined by commutators~\eqref{eq:structureconstants}. It is not, however, a fundamental requirement of Lie-algebraic simulation. In the original formulation of Ref.~\cite{Somma.2006-EfficientSolvabilityHamiltonians}, one may work in any faithful polynomial-dimensional representation of the abstract Lie algebra.

For a circuit of the form \eqref{eq:standardansatz}, combining \eqref{eq:eout_adjoint} and \eqref{eq:adjointidentity} yields
\begin{equation}
\mathbf{e}^{(\mathrm{out})}
=
\left(
\prod_{l=1}^{L}\prod_{k=1}^{K}
e^{\theta_{l,k}\,\Phi^{\ad}(H_k)}
\right)
\mathbf{e}^{(\mathrm{in})},
\label{eq:eout_product}
\end{equation}
and the desired expectation value follows from the dot product~\eqref{eq:expval_dotproduct}. This replaces manipulations of $2^n\times 2^n$ matrices by linear algebra operations on dimension-$d$ objects. Whenever $d{=}\dim(\g)$ grows only polynomially in $n$, this yields the basic source of exponential savings exploited by $\gsim$. 

Moreover, since $\langle O(\boldsymbol{\theta})\rangle$ is obtained from explicit products of matrix exponentials, exact gradients can be computed by differentiating \eqref{eq:eout_product}. Writing the parameters as a single index $m\in\{1,\dots,M\}$ with $M=LK$, one may express
\begin{equation}
\frac{\partial}{\partial \theta_m}\langle O(\boldsymbol{\theta})\rangle
=\boldsymbol{w}^T\left(\frac{\partial \Phi^{\Ad}(U(\boldsymbol{\theta}))}{\partial \theta_m}\right) \mathbf{e}^{(\mathrm{in})},
\label{eq:gradient_adjoint}
\end{equation}
and evaluate the derivative efficiently. Using forward-mode differentiation, a single derivative can be evaluated with the same order of cost as a forward pass; using reverse-mode differentiation through the adjoint-space propagation, all $M=LK$ parameter derivatives of a scalar expectation value can be obtained with the same asymptotic scaling as the forward evaluation, up to constant factors and memory overhead (see Ref.~\cite{Goh.2025-LiealgebraicClassicalSimulations}).

Finally, note that the same adjoint-space propagation extends naturally to correlators, i.e., expectation values of operator products. For example, if
\begin{equation}
    O^{(1)}=\sum_\alpha w^{(1)}_\alpha B_\alpha,
\quad
O^{(2)}=\sum_\beta w^{(2)}_\beta B_\beta,
\end{equation}
with both factors supported in $i\g$, one introduces the second-order classical representation
\begin{equation}
\big(E^{(\mathrm{in/out})}\big)_{\alpha\beta}
:=
\tr\!\big[B_\alpha B_\beta\,\rho^{(\mathrm{in/out})}\big].
\end{equation}
Under the same circuit $U(\boldsymbol\theta)$, this object evolves bilinearly under the adjoint action,
\begin{equation}
E^{(\mathrm{out})}
=
\Phi^{\Ad}\!\big(U(\boldsymbol\theta)\big)\,
E^{(\mathrm{in})}\,
\Phi^{\Ad}\!\big(U(\boldsymbol\theta)\big)^T,
\end{equation}
so that
\begin{equation}
\langle O^{(1)}O^{(2)}(\boldsymbol\theta)\rangle
=
\big(w^{(1)}\big)^T E^{(\mathrm{out})} w^{(2)}.
\label{eq:correlators}
\end{equation}
Higher-order products can be treated analogously in the generic $\gsim$ formulation by propagating higher-order tensors, although the cost grows quickly with correlator order. In special Lie-algebraic settings, additional structure can reduce this cost: for highest-weight or generalized coherent input states, Wick-type Cartan--Weyl normal-ordering techniques can evaluate fixed-order correlators of algebra elements without explicitly storing the full higher-order tensor~\cite{Somma.2006-EfficientSolvabilityHamiltonians}. More generally, one may also work in polynomial-dimensional invariant operator subspaces beyond $i\g$. See Refs.~\cite{Goh.2025-LiealgebraicClassicalSimulations,Somma.2006-EfficientSolvabilityHamiltonians} for details.

\cref{fig:g-sim-outline} summarizes the main algorithmic components of $\gsim$. \cref{tab:gsimcomplexities} collects schematic representation-agnostic complexity bounds expressed in terms of $t_{\mathrm{com}}$ and $t_{\mathrm{li}}$; details are given in Appendix~\ref{appsec:gsimcomplexity}. For conciseness, we write $d=\dim(\g)$ and separate \emph{preprocessing}, which constructs a basis and the adjoint data required for propagation, from \emph{main simulation}, which propagates $\mathbf{e}^{(\mathrm{in})}$ and evaluates exact expectation values and gradients for given circuit parameters. 
Note that ``basis construction'' refers to the algorithmic task of computing the Lie closure of a finite generator set and extracting a basis of the resulting DLA. If a suitable basis of the relevant algebra or invariant operator subspace is already known analytically, then this part of the preprocessing can be omitted, and one only needs to express the generators, input data, and observables in that basis and compute the corresponding representation elements.

\begin{figure}[htbp]
\includegraphics[width=1\linewidth]{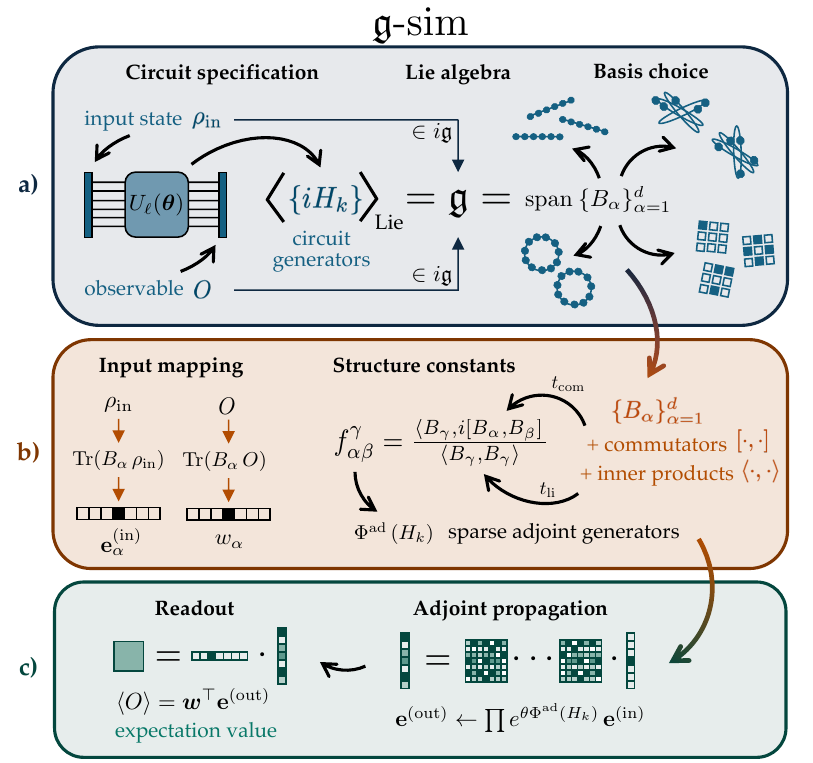}
\caption{\label{fig:g-sim-outline}Schematic overview of the $\gsim$ pipeline. a) From a circuit specification $U(\boldsymbol{\theta})$ with generators $H_k$, together with an initial state $\rho_{\mathrm{in}}$ and observable $O$, one first identifies the relevant dynamical Lie algebra $\mathfrak{g}=\langle iH_k\rangle_{\mathrm{Lie}}$ and chooses a convenient basis $\{B_\alpha\}_{\alpha=1}^d$ adapted to the structure at hand (e.g., Pauli strings, cycles, orbits, or MGGM matrices). b) Preprocessing maps $\rho_{\mathrm{in}}$ and $O$ to coordinate vectors $\mathbf{e}^{(\mathrm{in})}$ and $\mathbf{w}$ of size $d$ and computes the adjoint generators $\Phi^\mathrm{ad}(H_k)$ from structure constants $f_{\alpha\beta}^{\gamma}$, using commutator and inner-product primitives with costs governed by $t_{\mathrm{com}}$ and $t_{\mathrm{li}}$. c) Simulation then propagates $\mathbf{e}^{(\mathrm{in})}$ in the $d$-dimensional adjoint space via a product of matrix exponentials and returns the expectation value by the inner-product readout $\langle O\rangle=\boldsymbol{w}^\top\mathbf{e}^{(\mathrm{out})}$. The key reduction is from Hilbert-space dimension $2^n$ to Lie algebra dimension $d{:=}\dim(\mathfrak{g})$, so efficient simulation hinges on $d=\mathrm{poly}(n)$ together with a basis in which preprocessing primitives are tractable.}
\end{figure}

\begin{table}[htbp]
\centering
\makegapedcells
\begin{tabular}{c|c} Primitive & Complexity \\ \hline
\multicolumn{2}{c}{\textit{preprocessing}}\\
Lie closure / basis construction~\eqref{eq:DLAdefinition} & $\mathcal{O}\!\left(d^2\, (t_{\mathrm{com}}+ t_{\mathrm{li}})\right)$ \\  
structure constants / adjoint data~\eqref{eq:structureconstants} & $\mathcal{O}\!\left(d^2\, t_{\mathrm{com}}(n)\right)$ \\ \hline
\multicolumn{2}{c}{\textit{main simulation}}\\
matrix exponentials~\eqref{eq:adjointidentity} & $\mathcal{O}\!\left(K\, d^2\right)$ \\
operator~\eqref{eq:eout_product} / gradient~\eqref{eq:gradient_adjoint} evaluation & $\mathcal{O}\!\left(LK\, d^2\right)$ \\ 
\end{tabular}
\caption{\label{tab:gsimcomplexities}
Schematic complexity bounds for $\gsim$ in terms of $d=\dim(\g)$ and $n$. Here $t_{\mathrm{com}}(n)$ is the cost of evaluating a commutator primitive in the chosen basis representation, and $t_{\mathrm{li}}(d)$ is the cost of linear-independence checks during basis growth. The matrix-exponential cost assumes precomputed eigendecompositions of the (fixed) adjoint generators, so that applying $\exp(\theta\,\Phi^{\ad}(H_k))$ to a vector can be done in $\mathcal{O}(d^2)$ time~\cite{Goh.2025-LiealgebraicClassicalSimulations}. The final expectation value~\eqref{eq:expval_dotproduct} additionally requires one $\mathcal{O}(d)$ dot product with $\boldsymbol{w}$, which is omitted as it is typically negligible compared to the rest scaling as $\mathcal{O}(d^2)$.}
\end{table}

While the original proposal of $\gsim$~\cite{Goh.2025-LiealgebraicClassicalSimulations} emphasizes that operator evaluation in the adjoint space is efficient once a polynomial-dimensional algebra and its adjoint data are available, the practical bottleneck is often preprocessing: constructing a suitable basis representation and computing the required structure constants at scale. A main technical contribution of this paper is to develop symmetry-adapted or subspace-adapted representations and primitives that make this preprocessing efficient for broad families of polynomial DLAs, with explicit emphasis on optimizing $t_{\mathrm{com}}$ and $t_{\mathrm{li}}$ in settings where naive Pauli expansions are in general disadvantageous (e.g., exponential).

\subsection{Notions of classical simulability}
\label{ssec:notionSimulability}
Throughout this paper, we discuss classical simulability of quantum circuits, with a particular emphasis on Lie-algebraic simulation as outlined above. Since several notions of ``simulability'' coexist in the literature, we briefly fix terminology here and explain where our setting fits. 

Consider a uniform family of $n$-qubit quantum circuits $\mathcal{U}_n$. For a given $U\in \mathcal{U}_n$ acting on an input state $\ket{\psi_0}$, followed by a computational-basis measurement, the outcome distribution is
\begin{equation}
p(x)=|\langle x|U|\psi_0\rangle|^2,\quad x\in\{0,1\}^n.
\end{equation}
A standard distinction is between \emph{strong} and \emph{weak} classical simulation~\cite{Nest.2009-Classicalsimulationquantumcomputation}. We say that $\mathcal{U}_n$ is \emph{strongly simulable} if one can compute probabilities of specified measurement events (e.g., $p(x)$ for a given $x$, or more generally marginal probabilities) in time polynomial in $n$ (and in the requested output precision). It is \emph{weakly simulable} if one can sample $x$ from the distribution $p(x)$ in time polynomial in $n$. Strong simulation is at least as powerful as weak simulation: given efficient access to suitable marginal probabilities, one can sample sequentially from $p(x)$ using a polynomial number of such probability computations, and more recent reductions show alternative routes from probability/amplitude evaluation to sampling that can improve constants depending on the available primitive~\cite{Bravyi.2022-HowSimulateQuantum}. By contrast, physical quantum devices natively provide samples from $p(x)$ and thus expectation values must be estimated from repeated measurements with statistical error.

The notion of Lie-algebraic simulability~\cite{Somma.2006-EfficientSolvabilityHamiltonians} studied in this work accomplishes a complementary computational task: computing expectation values of the form~\eqref{eq:expectationvalue} for specified observables $O$ and gate-evolved input states $\rho_\mathrm{in}$, provided the relevant observables are supported in polynomial-dimensional Lie-algebraic invariant subspaces and the required overlaps with the input state can be computed efficiently. Thus, $\gsim$ is not, in general, a strong simulator in the standard computational-basis sense, nor does it necessarily provide weak simulation of measurement outcomes. Its output is instead exact access, up to classical arithmetic precision, to a restricted but important class of quantum mean values.

This setting is especially relevant for variational quantum algorithms, QML models, and many physical simulation tasks, where the central computational objects are expectation values, gradients, and correlators rather than samples from the full computational-basis distribution. This motivates renewed interest in classical methods that can evaluate such quantities for structured circuit families~\cite{Bravyi.2021-ClassicalAlgorithmsQuantum,Goh.2025-LiealgebraicClassicalSimulations,Rudolph.2025-PauliPropagationComputational}. 

Finally, we use the terms ``efficient'' and ``polynomial-time'' interchangeably throughout. This should not be conflated with \emph{practical} performance, which depends on constant factors and typical-instance behavior; assessing these aspects is an additional focus of this work.

\section{Systems with polynomial-dimensional dynamical Lie algebras}
\label{sec:polyDLAsystems}

The framework of Lie-algebraic simulation~\cite{Somma.2006-EfficientSolvabilityHamiltonians,Goh.2025-LiealgebraicClassicalSimulations} (reviewed in \cref{ssec:gsimTheory}) rests on a simple structural observation: for many structured circuit families, the Heisenberg-picture evolution of relevant observables is closed with respect to an invariant operator subspace whose dimension grows only polynomially with the system size $n$. When this happens, expectation values of the form~\eqref{eq:expectationvalue} can be computed by propagating in a low-dimensional representation of the induced operator dynamics, rather than in the $2^n$-dimensional Hilbert space.
This shifts the central question of simulability from the size of the Hilbert space to two algebraic questions: when does a circuit family induce a DLA $\g\subseteq\mathfrak{su}(2^n)$ with $\dim(\g)=\mathcal{O}(\poly(n))$,
and in which representation can the basic preprocessing primitives required by $\gsim$ be implemented efficiently?

A natural baseline for the first question is the regime of \emph{Pauli Lie algebras}, where the generators are individual Pauli strings. In this setting, recent classification results imply a strong dichotomy: apart from essentially free-fermionic families with $\Theta(n^2)$-dimensional algebras and degenerate low-dimensional cases such as commuting or effectively one-body structures, connected Pauli Lie algebras grow exponentially in $n$~\cite{Aguilar.2024-FullClassificationPauli,Wiersema.2024-ClassificationDynamicalLie,Kokcu.2024-ClassificationDynamicalLie}. Within the Pauli-generator setting, the free-fermion/matchgate case therefore appears as the main nontrivial polynomial regime, and indeed in that case efficient simulation is already well understood through fermionic linear optics~\cite{Knill.2001-FermionicLinearOptics,Valiant.2001-QuantumComputersThat,Terhal.2002-ClassicalSimulationNoninteractingfermion}. In this sense, $\gsim$ recovers free-fermionic simulability as one important instance of a broader Lie-algebraic simulability paradigm.

However, Pauli string generators do not provide a generic description of parametrized quantum circuits. Many ans\"atze used in practice are instead generated by sums of Pauli strings, including the quantum approximate optimization algorithm (QAOA)~\cite{Farhi.2014-QuantumApproximateOptimization,Blekos.2023-ReviewQuantumApproximate}, Hamiltonian variational ans\"atze~\cite{Wecker.2015-ProgressPracticalQuantum,Wiersema.2020-ExploringEntanglementOptimization}, symmetry-constrained models~\cite{Kerenidis.2022-QuantumMachineLearning,Cherrat.2023-QuantumDeepHedging}, and collective-spin constructions~\cite{Skolik.2023-EquivariantQuantumCircuits,Meyer.2023-ExploitingSymmetryinQML,Schatzki.2024-TheoreticalGuaranteesPermutationequivariant}.
Abstractly, such generators take the form
\begin{equation}
H=\sum_{k=1}^{N} c_k\,P_k^{(n)}, \qquad c_k\in\mathbb{R},
\label{eq:sumofpaulistrings}
\end{equation}
with $P_k^{(n)}\in\mathcal{P}_n$. Once one moves to this broader sum-generator regime, polynomially-growing DLAs can arise from additional structure such as symmetry, constrained connectivity, or restriction to invariant subspaces. Thus, the sharp Pauli algebra dichotomy should not be mistaken for a general classification of polynomial-dimensional Lie-algebraic dynamics. Doing so would incorrectly suggest that practical Lie-algebraic simulability is essentially exhausted by free fermions.

This brings us to the second question, which is the main algorithmic bottleneck of the present work. Even when $\dim(\g)=\mathcal{O}(\poly(n))$, efficient simulation does not follow automatically from a naive Pauli representation: basis elements of $\g$ may have exponentially-large Pauli expansions, and naive commutator or inner-product computations in that basis may become infeasible. This has motivated the conjecture that efficient $\gsim$ might require not only polynomial DLA dimension, but also a basis with ``slow Pauli expansion''~\cite{Heredge.2025-CharacterizingPrivacyQuantum}, meaning that each basis element has only polynomially many nonzero coefficients in the standard Pauli basis. One central point of this work is that this extra restriction is not necessary in general. What matters is not whether the algebra is sparse in the standard Pauli basis, but whether one can find a representation adapted to its symmetry or subspace structure in which the preprocessing primitives of $\gsim$ become efficient.

The remainder of this paper develops such representations for three representative sources of polynomial-dimensional DLAs, illustrating progressively richer symmetry structures and algebraic scalings (cf.~\cref{fig:summary}):

\begin{itemize}[leftmargin=*]
\item \emph{Matchgate structure.}
Free-fermionic/matchgate algebras provide the paradigmatic quadratic-dimensional setting, with $\dim(\mathfrak{g})=\mathcal{O}(n^2)$. We revisit this regime from the perspective of $\gsim$ preprocessing and introduce symmetry-adapted ``Pauli cycle'' representations for translation-invariant variants, reducing commutator costs by exploiting cyclic structure (\cref{sec:FFalgebras}).

\item \emph{Permutation symmetry.}
For permutation-equivariant ($S_n$-equivariant) dynamics, collective generators can produce polynomial-dimensional operator algebras even though their Pauli expansions are exponentially large. We provide the symmetry-adapted ``Pauli orbit'' basis, which gives an exponentially compressed representation of the cubic-dimensional ambient algebra and enables efficient $\gsim$ preprocessing (\cref{sec:Pauliorbitalgebra}).

\item \emph{$U(1)$ symmetry and fixed excitation sectors.}
For Hamming-weight preserving ($U(1)$-equivariant) dynamics, polynomial-dimensional effective algebras arise after restricting the dynamics to fixed-weight sectors $\mathcal{H}_k$ with bounded $k$. On such sectors, $d_k=\binom{n}{k}=\mathcal{O}(n^k)$ and the effective algebra is contained in $\mathfrak{u}(d_k)$, whose dimension scales as $\mathcal{O}(n^{2k})$. We provide an explicit modified generalized Gell--Mann basis with closed-form commutation rules for this subspace-adapted representation (\cref{sec:restrictedsubspace-algebras}).
\end{itemize}

Taken together, these families support the main message of this paper: Lie-algebraic simulability identifies a broad class of potentially tractable quantum dynamics, but it is fundamentally representation-sensitive. Polynomial DLA dimension is therefore a structural promise, not by itself an implementation: it becomes algorithmically feasible only once the algebra is represented in a basis where commutators, inner products, and adjoint-space propagation are efficient.

\section{Free-fermion algebras}
\label{sec:FFalgebras}
Free-fermion (or Gaussian) dynamics are generated by Hamiltonians that are at most quadratic in fermionic creation and annihilation operators. Concretely, for $n$ fermionic modes with operators $\{a_j,a_j^\dagger\}_{j=1}^n$ satisfying the canonical anticommutation relations, consider the Majorana operators 
\begin{equation}
\gamma_{2j-1}:=a_j+a_j^\dagger,
\qquad
\gamma_{2j}:=-i(a_j-a_j^\dagger),
\label{eq:Majoranas}
\end{equation}
which satisfy $\{\gamma_\mu,\gamma_\nu\}=2\delta_{\mu\nu}I$. Any free-fermion Hamiltonian can be written as
\begin{equation}
H_{\mathrm{FF}}=\frac{i}{4}\sum_{\mu,\nu=1}^{2n}\alpha_{\mu\nu}\,\gamma_\mu \gamma_\nu,
\qquad
\alpha=-\alpha^{\mathsf{T}}\in\mathbb{R}^{2n\times 2n},
\label{eq:FF_majorana_quadratic}
\end{equation}
so the dynamical generators are precisely the quadratic bilinears $\gamma_\mu \gamma_\nu$. The Lie algebra generated by these bilinears is isomorphic to the special orthogonal Lie algebra
\begin{equation}
\mathfrak{g}^{\mathrm{FF}}
\cong
\mathfrak{so}(2n),
\qquad
\dim(\mathfrak{g}^{\mathrm{FF}})=n(2n-1),
\label{eq:FF_so2n_dim}
\end{equation}
and in certain ``extended'' variants (allowing additional linear/odd terms, or equivalently adding an auxiliary mode) one obtains $\mathfrak{so}(2n+1)$ with $\dim= n(2n+1)=2n^2+n$; see, e.g., Refs.~\cite{Knill.2001-FermionicLinearOptics,Aguilar.2024-FullClassificationPauli}. 

Matchgates arise as the qubit realization of such free-fermion evolutions under the Jordan--Wigner transformation. In particular, on nearest-neighbor qubits they are exactly the parity-preserving two-qubit unitaries, i.e.\ gates that act block-diagonally on the even- and odd-parity subspaces ($\mathrm{span}\{|00\rangle,|11\rangle\}$ and $\mathrm{span}\{|01\rangle,|10\rangle\}$, respectively). Equivalently (up to a global phase), they admit the standard block form~\cite{Valiant.2001-QuantumComputersThat,Terhal.2002-ClassicalSimulationNoninteractingfermion}
\begin{equation}
U_{\mathrm{MG}}
=
\begin{pmatrix}
U^{(1)}_{11} & 0 & 0 & U^{(1)}_{12}\\
0 & U^{(2)}_{11} & U^{(2)}_{12} & 0\\
0 & U^{(2)}_{21} & U^{(2)}_{22} & 0\\
U^{(1)}_{21} & 0 & 0 & U^{(1)}_{22}
\end{pmatrix},
\quad
U^{(1)},U^{(2)}\in \mathrm{SU}(2).
\label{eq:matchgates}
\end{equation}
when written in the ordered basis $\{|00\rangle,|01\rangle,|10\rangle,|11\rangle\}$. These form precisely the class of non-interacting fermion circuits whose classical simulability was established in the early matchgate literature \cite{Knill.2001-FermionicLinearOptics,Valiant.2001-QuantumComputersThat,Terhal.2002-ClassicalSimulationNoninteractingfermion}. 

Notice that under the Jordan--Wigner transformation, Majoranas map to Pauli strings
\begin{equation}
\gamma_{2j-1}=\bigotimes_{k=1}^{j-1}Z_k\ X_j,
\qquad
\gamma_{2j}=\bigotimes_{k=1}^{j-1}Z_k\ Y_j,
\label{eq:JW_majoranas}
\end{equation}
so the quadratic bilinears $\gamma_\mu \gamma_\nu$ become Pauli strings with at most two non-identity letters and (possibly) a $Z$-string in between~\cite{Knill.2001-FermionicLinearOptics}. For nearest neighbors, this reduces to a local Pauli generator set on qubits $i, i+1$, 
\begin{equation}
G_\mathrm{FF}=\{Z_i,\,Z_{i{+}1},\,X_iX_{i+1},\,Y_iY_{i+1},\,X_iY_{i+1},\,Y_iX_{i+1}\},
\label{eq:matchgate_Pauli_generators}
\end{equation}
which generates exactly the matchgate Lie algebra on that edge. Consequently, in the Pauli-basis DLA language used throughout this paper, free-fermion/matchgate circuits correspond to quadratic-dimensional Pauli Lie algebras whose Lie closure is isomorphic to $\mathfrak{so}(2n)$ (or $\mathfrak{so}(2n+1)$), hence ${\dim(\g^\mathrm{FF})=\Theta(n^2)}$, consistent with the quadratic families classified in Refs.~\cite{Wiersema.2024-ClassificationDynamicalLie,Kokcu.2024-ClassificationDynamicalLie,Aguilar.2024-FullClassificationPauli}.

\subsection{$\gsim$ for free-fermionic Pauli Lie algebras}
\label{ssec:FF_Paulialgebras}
Free-fermionic (matchgate) circuits already admit well-established, optimized classical simulation methods in the Pauli basis such as fermionic linear optics (FLO)~\cite{Knill.2001-FermionicLinearOptics,Valiant.2001-QuantumComputersThat,Terhal.2002-ClassicalSimulationNoninteractingfermion}, which is why they are not the primary target of this paper.
Nevertheless, their Lie-algebraic description forms an instructive baseline for the notion of classical simulability via $\gsim$: in this regime, $\gsim$ can be implemented particularly efficiently because the DLA admits a Pauli string basis with fast commutator and inner-product primitives.

A natural basis of a free-fermion Pauli Lie algebra $\g\subseteq\g^\mathrm{FF}$ is a subset of the traceless Pauli strings $\mathcal{P}_n$~\eqref{eq:Paulistringbasis}. Efficient preprocessing and simulation then hinge on representing Pauli strings in a way that supports fast linear independence checks and commutator evaluation. A standard choice is the binary symplectic representation.

\begin{definition}[Binary symplectic representation of Pauli strings]
Let $P^{(n)}=\bigotimes_{i=1}^n P_i\in\mathcal{P}_n$ with $P_i\in\{I,X,Y,Z\}$. Its binary symplectic representation is the bit string $(\boldsymbol{x},\boldsymbol{y})\in\{0,1\}^{2n}$ defined by
\begin{equation}
x_i=\begin{cases}
        1 & \text{if}\ P_i\in\{X,Z\} \\
        0 & \text{else}
    \end{cases}, \quad y_i=\begin{cases}
        1 & \text{if}\ P_i\in\{Y,Z\} \\
        0 & \text{else}
    \end{cases}
\label{eq:PSbinsymprep}
\end{equation}
so that $P_i$ is determined by $(x_i,y_i)$ via
$I\leftrightarrow(0,0)$, $X\leftrightarrow(1,0)$, $Y\leftrightarrow(0,1)$, $Z\leftrightarrow(1,1)$.
This stores $P^{(n)}$ using $2n=\mathcal{O}(n)$ bits (up to a global phase convention).
\label{def:symprepPaulistrings}
\end{definition}

Pauli strings are orthogonal under the Hilbert-Schmidt inner product, hence inner products reduce to equality checks. 

\begin{corollary}[Linear-time Hilbert--Schmidt inner product for Pauli strings]
Let $P,P'\in\mathcal{P}_n$. Then
\begin{equation}
   \langle P,P'\rangle =\tr[P^\dagger P']=\begin{cases}
       0 & \text{if}\ P\neq P'\\
       2^n $ \text{else}$
   \end{cases}.
\end{equation}
Thus, normalized inner products reduce to comparing the $2n$ binary representations $(\boldsymbol{x},\boldsymbol{y})$ and $(\boldsymbol{x}',\boldsymbol{y}')$ from \cref{def:symprepPaulistrings}, which can be done in $\mathcal O(n)$ bit operations.
\label{cor:PSlineartimeLIchecks}
\end{corollary}

\begin{corollary}[Linear-time commutators of Pauli strings]
Let $P,P'\in\mathcal{P}_n$ be represented by $(\boldsymbol{x},\boldsymbol{y})$ and $(\boldsymbol{x'},\boldsymbol{y'})$ as in \cref{def:symprepPaulistrings}. Then $[P,P']$ is either zero or proportional to a single Pauli string, and can be computed in $\mathcal{O}(n)$ time using bitwise operations (XOR) on the length-$2n$ bit strings.
\label{cor:PSlineartimecommutators}
\end{corollary}

In practice, bit-packed implementations and hash-based data structures can substantially reduce constant factors. For example, if Pauli strings are stored with precomputed hashes, membership queries can have expected constant-time lookup after the representation has been constructed. Nevertheless, computing or comparing the underlying length-$2n$ representation has worst-case cost $\mathcal O(n)$, and throughout the paper we use this conservative bound.

Combining these primitive costs with the complexity bounds in \cref{tab:gsimcomplexities} yields a preprocessing cost of $\mathcal{O}(d^2\,n)$. For free-fermion algebras, one has $d{=}\dim(\g){=}\mathcal{O}(n^2)$, hence $\mathcal{O}(n^5)$ for preprocessing in the worst case. Empirically, our optimized implementation exhibits smaller exponents (below $n^4$ in the numerical benchmarks in Appendix~\ref{appssec:BMforFFalgebrasTI}), consistent with additional sparsity structure in typical instances.

Beyond preprocessing, Pauli-basis algebras admit further speedups during evaluation. In particular, for a Pauli basis the adjoint action of a generator on the basis is sparse: for fixed $\alpha$ and $\beta$, the commutator $i[B_\alpha,B_\beta]$ is either $0$ or a single basis element up to sign, so each adjoint generator $\Phi^{\ad}(B_\alpha)$ has at most one nonzero entry per row/column (after a suitable ordering). This induces a block structure consisting of $2\times 2$ rotations and identities, allowing \cref{eq:adjointidentity} to be applied efficiently via precomputed block decompositions~\cite{Goh.2025-LiealgebraicClassicalSimulations}.

The preceding discussion treats the case where the DLA basis itself consists of single Pauli strings. In this \emph{Pauli Lie algebra} setting, classification results imply a strong dichotomy: apart from essentially free-fermionic families, connected Pauli Lie algebras are typically exponentially large in $n$~\cite{Aguilar.2024-FullClassificationPauli,Wiersema.2024-ClassificationDynamicalLie,Kokcu.2024-ClassificationDynamicalLie}. In many practically relevant ansätze, however, the dynamical generators are not single Pauli strings but general operators, i.e., structured linear combinations of Pauli strings. Such ``sum generators'' can still induce quadratic-dimensional DLAs which correspond to free-fermion equivalent. The difference is that this free-fermion structure is no longer manifest in the Pauli basis and is often revealed only after exploiting additional algebraic or symmetry structure. This motivates the next subsection, where we demonstrate the power of $\gsim$ to tackle free-fermion equivalents generated by general (non-Pauli) operators and show how symmetry-adapted representations can be essential for efficient preprocessing in this more general setting.

\subsection{$\gsim$ for free-fermion algebras generated by arbitrary operators}
The sum-generator regime contains many models that are free-fermion equivalents, even when this structure is not manifest at the level of the chosen spin representation. Although their generators are not given in an explicit quadratic-fermion normal form such as \cref{eq:matchgate_Pauli_generators}, they admit a possibly non-local mapping to non-interacting fermions based on structural criteria such as frustration-graph invariants and associated conserved quantities~\cite{Chapman.2020-CharacterizationSolvableSpin,Chapman.2023-UnifiedGraphtheoreticFramework,Fukai.2025-QuantumCircuitsFreeFermion}. This includes generalized Jordan--Wigner solvable models such as the Kitaev honeycomb model, as well as more non-local ``free fermions in disguise'' related to Fendley's four-fermion construction. Ref.~\cite{Chapman.2023-UnifiedGraphtheoreticFramework}, in particular, gives a unified graph-theoretic solvability criterion formulated in terms of the Hamiltonian frustration graph, thereby capturing broad families beyond the textbook Jordan--Wigner setting.

From the perspective of $\gsim$, this motivates a complementary workflow. Rather than first constructing an explicit fermionization map and switching to a specialized fermionic simulation method, one may probe the DLA generated by a given operator set directly using the same Lie-closure and structure constant pipeline required by $\gsim$. If the resulting algebra exhibits free-fermion-type closure, for instance $\mathfrak{so}(2n)$- or $\mathfrak{so}(2n{+}1)$-like scaling and commutation structure, expectation values can then be simulated in the corresponding adjoint representation using the same framework as for manifestly free-fermionic Pauli algebras. This keeps the simulation pipeline uniform across manifestly free-fermionic systems, disguised free-fermion systems, and the genuinely non-free-fermionic polynomial families discussed in \cref{sec:Pauliorbitalgebra,sec:restrictedsubspace-algebras}.

A conceptual caveat is that quadratic scaling alone does not characterize free fermions. While any subalgebra of a free-fermion algebra has $\dim(\g){=}\mathcal{O}(n^2)$, the converse implication fails in general: there exist physically meaningful circuits with $\dim(\g){=}\Theta(n^2)$ whose induced algebras are not isomorphic to $\mathfrak{so}(2n)$ or $\mathfrak{so}(2n{+}1)$; see \cref{rem:Counterexample-quadratic=FF} for an explicit $\mathfrak{su}(n)$ counterexample arising from universal subspace dynamics on the single-excitation sector.

We now illustrate this viewpoint in a concrete and practically common setting: translation-invariant free-fermion generators. This example is simple enough to expose the relevant $\gsim$ primitives explicitly, while already showing why one sometimes benefits from a symmetry-adapted representation rather than a naive expansion into individual Pauli strings.

\subsubsection{Translation-invariant systems}
\label{sssec:TI-systems}
A general nearest-neighbor Pauli Hamiltonian on a chain that admits a Jordan--Wigner free-fermion description has the support pattern~\cite{Chapman.2020-CharacterizationSolvableSpin}
\begin{equation}
H=
\sum_{i=1}^{n-1}\ \sum_{P,P'\in\{X,Y\}} \alpha^{PP'}_i\, P_i P_{i+1}
\;+\;
\sum_{i=1}^{n}\beta_i Z_i,
\label{eq:chapman_1D_general}
\end{equation}
i.e.\ all two-body couplings are restricted to the $\{X,Y\}$-plane, supplemented by on-site $Z$-fields. A convenient generating set that spans this support pattern is
\begin{equation}
\begin{aligned}
G_{\mathrm{TI}}
=
\Bigl\{
&\sum_{i=1}^{n-1} X_iX_{i+1},\ 
\sum_{i=1}^{n-1} X_iY_{i+1},\ 
\sum_{i=1}^{n-1} Y_iX_{i+1},\ \\
&\sum_{i=1}^{n-1} Y_iY_{i+1},\ 
\sum_{i=1}^{n} Z_i
\Bigr\}.
\end{aligned}
\label{eq:TI_generator_set}
\end{equation}
Each generator in \eqref{eq:TI_generator_set} is quadratic under the Jordan--Wigner transformation and therefore induces dynamics contained in the free-fermion algebra $\mathfrak{g}^{\mathrm{FF}}$; in particular,
\begin{equation}
\mathfrak{g}=\langle iG_{\mathrm{TI}}\rangle_{\mathrm{Lie}}
\subseteq \mathfrak{g}^{\mathrm{FF}}
\cong \mathfrak{so}(2n),
\end{equation}
so $\dim(\g)=\mathcal{O}(n^2)$. The same construction applies to periodic boundary conditions by adding the wrap-around terms $P_nP'_1$ to the two-body sums~\eqref{eq:chapman_1D_general}.

Beyond free-fermion structure, the periodic (translation-invariant) setting also introduces an independent symmetry constraint: commutation with the cyclic shift. Let $C_n=\langle\tau\rangle$ be the cyclic translation group generated by the one-step shift $\tau$, and let $U_\tau$ denote its unitary representation on $(\mathbb{C}^2)^{\otimes n}$. Define the translation-invariant operator algebra
\begin{equation}
    \g^\mathrm{TI}:=\{A\in\mathfrak{u}(2^n)\ : \ [A,U_\tau]=0\},
\label{eq:gTIdefinition}
\end{equation}
with $\dim(\g^\mathrm{TI})=\Theta(4^n/n)$~\cite{Mansky.2025-ScalingSymmetryrestrictedQuantum}. Equivalently, $\mathfrak{g}^{\mathrm{TI}}$ is the fixed-point subspace of the conjugation action of $C_n$ on operators. The corresponding Reynolds (twirling) projector is
\begin{equation}
\mathcal{T}(A)
=
\frac{1}{n}\sum_{i=0}^{n-1} U_\tau^{i}\,A\,(U_\tau^{\dagger})^{i},
\label{eq:cyclic_twirl}
\end{equation}
which satisfies $\mathcal{T}(A)\in \g^{\mathrm{TI}}$ for all $A$ and $\mathcal{T}(A)=A$ for all $A\in \g^{\mathrm{TI}}$.  
This naturally motivates a symmetry-adapted basis for subsystems of $\g^\mathrm{TI}$ (or $\g^\mathrm{TI}\cap\mathfrak{su}(2^n)$, excluding the identity): the twirling map $\mathcal{T}$ is the orthogonal projector onto the $C_n$-invariant subspace, and applying $\mathcal{T}$ to Pauli strings produces the corresponding cyclic sums. We refer to these symmetrized Pauli string sums as \emph{Pauli cycles}.

\begin{definition}[Pauli cycle]
Let $P\in\mathcal{P}_n$~\eqref{eq:Paulistringbasis} be a (non-identity) Pauli string on an $n$-qubit system. Its associated Pauli cycle is the $C_n$-invariant operator
\begin{equation}
O_P:=i\mathcal{T}(P)=\frac{i}{n}\sum_{i=0}^{n-1} U_\tau^{i}\,P\,(U_\tau^{\dagger})^{i}.
\label{eq:Paulicycle}
\end{equation}
We write $P{\sim}P'$ if $P'{=}U_\tau^{i}P(U_\tau^\dagger)^{i}$ for some ${i\in\{0,\dots,n-1\}}$, i.e.\ if $P$ and $P'$ lie in the same cycle. Choosing one representative for each equivalence class yields a set of Pauli cycles $\mathcal{C}_n$ that spans the subspace $\mathfrak{g}^{\mathrm{TI}}\cap\mathfrak{su}(2^n)$ (or the full $\mathfrak{g}^{\mathrm{TI}}$ as soon as $iI$ is included).
\label{def:Paulicycle}
\end{definition}
In particular, after fixing a canonical representative for each cyclic orbit, testing whether two Pauli cycles coincide reduces to comparing their representatives, rather than comparing all translated summands individually. This gives the analogue of fast linear-independence checks in the cycle basis (cf.~\cref{cor:PSlineartimeLIchecks}).

A naive commutator computation between two Pauli cycles in the standard Pauli basis would expand each cycle into $n$ translated copies, leading to $\mathcal{O}(n^2)$ pairwise Pauli string commutators. Since each such commutator can be evaluated in $\mathcal{O}(n)$ time in the standard representation (cf.~\cref{cor:PSlineartimecommutators}), this yields an overall $\mathcal{O}(n^3)$ procedure. The cyclic symmetry reduces this cost by one factor of $n$:

\begin{proposition}[Commutator of Pauli cycles]
\label{th:PCcommutator}
For two Pauli cycles $O_P$ and $O_{P'}$ defined as in \cref{def:Paulicycle}, their commutator closes in the span of Pauli cycles and is given by
\begin{equation}
[O_P,O_{P'}]
=
\frac{1}{n}\sum_{k=0}^{n-1} \ O_{i[P,P'^{(k)}]},
\label{eq:PCcommutator}
\end{equation}
with $P'^{(k)}=U_\tau^kP'(U_\tau^\dagger)^k$.
In particular, the commutator can be evaluated in $\mathcal{O}(n^2)$ time.
\end{proposition}

The identity~\eqref{eq:PCcommutator} follows by expanding $O_P=i\mathcal{T}(P)$ and $O_{P'}=i\mathcal{T}(P')$, using linearity of the commutator, and exploiting the translation invariance of $O_{P'}$. A short derivation is given in Appendix~\ref{prf:PCcommutator}.

\begin{corollary}[Linear-time commutators for bounded-weight Pauli cycles]
\label{cor:bounded_weight_PC_commutator}
Let $w_P$ and $w_{P'}$ denote the Pauli weights of $P$ and $P'$, i.e.\ the numbers of non-identity letters in the generating strings. If $w_P,w_{P'}=\mathcal{O}(1)$, then the commutator of the corresponding Pauli cycles can be computed in $\mathcal{O}(n)$ time.
\end{corollary}

Thus, for translation-invariant generators of bounded Pauli weight, passing to the symmetry-adapted basis of Pauli cycles reduces the cost of Lie-algebraic manipulations to the same asymptotic scaling as for single Pauli strings, making $\gsim$ efficient also in this symmetry-restricted free-fermion setting. A proof for \cref{cor:bounded_weight_PC_commutator} is given in Appendix~\ref{prf:bounded_weight_PC_commutator}. 

\begin{remark}[Translation-invariant sums with open boundaries]
The efficiency of the Pauli cycle formalism extends to translation-invariant sums on open chains via a bulk-boundary decomposition. By formally representing operators as a periodic bulk cycle supplemented by strictly localized boundary defects (e.g., $A = O_{\mathrm{bulk}} + B_{\mathrm{left}} + B_{\mathrm{right}}$), commutators separate into bulk-bulk interactions evaluated via \cref{eq:PCcommutator} and purely local boundary corrections. For bounded-weight generators, these defect interactions overlap on at most $\mathcal{O}(1)$ lattice sites, thus strictly preserving the $\mathcal{O}(n)$ computational complexity established in \cref{cor:bounded_weight_PC_commutator}.
\end{remark}

\begin{example}[Transverse-field Ising model (TFIM)]
\label{ex:TFIM}
The transverse-field Ising Hamiltonian is
\begin{equation}
H_{\mathrm{TFIM}}
=
J\sum_{i=1}^{n-1} Z_iZ_{i+1}
-
g\sum_{i=1}^{n} X_i
\end{equation}
for open boundary conditions, and analogously with $Z_nZ_1$ added for periodic boundaries. A Hamiltonian variational ansatz (HVA) built from this model uses the generator set
\begin{equation}
G_{\mathrm{TFIM}}
=
\left\{
\sum_{i\in V} X_i,\ 
\sum_{(i,j)\in E} Z_iZ_{j}
\right\}
\label{eq:stdTFIMgenerators}
\end{equation}
for $G=(V,E)$ either the path or the cycle graph on $n$ vertices. This is also exactly the generator set of QAOA for MaxCut on a path or cycle graph. Up to single-qubit Clifford conjugations, it is a special case of the free-fermion support pattern~\eqref{eq:chapman_1D_general}; in particular, its DLA is contained in a free-fermion algebra and therefore remains polynomially bounded. More precisely, the exact DLA dimension is $n^2$ for open boundaries~\cite{Larocca.2022-DiagnosingBarrenPlateaus} and $3n-1$ for periodic boundaries~\cite{Allcock.2024-DynamicalLieAlgebras}.
It is instructive to contrast this with the corresponding \emph{freely parameterized} (edge-local) variant, in which each single-qubit and two-qubit term carries an independent parameter, i.e.
\begin{equation}
G_{\mathrm{TFIM}}^{\mathrm{free}}
=
\{X_i\}_{i\in V}\ \cup\ \{Z_iZ_j\}_{(i,j)\in E}.
\label{eq:freeTFIMgenerators}
\end{equation}
The resulting DLAs have been characterized in Ref.~\cite{Kazi.2025-AnalyzingQuantumApproximate}; for a path one obtains $\dim(\mathfrak{g})=n(2n-1)$, while for a cycle $\dim(\mathfrak{g})=2n(2n-1)$.
\end{example}

Numerical benchmarks illustrating the asymptotic advantage of working in the Pauli cycle basis are given in Appendix~\ref{appssec:BMforFFalgebrasTI} for the TFIM instance in~\cref{ex:TFIM}.


\section{Permutation-invariant algebras}
\label{sec:Pauliorbitalgebra}
Permutation symmetry is a natural inductive bias for problems which are invariant under relabeling of qubits, as in exchange-symmetric (collective) many-body settings or graph-based quantum machine learning models. More generally, symmetry-respecting architectures have been observed to outperform problem-agnostic ones, and in the specific $S_n$-equivariant case they can exhibit improved trainability and generalization properties~\cite{Mernyei.2022-EquivariantQuantumGraph,mansky2023permutationinvariantquantumcircuits,Skolik.2023-EquivariantQuantumCircuits,Meyer.2023-ExploitingSymmetryinQML,Schatzki.2024-TheoreticalGuaranteesPermutationequivariant}.
This section proposes a symmetry-adapted operator basis for permutation-invariant dynamics on $n$ qubits. The resulting \emph{Pauli orbit} representation allows us to carry out the preprocessing steps of $\gsim$ (basis construction and structure constants) efficiently even when the generators have exponentially large Pauli expansions.

Let $U_\pi$ denote the unitary representation of $\pi\in S_n$ acting on $(\mathbb{C}^2)^{\otimes n}$ by permuting tensor factors. We define the permutation-invariant operator algebra
\begin{equation}
\g^\mathrm{PI} := \{ A\in \mathfrak{u}(2^n)\,:\, [A,U_\pi]=0 \ \forall \pi\in S_n \},
\label{eq:gPIdefinition}
\end{equation}
where $S_n$ is the symmetric group of all permutations of $n$ elements. This algebra is the commutant of the $S_n$ action and forms a proper subalgebra of $\mathfrak{u}(2^n)$.

\begin{proposition}[Decomposition and dimension of $\g^\mathrm{PI}$]
\label{prop:schur_weyl_commutant}
The algebra $\g^\mathrm{PI}$ is isomorphic to
\begin{equation}
\g^\mathrm{PI} \cong \bigoplus_{k=0}^{\lfloor n/2 \rfloor} \mathfrak{u}(n-2k+1),
\label{eq:PIalgebradecomposition}
\end{equation}
and has dimension
\begin{equation}
\dim(\g^\mathrm{PI})=\sum_{k=0}^{\lfloor n/2 \rfloor} (n-2k+1)^2 = \binom{n+3}{3}.
\label{eq:dimgPI}
\end{equation}
\end{proposition}

\cref{prop:schur_weyl_commutant} is a standard consequence of the Schur--Weyl duality. Related statements on this decomposition and the resulting dimension of $\g^\mathrm{PI}$ appear, e.g., in Refs.~\cite{Albertini.2018-ControllabilitySymmetricSpin,Anschuetz.2023-EfficientClassicalAlgorithms,Allcock.2024-DynamicalLieAlgebras,Nguyen.2024-TheoryEquivariantQuantum,Schatzki.2024-TheoreticalGuaranteesPermutationequivariant,Kazi.2024-UniversalitySnEquivariant,Kazi.2025-AnalyzingQuantumApproximate}. For completeness, we provide a short derivation in Appendix~\ref{prf:schur_weyl_commutant}.

\begin{corollary}[Cubic upper bound for permutation-invariant DLAs]
\label{cor:cubic_upper_bound}
Let $\g=\langle \mathcal{G}\rangle_{\mathrm{Lie}}$ be a DLA generated by a set of permutation-invariant generators $\mathcal{G}\subseteq \g^\mathrm{PI}$. Then $\g\subseteq \g^\mathrm{PI}$ and hence
\begin{equation}
\dim(\g) \le \binom{n+3}{3} = \mathcal{O}(n^3).
\end{equation}
\end{corollary}

\subsection{Explicit basis via Pauli orbits}
\label{ssec:Pauliorbitformalism}

A convenient way to construct explicit elements of $\g^{\mathrm{PI}}$ is to view $S_n$ as acting on the operator space by conjugation. Concretely, for each $\pi\in S_n$ we have a linear map
\begin{equation}
\Ad_{U_\pi}:\mathfrak{u}(2^n)\to \mathfrak{u}(2^n),
\qquad
\Ad_{U_\pi}(A):=U_\pi A U_\pi^\dagger.
\end{equation}
The permutation-invariant algebra $\g^{\mathrm{PI}}$ is precisely the fixed-point subspace of this representation. A standard construction from representation theory, the Reynolds operator, gives the orthogonal projection onto this trivial subrepresentation. Define the symmetrization map
\begin{equation}
\mathcal{S}(\cdot):=\frac{1}{|S_n|}\sum_{\pi\in S_n} \Ad_{U_\pi}(\cdot)
=\frac{1}{n!}\sum_{\pi\in S_n} U_\pi (\cdot) U_\pi^\dagger.
\label{eq:ReynoldsOperator}
\end{equation}
Then $\mathcal{S}(A)\in\g^{\mathrm{PI}}$ for all $A\in\mathfrak{u}(2^n)$, and $\mathcal{S}(A)=A$ for all $A\in\g^{\mathrm{PI}}$. Moreover, with respect to the Hilbert--Schmidt inner product, $\mathcal{S}$ is an orthogonal projector onto $\g^{\mathrm{PI}}$.

This naturally allows for defining a concrete operator basis for $\g^\mathrm{PI}$ (cf.~Ref.~\cite{Anschuetz.2023-EfficientClassicalAlgorithms,Allcock.2024-DynamicalLieAlgebras,chang2026practicalframeworkpermutationequivariant}) that avoids writing generators as explicit sums over exponentially many Pauli strings.

\begin{definition}[Pauli orbit]
\label{def:Pauliorbit}
For an $n$-qubit system, define the Pauli orbit labeled by $(p,q,r)$ as the skew-Hermitian operator
\begin{equation}
B_{p,q,r}:= i\: \mathcal{S}(P_{p,q,r}) \in\g^\mathrm{PI}
\label{eq:Pauliorbit}
\end{equation}
with $p,q,r\in\mathbb N_0$ and $p+q+r\leq n$, where 
\begin{equation}
    P_{p,q,r} = X^{\otimes p}\otimes Y^{\otimes q}\otimes Z^{\otimes r}\otimes I^{\otimes (n-p-q-r)}.
\end{equation} 
\end{definition}
Equivalently, $B_{p,q,r}$ is (up to the global factor $i$) the permutation-invariant sum of all Pauli strings containing $X$, $Y$, and $Z$ on exactly $p$, $q$, and $r$ sites, and identities on the remaining sites.
The expansion of $B_{p,q,r}$ in the standard Pauli basis contains
\begin{equation}
N_{\mathrm{terms}}(p,q,r;n)=\frac{n!}{p!\,q!\,r!\,(n-p-q-r)!}
\label{eq:PStermsinPO}
\end{equation}
distinct Pauli strings, which is exponential in $n$ when $p,q,r=\Theta(n)$. The Pauli orbit representation is exponentially more efficient by storing each basis element $B_{p,q,r}$ instead via the tuple $(p,q,r)$ and the global system size $n$.

\begin{proposition}[Pauli orbits form a basis of $\g^{\mathrm{PI}}$]
\label{prop:orbit_basis}
Let
\begin{equation}
\mathcal{B}_n := \{\, B_{p,q,r}\ :\ p,q,r\in\mathbb N_0,\ p+q+r\le n \,\},
\label{eq:Pauliorbitbasis}
\end{equation}
where $B_{p,q,r}$ is the Pauli orbit defined in \cref{def:Pauliorbit}. Then $\mathcal{B}_n$ is a basis of the permutation-invariant operator algebra $\g^{\mathrm{PI}}$. In particular,
\begin{equation}
\mathrm{span}_\mathbb{R}(\mathcal{B}_n)=\g^{\mathrm{PI}}
\qquad\text{and}\qquad
|\mathcal{B}_n| = \binom{n+3}{3}.
\end{equation}
\end{proposition}
\begin{proof}[Proof sketch]
Every $A\in\mathfrak{u}(2^n)$ has a (unique) expansion in the Pauli string basis. Applying $\mathcal{S}$ termwise groups Pauli strings into $S_n$-orbits under conjugation by permutations, and produces an element of $\g^{\mathrm{PI}}$. Conversely, if $A\in\g^{\mathrm{PI}}$ then $\mathcal{S}(A)=A$, so $A$ equals a linear combination of such orbit-sums. Choosing for each orbit a canonical representative of the form $P_{p,q,r}$ yields spanning by $\mathcal{B}_n$. To see linear independence, note that each $B_{p,q,r}$ has disjoint Pauli support for distinct triples $(p,q,r)$. Since Pauli strings form an orthogonal basis under the Hilbert--Schmidt inner product, the $B_{p,q,r}$ are mutually orthogonal and hence linearly independent. The full proof can be found in Appendix~\ref{prf:Pauliorbitbasis}.
\end{proof}

\cref{prop:orbit_basis} is consistent with \cref{prop:schur_weyl_commutant}: the set $\mathcal{B}_n$ has cardinality $\binom{n+3}{3}$, matching $\dim(\g^{\mathrm{PI}})$. 

Finally, note that $B_{0,0,0}$ is proportional to the central element $iI_n$. If desired one may pass to $\mathfrak{su}(2^n)$ by omitting this single basis element.

\subsection{Efficient primitives: inner products and commutators}
\label{ssec:POalgebraPRIMITIVES}

The preprocessing stage of $\gsim$ requires efficient evaluation of inner products and commutators in order to construct a DLA basis and compute the associated structure constants (cf.~Fig.~\ref{fig:g-sim-outline}).

\begin{lemma}[Hilbert--Schmidt inner product of Pauli orbits]
\label{th:POinnerproduct}
For Pauli orbits $B_{p,q,r}$ and $B_{p',q',r'}$,
\begin{equation}
\langle B_{p,q,r}, B_{p',q',r'} \rangle
=
\delta_{(p,q,r),(p',q',r')}\, \|B_{p,q,r}\|^2.
\end{equation}
In particular, the normalized inner product required for structure constants reduces to an equality check of the orbit labels and can be evaluated in constant time. Moreover, under the distinct string convention,
\begin{equation}
    \|B_{p,q,r}\|^2 = \frac{2^n}{N_{\mathrm{terms}}(p,q,r;n)}.
\end{equation}
\end{lemma}

\begin{theorem}[Commutator of Pauli orbits]
\label{th:POcommutator}
For two Pauli orbits $B_{p,q,r}$ and $B_{p',q',r'}$ defined as in \cref{def:Pauliorbit}, the commutator closes in the Pauli orbit basis:
\begin{equation}
[B_{p,q,r}, B_{p',q',r'}]
= \sum_{(\tilde{p}, \tilde{q}, \tilde{r})}\: c_{\tilde{p}, \tilde{q}, \tilde{r}}\: B_{\tilde{p}, \tilde{q}, \tilde{r}},
\label{eq:POcommutator}
\end{equation}
with $B_{\tilde{p}, \tilde{q}, \tilde{r}}\in\mathcal{B}_n$.
In particular, at most $|\mathcal{B}_n|=\mathcal{O}(n^3)$ distinct output orbits can appear. Moreover, the coefficients $c_{\tilde{p}, \tilde{q}, \tilde{r}}$ admit an explicit finite-sum formula over contingency tables and can be evaluated in $\mathcal{O}(\poly(n))$ time; in the worst case, evaluating the full commutator by direct enumeration costs $O(n^9)$ arithmetic operations after factorial precomputation.
\end{theorem}
The explicit coefficient formula, algorithms, and runtime benchmarks are given in Appendix~\ref{prf:POcommutator}. Here, we merely emphasize that, for $\gsim$, commutators close within the orbit basis and remain efficiently evaluable despite the exponential Pauli support of individual orbit elements.

\begin{corollary}[Constant-time commutator coefficients for bounded-weight generators]
\label{cor:POcommutator_constant_time}
Fix any finite target set of output orbits $T=\{B_{\tilde p,\tilde q,\tilde r}\in \mathcal{B}_n\ |\ \tilde p+\tilde q+\tilde r=\mathcal{O}(1)\}\subset\mathcal{B}_n$. Then for any two Pauli orbits $B_{p,q,r}$ and $B_{p',q',r'}$, all targeted coefficients $\{c_{\tilde p,\tilde q,\tilde r}:(\tilde p,\tilde q,\tilde r)\in T\}$ in \eqref{eq:POcommutator} can be computed in constant time. Equivalently, the enumeration of valid contingency tables is bounded by
\begin{equation}
    \mathcal{O}\left(\sum_{(\tilde p,\tilde q,\tilde r)\in T}
\binom{\tilde p+3}{3}\binom{\tilde q+3}{3}\binom{\tilde r+3}{3}\right),
\end{equation}
which is $O(1)$ for $|T|=\mathcal{O}(1)$.
\end{corollary}

In particular, for variational ansätze generated by a constant number of symmetric $k$-local Hamiltonians with $k$ independent of the system size $n$, the commutator information required to build the relevant adjoint representation involves only bounded-weight orbit labels. Hence, the runtime cost of computing the commutators for the respective structure constants~\eqref{eq:structureconstants} are $t_\mathrm{com}=\mathcal{O}(1)$).

\begin{example}[Permutation-equivariant quantum neural networks (QNNs)]
A canonical bounded-weight generator set in $\g^{\mathrm{PI}}$ is
\begin{equation}
\begin{aligned}
G
& =\left\{\sum_{i=1}^n X_i,\ \sum_{i=1}^n Y_i,\ \sum_{1\le i<j\le n} Z_i Z_j\right\}\\
& \equiv \{B_{1,0,0},\,B_{0,1,0},\,B_{0,0,2}\},
\end{aligned}
\label{eq:peQNNgenerators}
\end{equation}
where the equivalence is in the Pauli orbit representation introduced above. This set appears, for instance, in permutation-equivariant quantum neural networks~\cite{Larocca.2022-GroupInvariantQML,Schatzki.2024-TheoreticalGuaranteesPermutationequivariant,Nguyen.2024-TheoryEquivariantQuantum} and is known to generate the full commutant algebra up to the central identity, i.e.\ it spans $\g^{\mathrm{PI}}\cap\mathfrak{su}(2^n)$~\cite{Albertini.2018-ControllabilitySymmetricSpin,Schatzki.2024-TheoreticalGuaranteesPermutationequivariant}. Since all generators have bounded Pauli weight, the targeted commutator coefficients required for $\gsim$ preprocessing fall into the constant-time regime of Corollary~\ref{cor:POcommutator_constant_time}.
\label{ex:SN1}
\end{example}

\begin{example}[QAOA for MaxCut on a complete graph]
Consider the QAOA-type generator set
\begin{equation}
G=\left\{\sum_{i=1}^n X_i,\ \sum_{1\le i<j\le n} Z_i Z_j\right\}
\equiv \{B_{1,0,0},\,B_{0,0,2}\},
\end{equation}
which corresponds to a MaxCut problem defined on a complete graph (and related fully-connected tranverse field Ising problems). The associated DLA is permutation-invariant and has been analyzed in Ref.~\cite{Allcock.2024-DynamicalLieAlgebras}. Again, both generators have bounded weight, so the commutator information needed to build their adjoint representations can be obtained in constant time per targeted coefficients (Corollary~\ref{cor:POcommutator_constant_time}).
\label{ex:SN2}
\end{example}

This highlights that most practical cases of $S_n$-invariant quantum systems can be efficiently simulated with $\gsim$. While Theorem~\ref{th:POcommutator} provides a uniform worst-case polynomial guarantee, the practically relevant regime of low-weight (few-body) generators leads to much smaller effective search spaces: the required commutator information reduces to a constant-size combinatorial enumeration independent of $n$ (Corollary~\ref{cor:POcommutator_constant_time}). Moreover, even when computing large collections of commutators within $\mathcal{B}_n$ (as in structure constant preprocessing), our benchmarks indicate that the \emph{observed} average-case runtime per commutator grows only mildly with $n$ reflecting that most orbit pairs lie near the boundary of the feasible contingency polytope and admit comparatively few overlap patterns (cf.~Appendix~\ref{appssec:PObenchmarksstructureconstants}).

\subsection{Algorithmic consequence for $\gsim$}
\label{ssec:PO_gsim_consequence}

\cref{th:POinnerproduct} and \cref{th:POcommutator} provide the primitives to execute the $\gsim$ preprocessing pipeline directly in the Pauli orbit basis: constant-time inner products and efficiently computable commutators (via the explicit contingency-table formula in Appendix~\ref{prf:POcommutator}). As a result, Lie closure and structure constant computation can be carried out in polynomial-time for permutation-invariant generator sets even when their expansions in the standard Pauli basis contain exponentially many terms.

\begin{corollary}[Efficient $\gsim$ for permutation-invariant circuits]
\label{cor:efficient_gsim_PI}
Let $U(\boldsymbol{\theta})$ be a circuit of the form~\eqref{eq:standardansatz} whose dynamical generators lie in $\g^\mathrm{PI}$ and are provided in Pauli orbit representation. Then the $\gsim$ preprocessing stage (basis construction and structure constants) can be implemented using arithmetic resources polynomial in $n$ and in the resulting Lie-algebra dimension $\dim(\g)\le \binom{n+3}{3}$.
In particular, commutators between basis elements can be evaluated in worst-case time $\mathcal{O}(n^9)$ (and, in the bounded-weight targeted regime relevant to few-body generators, with $\mathcal{O}(1)$ enumeration cost per required coefficient; see Appendix~\ref{appssec:POcommutatorevaluationtargeted}).
\end{corollary}

Corollary~\ref{cor:efficient_gsim_PI} shows that exponential Pauli support of individual generators does not preclude efficient $\gsim$ preprocessing (as conjectured as a possible requirement in Ref.~\cite{Heredge.2025-CharacterizingPrivacyQuantum}). Instead, permutation symmetry admits a compact orbit description in which the Lie-algebraic operations needed by $\gsim$ are polynomial-time (and often substantially faster in typical few-body settings).

\begin{remark}[Relation to Schur-basis simulation methods]
Permutation-invariant systems can also be simulated by working directly in the Schur basis. In particular, Ref.~\cite{Anschuetz.2023-EfficientClassicalAlgorithms} developed a general Schur-basis framework for symmetric quantum circuits, and Ref.~\cite{chang2026practicalframeworkpermutationequivariant} recently gave a substantially more efficient sparse Schur-block formulation for local $S_n$-invariant generators. Our approach is complementary: rather than simulating in the Schur basis, we represent permutation-equivariant dynamics natively within the adjoint-space formalism required by $\gsim$ using the Pauli orbit basis. This basis yields the Lie-algebraic data required for efficient computation of the structure constants needed for adjoint space propagation, and places the permutation-invariant setting on the same footing as the free-fermionic and bounded-Hamming-weight families treated in this paper. Appendix~\ref{appsec:schurcomparison} discusses the comparison in more detail.
\end{remark}

\section{Bounded Hamming-weight and few-body $U(1)$-invariant algebras}
\label{sec:restrictedsubspace-algebras}
Another non-trivial source of polynomially growing DLAs, beyond free-fermion algebras (with ${\dim(\g^{\mathrm{FF}}){=}\mathcal{O}(n^2)}$) and permutation-invariant algebras (with ${\dim(\g^{\mathrm{PI}}){=}\mathcal{O}(n^3)}$), is given by Hamming-weight (HW) preserving circuits, which are $U(1)$-equivariant, meaning that they act invariantly on fixed-excitation sectors of the computational basis. Under standard fermionic encodings (such as the Jordan--Wigner or parity mapping), HW-preservation corresponds strictly to particle-number conservation, making these circuits the natural algebraic setting for quantum chemistry~\cite{Gard.2020-EfficientSymmetrypreservingState,Tilly.2022-VariationalQuantumEigensolver}. $U(1)$-equivariant circuits have also received increased attention as expressive yet symmetry-respecting ansätze in QML~\cite{Kerenidis.2022-QuantumMachineLearning,Cherrat.2023-QuantumDeepHedging,Cherrat.2024-QuantumVisionTransform,Monbroussou.2025-SubspacePreservingQCNNarchitectures}.

Let
\begin{equation}
S_Z := \sum_{i=1}^{n} Z_i
\label{eq:totalmagnetization}
\end{equation}
denote the total magnetization operator. We define the HW-preserving (equivalently, $U(1)$-invariant) operator algebra as the commutant
\begin{equation}
\g^{\mathrm{HW}}
:=
\{A\in \mathfrak{u}(2^n)\ :\ [A,S_Z]=0\}.
\label{eq:HWpreservingALGEBRA}
\end{equation}
Equivalently, one may use the excitation-number operator
\begin{equation}
N:=\sum_{i=1}^n \frac{I-Z_i}{2}
=\frac{n}{2}I-\frac{1}{2}S_Z,
\label{eq:excitation_number}
\end{equation}
so that $[A,S_Z]=0$ if and only if $[A,N]=0$. In either form, $\g^{\mathrm{HW}}$ consists precisely of those operators that are invariant under the global $U(1)$ action generated by $S_Z$, i.e.\ conjugation by $e^{-i\frac{\theta}{2} S_Z}$.

The eigenspaces of $N$ (or $S_Z$) define the fixed Hamming-weight sectors
\begin{equation}
\begin{gathered}
\mathcal{H}_k := \mathrm{span}\{\,|x\rangle : x\in\{0,1\}^n,\ |x|=k\,\},\\
\ \text{with}\quad \dim(\mathcal{H}_k)=d_k:=\binom{n}{k}.
\label{eq:k-weight-subspace-dimension}
\end{gathered}
\end{equation}
Since every $A\in \g^{\mathrm{HW}}$ commutes with $N$, it cannot mix basis states of different Hamming weight. Consequently, every element of $\g^{\mathrm{HW}}$ is block diagonal with respect to the decomposition ${\mathcal{H}=\bigoplus_{k=0}^{n}\mathcal{H}_k}$. 

\begin{proposition}[Decomposition and dimension of $\g^\mathrm{HW}$]
The $U(1)$-invariant Lie algebra $\g^\mathrm{HW}$ decomposes as a direct sum of orthogonal subspace algebras,
\begin{equation}
\mathfrak{g}^{\mathrm{HW}} \cong \bigoplus_{k=0}^n\ \mathfrak{u}(d_k)
\label{eq:HWalgebradecomp}
\end{equation}
and therefore has dimension
\begin{equation}
\dim(\g^\mathrm{HW})=\sum_{k=0}^n d_k^2=\binom{2n}{n}=\mathcal{O}\left(\frac{4^n}{\sqrt{n}}\right).
\label{eq:HWalgebradim}
\end{equation}
In particular, circuits generated by elements of $\mathfrak{g}^{\mathrm{HW}}$ act independently on each $\mathcal{H}_k$.
\label{prop:HWalgebradecompdim}
\end{proposition}
A full derivation of \cref{prop:HWalgebradecompdim} is given in Appendix~\ref{prf:HWalgebradecompdim}

A convenient “universal” parametrization of two-qubit HW-preserving Hamiltonians follows from the observation that HW-preservation on qubits $i,j$ means that the Hamiltonian acts nontrivially only on $\operatorname{span}\{\ket{01},\ket{10}\}$. 

\begin{definition}[Universal HW-preserving generators]
\label{def:HWgenerators}
Up to terms acting only on the trivial $|00\rangle$ and $|11\rangle$ sectors, any two-qubit HW-preserving Hermitian generator can be written as~\cite{Yan.2024-UniversalHammingWeight}
\begin{equation}
    H^{\mathrm{HW}}_{ij}=\left(\begin{array}{cccc}
0 & 0 & 0 & 0 \\
0 & a & b & 0 \\
0 & \bar{b} & c & 0 \\
0 & 0 & 0 & 0
\end{array}\right)_{ij}=e\, \mathbf{E}_{i j}+s\, \mathbf{S}_{i j}+r\, \mathbf{R}_{i j}+ j\,\mathbf{J}_{i j},
\label{eq:HWgenerator}
\end{equation}
where $a,c\in\mathbb{R}$ and $b\in\mathbb C$, and we use the real parameters 
\begin{equation}
e= \frac{a+c}{2},\quad s=\frac{a-c}{2},\quad r=\operatorname{Re}(b), \quad j = \operatorname{Im}(b).
\label{eq:HWgeneratorDECOMPparams}
\end{equation}
The four basis terms admit the Pauli representation 
\begin{align}
\begin{split}
    &\mathbf{E}_{i j}=\frac{1}{2}(I-Z_iZ_j),\quad \mathbf{S}_{i j}=\frac{1}{2}(Z_i-Z_j),\\
    &\mathbf{R}_{i j}=\frac{1}{2}(X_iX_j+Y_iY_j),\quad  \mathbf{J}_{i j}= \frac{1}{2}(X_iY_j-Y_iX_j).
\end{split}
\label{eq:HWgeneratorDECOMPOSED}
\end{align}
\end{definition}
The decomposition~\eqref{eq:HWgeneratorDECOMPOSED} makes the boundary between free-fermionic (Gaussian) and genuinely interacting HW-preserving dynamics explicit. The terms $\mathbf{S}_{ij}$, $\mathbf{R}_{ij}$, and $\mathbf{J}_{ij}$ are quadratic under the Jordan--Wigner transformation and thus lie within the free-fermion algebra $\g^\mathrm{FF}$~\eqref{eq:FF_so2n_dim} (cf.\ \cref{sec:FFalgebras}). In particular, circuits restricted to these generators (including matchgates and nearest-neighbor Reconfigurable Beam Splitter (RBS) gates~\cite{Gard.2020-EfficientSymmetrypreservingState,Tilly.2022-VariationalQuantumEigensolver,Cherrat.2023-QuantumDeepHedging,Monbroussou.2025-TrainabilityExpressivityHammingweight}) generate DLAs of dimension $\mathcal{O}(n^2)$, independently of which subspace $\mathcal{H}_k$ is considered. By contrast, $\mathbf{E}_{ij}$ contains an interaction component (i.e., $Z_iZ_j$) which breaks free-fermionic integrability and can lead to full controllability on fixed-HW subspaces. 

Ref.~\cite{Yan.2024-UniversalHammingWeight} makes this precise: under full connectivity, subspace universality on $\mathcal{H}_k$ holds if and only if $e\neq0$ and either $j\neq 0$ or ($r\neq0\,\wedge\,s\neq0$). In these interacting cases, the induced DLA on $\mathcal{H}_k$ becomes the full special unitary algebra $\mathfrak{su}(d_k)$ with dimension $d_k^2-1$ (or $\mathfrak{u}(d_k)$ with dimension $d_k^2$, if the central identity is included). In contrast, if the circuit is restricted to a real orthogonal gate set on $\mathcal{H}_k$ (e.g., non-local $\mathbf{J}_{i j}$ generators), the induced algebra is contained in $\mathfrak{so}(d_k)$ with dimension $\frac{1}{2}d_k(d_k - 1)$~\cite{Monbroussou.2025-TrainabilityExpressivityHammingweight}. These regimes therefore yield polynomially sized, yet highly expressive, DLAs beyond both free-fermion and permutation-invariant families. 

\begin{proposition}[Polynomial DLAs for bounded-HW circuits]
\label{prop:bounded_hw_dla}
Although the global $U(1)$-invariant Lie algebra $\g^\mathrm{HW}$ has exponential dimension (cf.\ \cref{prop:HWalgebradecompdim}), the effective dynamics of any HW-preserving circuit initialized in a fixed sector $\mathcal{H}_k$ is governed by the restriction of the DLA to that subspace, which is a subalgebra of $\mathfrak{u}(d_k)$. 
If $k$ is bounded (i.e., $k = \mathcal{O}(1)$), then $d_k = \binom{n}{k} = \mathcal{O}(n^k)$ and hence any maximally expressive interacting DLA on $\mathcal{H}_k$ has polynomial dimension
\begin{equation}
    \dim(\mathfrak{u}(d_k))=d_k^2 = \mathcal{O}(n^{2k}),
\end{equation}
and is efficiently simulable via $\gsim$.
\end{proposition}
\begin{proof}
Let $\mathfrak{g}\subseteq \g^{\mathrm{HW}}$ be the DLA generated by the HW-preserving circuit. By definition~\eqref{eq:HWpreservingALGEBRA}, any generator $A\in\mathfrak{g}$ satisfies $[A,N]=0$. Hence, $A$ preserves the eigenspaces of $N$~\eqref{eq:excitation_number}, implying that each fixed-weight subspace $\mathcal{H}_k$ is invariant under the action of $\mathfrak{g}$. Therefore the effective algebra governing the evolution of any initial state supported on $\mathcal{H}_k$ is a subalgebra of the space of all skew-Hermitian operators on a $d_k$-dimensional Hilbert space, which is exactly $\mathfrak{u}(d_k)$ with $\dim(\mathfrak{u}(d_k))=d_k^2$.
If $k=\mathcal{O}(1)$, then $d_k=\binom{n}{k}\le n^k/k!=\mathcal{O}(n^k)$, and consequently $\dim(\mathfrak{u}(d_k))=\mathcal{O}(n^{2k})$. 
\end{proof}

The same conclusion applies to sectors with a bounded number of holes. Indeed, since $d_{n-k}=\binom{n}{n-k}=\binom{n}{k}=d_k$, the sector $\mathcal H_{n-k}$ has the same dimension as $\mathcal H_k$. Equivalently, the global bit-flip $X^{\otimes n}$ maps $\mathcal H_k$ to $\mathcal H_{n-k}$ and sends the excitation-number operator $N$ to $nI-N$. Thus, all bounded-$k$ statements below apply equally whenever either the number of excitations $k$ or the number of holes $n-k$ is $O(1)$.

Consequently, for sectors with either a constant number of excitations or a constant number of holes, HW-preserving circuits that break matchgate structure remain efficiently simulable by $\gsim$ when the simulation is carried out in the relevant invariant sector. While the $k=1$ case represents a single-excitation manifold (physically corresponding to non-interacting particle hopping), for $k\ge 2$, the system supports genuine particle interactions. In this truly interacting regime of $k\ge 2$, the resulting $\mathcal{O}(n^{2k})$ DLA dimension is at least quartic, exceeding both the free-fermion $\mathcal{O}(n^2)$ and permutation-invariant $\mathcal{O}(n^3)$ regimes. This defines a distinct class of tractable yet expressive quantum dynamics.

\begin{remark}[Counterexample for $\dim(\g){=}\mathcal{O}(n^2)$ implying $\g\subseteq\g^\mathrm{FF}$]
It is often suggested in recent literature on Pauli Lie algebras that all quadratically scaling dynamical Lie algebras correspond to free-fermions, i.e., they are subalgebras of $\mathfrak{so}(2m)$ or $\mathfrak{so}(2m+1)$ (cf.\ \cref{sec:FFalgebras}). 
A simple physically relevant counterexample is provided by HW-preserving dynamics restricted to the single-excitation sector $\mathcal{H}_1$, which has dimension $d_1=n$. With a universal generator set as in \cref{def:HWgenerators}, the induced Lie algebra is $\mathfrak{su}(n)$~\cite{Yan.2024-UniversalHammingWeight} and hence has quadratic dimension $\dim(\mathfrak{su}(n))=n^2-1$. In general, $\mathfrak{su}(n)$ is not isomorphic to $\mathfrak{so}(2m)$ or $\mathfrak{so}(2m{+}1)$ (up to the exceptional low-dimensional coincidences $\mathfrak{su}(2)\cong\mathfrak{so}(3)$ and $\mathfrak{su}(4)\cong\mathfrak{so}(6)$). Thus, even though the underlying physics corresponds to a single-particle (non-interacting) sector, the resulting quadratic-dimensional DLA is algebraically distinct from the matchgate/free-fermion $\mathfrak{so}$ family.
\label{rem:Counterexample-quadratic=FF}
\end{remark}

In the following, we discuss how to apply $\gsim$ to bounded-HW algebras in practice. 

\subsection{Explicit basis and primitives for $\gsim$}
\label{ssec:HWalgebraBasisPrimitives}
As in the settings discussed previously, efficient Lie-algebraic simulation hinges on an explicit operator basis for the relevant DLA together with efficient primitives for linear-independence tests, commutators, and hence structure constant evaluation. In the bounded-HW regime, the dynamics of interest is confined to a fixed sector $\mathcal{H}_k$ of dimension $d_k=\binom{n}{k}$, so the natural Lie algebra for simulation is the full subspace algebra $\mathfrak{u}(d_k)$ (and its subalgebras such as $\mathfrak{so}(d_k)\subset \mathfrak{su}(d_k)\subset \mathfrak{u}(d_k)$). For $k=\mathcal{O}(1)$, we have $\dim(\mathfrak{u}(d_k))=d_k^2=\poly(n)$ by \cref{prop:bounded_hw_dla}. 

A key technical difference from earlier sections is that Pauli strings, while linearly independent on the full Hilbert space $(\mathbb{C}^2)^{\otimes n}$, become highly redundant after restriction to a fixed-weight sector $\mathcal{H}_k$: many distinct Pauli strings induce the same operator on $\mathcal{H}_k$, and many others restrict to zero. Consequently, the Pauli string basis is not well suited as a basis of $\mathfrak{u}(d_k)$. We therefore propose an explicit matrix-unit basis closely related to the generalized Gell--Mann (GGM) basis for $\mathfrak{su}(d)$~\cite{Kimura.2003-BlochVectorforNlevel,Bertlmann.2008-BlochVectorsQudits}, which we extend to $\mathfrak{u}(d)$ (skew-Hermitian generators) by including general diagonal projectors.

\begin{definition}[Modified Generalized Gell-Mann (MGGM) basis for $\mathfrak{u}(d_k)$]
\label{def:MGGMbasisforHWalgebras}
Fix an ordering $\{\ket{a}\}_{a=1}^{d_k}$ of all computational basis states of Hamming weight $k$, so that $\mathcal{H}_k=\operatorname{span}(\ket{1},\dots,\ket{d_k})$. Define, for $1\le a<b \le d_k$,
\begin{align}
\Lambda_A^{ab}&:=|a\rangle\langle b|-|b\rangle\langle a|,\\
\Lambda_S^{ab}&:=i(\ket{a}\bra{b}+\ket{b}\bra{a}),
\label{eq:MGGM_offdiag}
\end{align}
and for $1\leq a\leq d_k$,
\begin{equation}
    \Lambda_P^{a}:=i\,\ket{a}\bra{a}.
\label{eq:MGGM_diag}
\end{equation}
Then the joint set
\begin{equation}
\mathcal{M}_{n}^{(k)}:=
\{\Lambda_A^{ab}\}_{a<b}\ \cup\ \{\Lambda_S^{ab}\}_{a<b}\ \cup\ \{\Lambda_P^{a}\}_{a}
\label{eq:MGGMbasis}
\end{equation}
forms a (real) basis of $\mathfrak{u}(d_k)$. Moreover,
$|\{\Lambda_A^{ab}\}_{a<b}|=|\{\Lambda_S^{ab}\}_{a<b}|=\tfrac12 d_k(d_k-1)$ and $|\{\Lambda_P^{a}\}_{a}|=d_k$, hence $|\mathcal{M}_{n}^{(k)}|=d_k^2$. Finally, $\mathrm{span}\{\Lambda_A^{ab}\}_{a<b}\cong \mathfrak{so}(d_k)$.
\end{definition}
Note that the standard GGM basis for $\mathfrak{su}(d)$ uses $d_k-1$ traceless diagonal generators, which we replace by the $d_k$ diagonal projectors $\{\Lambda_P^{a}\}_{a}$. This yields a particularly transparent representation of HW-preserving generators, and directly spans $\mathfrak{u}(d_k)$.

To apply $\gsim$ on $\mathcal{H}_k$, we need to express the dynamical generators (e.g.\ the two-qubit HW-preserving terms of~\cref{def:HWgenerators}) in the basis $\mathcal{M}_n^{(k)}$. Intuitively, a two-qubit HW-preserving gate on qubits $i,j$ only couples weight-$k$ basis states that differ by swapping the occupation of $i$ and $j$. This gives rise to a sparse representation in the MGGM basis supported only on those basis pairs.

\begin{proposition}[MGGM basis representation of HW-preserving generators]
\label{prop:HWgeneratosinMGGMbasis}
Fix $k$ and an ordering $\{|a\rangle\}_{a=1}^{d_k}$ of $\mathcal{H}_k$ as in \cref{def:MGGMbasisforHWalgebras}. For a fixed pair of qubits $i,j$, define the set of $\emph{active basis pairs}$
\begin{equation}
C_{ij}^{(k)}
:=
\left\{
(a,b): a<b,\;
x^{(b)} = x^{(a)}\oplus e_i\oplus e_j,\;
x^{(a)}_i+x^{(a)}_j=1
\right\},
\label{eq:HWactivebasepairs}
\end{equation}
where $x^{(a)}\in\{0,1\}^n$ is the bit string corresponding to $|a\rangle$, $e_i$ is the unit bit string at site $i$, and $\oplus$ denotes bitwise addition modulo two. Equivalently, $C_{ij}^{(k)}$ contains precisely those pairs of HW-$k$ basis states that agree on all sites except $i,j$, where their local patterns are $01$ and $10$.
Then, on $\mathcal{H}_k$, the four HW-preserving basis terms from \cref{eq:HWgeneratorDECOMPOSED} admit the MGGM expansions
\begin{equation}
\begin{aligned}
i\,\mathbf{J}_{ij}\big|_{\mathcal{H}_k}
&=\sum_{(a,b)\in C_{ij}^{(k)}} \Lambda_A^{ab},\\
i\,\mathbf{R}_{ij}\big|_{\mathcal{H}_k}
&=\sum_{(a,b)\in C_{ij}^{(k)}} \Lambda_S^{ab},\\
i\,\mathbf{S}_{ij}\big|_{\mathcal{H}_k}
&=\sum_{(a,b)\in C_{ij}^{(k)}} \bigl(\Lambda_P^{a}-\Lambda_P^{b}\bigr),\\
i\,\mathbf{E}_{ij}\big|_{\mathcal{H}_k}
&=\sum_{(a,b)\in C_{ij}^{(k)}} \bigl(\Lambda_P^{a}+\Lambda_P^{b}\bigr),
\end{aligned}
\label{eq:HWgeneratortermsinMGGMbasis}
\end{equation}
where the identification of indices $a,b$ depends on the chosen ordering of HW-$k$ basis states.
Consequently, any two-qubit HW-preserving generator $H_{ij}$ of the form in~\cref{def:HWgenerators} restricts to
\begin{equation}
\begin{gathered}
iH_{ij}\big|_{\mathcal{H}_k}
= \\
\sum_{(a,b)\in C_{ij}^{(k)}}
\left[
j\,\Lambda_A^{ab}
+
r\,\Lambda_S^{ab}
+
(s+e)\,\Lambda_P^{a}
+
(e-s)\,\Lambda_P^{b}
\right],
\end{gathered}
\label{eq:HWgeneratosinMGGMbasis}
\end{equation}
with $e,s,r,j$ as in \cref{eq:HWgeneratorDECOMPparams}.
\end{proposition}
For fixed $k$, the number of active pairs scales as ${|C_{ij}^{(k)}|=\binom{n-2}{k-1}=\Theta(n^{k-1})}$, since a basis state contributes if and only if exactly one of the two sites $(i,j)$ is occupied. A detailed derivation of \cref{prop:HWgeneratosinMGGMbasis} is given in Appendix~\ref{prf:HWgeneratosinMGGMbasis}.

For $\gsim$, commutators between basis elements are the critical primitive, since they generate Lie closure and the structure constants required for efficient observable evolution in the adjoint space. In the MGGM basis, all commutators can be written in closed form using Kronecker-delta identities, reducing them to simple index matching problems.

\begin{lemma}[Commutation relations in the MGGM basis]
For $1\le a<b\le d_k$ and $1\le c<d\le d_k$, the MGGM basis elements satisfy
\begin{equation}
\begin{gathered}
{\left[\Lambda_P^{a}, \Lambda_P^{b}\right]=0}, \\
{\left[\Lambda_A^{ab},\Lambda_P^{c}\right]=\delta_{b c}\, \Lambda_S^{ac}-\delta_{a c}\, \Lambda_S^{bc}}, \\
{\left[\Lambda_S^{ab},\Lambda_P^{c}\right]=-\delta_{b c}\, \Lambda_A^{ac}-\delta_{a c}\, \Lambda_A^{bc}}, \\[2mm]
{\left[\Lambda_A^{ab},\Lambda_A^{cd}\right]=\delta_{b c}\, \Lambda_A^{ad}-\delta_{bd}\, \Lambda_A^{ac}-\delta_{ac}\, \Lambda_A^{bd}+\delta_{ad}\, \Lambda_A^{bc}}, \\
{\left[\Lambda_S^{ab},\Lambda_S^{cd}\right]=-\delta_{b c}\, \Lambda_A^{ad}-\delta_{bd}\, \Lambda_A^{ac}-\delta_{ac}\, \Lambda_A^{bd}-\delta_{ad}\, \Lambda_A^{bc}}, \\
{\left[\Lambda_A^{ab},\Lambda_S^{cd}\right]=\delta_{b c}\, \Lambda_S^{ad}+\delta_{bd}\, \Lambda_S^{ac}-\delta_{ac}\, \Lambda_S^{bd}-\delta_{ad}\, \Lambda_S^{bc}}.
\end{gathered}
\label{eq:MGGM_comm_relations}
\end{equation}
In particular, each commutator expands into at most $\mathcal{O}(1)$ basis elements, and can thus be evaluated in constant time given the indices.
\label{lm:MGGM_comm_relations}
\end{lemma}

These relations follow by direct multiplication from the definitions in~\cref{def:MGGMbasisforHWalgebras}. Since commutators are constant-time and sparse in this basis, the structure constants of $\mathfrak{u}(d_k)$ can be generated efficiently without explicitly forming dense $d_k\times d_k$ matrices. Thus the usual basis-construction bottleneck of Lie-algebraic simulation is avoided: the full sector algebra is known analytically, and subalgebras generated by a particular architecture can be represented by their coordinates in this fixed basis.

\begin{theorem}[Structure constants of $\mathfrak{u}(d_k)$ in the MGGM basis]
Let $\mathcal{M}_{n}^{(k)}=\{\Lambda_\alpha\}_{\alpha=1}^{d_k^2}$ be an ordered MGGM basis as in \cref{def:MGGMbasisforHWalgebras}. Then the full set of \emph{nonzero} structure constants
\begin{equation}
[\Lambda_\alpha,\Lambda_\beta]=\sum_{\gamma} f_{\alpha\beta}^{\gamma}\,\Lambda_\gamma
\end{equation}
can be enumerated and computed in total time $\mathcal{O}(d_k^3)$.
\label{th:HWstructureconstants-dk3}
\end{theorem}

\begin{proof}
Each MGGM basis element depends on at most two state indices: $\Lambda_P^{a}$ carries one index $a$, while $\Lambda_A^{ab}$ and $\Lambda_S^{ab}$ carry a pair $(a,b)$. By \cref{lm:MGGM_comm_relations}, a commutator can be nonzero only if the two operands share at least one state index (otherwise all Kronecker deltas vanish). Hence nonzero commutators are supported on index patterns of size at most three: they are determined by choosing either a triple $(a,b,c)$ of distinct indices (e.g.\ commutators among $\Lambda_{\bullet}^{ab}$ and $\Lambda_{\bullet}^{bc}$), or a pair $(a,b)$ together with one projector index (e.g.\ commutators between $\Lambda_{\bullet}^{ab}$ and $\Lambda_P^{a}$ or $\Lambda_P^{b}$). The number of such index patterns is $\mathcal{O}(d_k^3)$. For each fixed pattern, \cref{lm:MGGM_comm_relations} gives an explicit $\mathcal{O}(1)$-term expansion, so computing all corresponding nonzero structure constants costs $\mathcal{O}(d_k^3)$ total time.
\end{proof}

\begin{lemma}[Targeted adjoint matrices for sparse MGGM generators]
\label{lem:targetedMGGMadjoint}
Let $\mathcal{M}_{n}^{(k)}=\{\Lambda_\alpha\}_{\alpha=1}^{d_k^2}$ be the MGGM basis of~\cref{def:MGGMbasisforHWalgebras}, and let
\begin{equation}
    G=\sum_{\mu\in S} c_\mu \Lambda_\mu
\end{equation}
be a generator whose support $S$ in the MGGM basis has size $|S|=s$. Then, the sparse adjoint matrix $\Phi^{\ad}(G)$ defined by~\cref{eq:adjoint_generator_definition,eq:adjoint_generator_matrix} has at most $\mathcal{O}(s\,d_k)$ nonzero entries and can be assembled in $\mathcal{O}(s\,d_k)$ time, assuming constant-time access to the MGGM index map. In particular, if a circuit uses a constant number of generators, each with $s=\mathcal{O}(1)$ support in the MGGM basis, then all adjoint matrices required for $\gsim$ can be constructed in total time $\mathcal{O}(d_k)=\mathcal{O}(n^k)$, after the one-time construction of the basis indexing data.
\end{lemma}
\begin{proof}
By the index-sharing argument used in the proof of \cref{th:HWstructureconstants-dk3}, a fixed MGGM basis element $\Lambda_\mu$ can have nonzero commutator only with basis elements $\Lambda_\beta$ whose index set intersects that of $\Lambda_\mu$. Since $\Lambda_\mu$ contains at most two state indices, there are only $\mathcal{O}(d_k)$ such basis elements $\Lambda_\beta$, and each corresponding commutator has an $\mathcal{O}(1)$-term expansion by \cref{lm:MGGM_comm_relations}. Thus each active term in $G$ contributes $\mathcal{O}(d_k)$ nonzero entries to $M_G$, giving $\mathcal{O}(s\,d_k)$ entries and assembly time in total. The final statement follows for $s=\mathcal{O}(1)$ and a constant number of generators.
\end{proof}

Two practical consequences are worth emphasizing. First, since $\mathcal{M}_{n}^{(k)}$ is defined analytically, one never needs to store the basis elements as dense matrices, but it suffices to store their labels (type and indices) and their sparse structure constants.
Second, \cref{th:HWstructureconstants-dk3} concerns the full structure constant tensor. If only a fixed set of sparse MGGM generators is required, the targeted construction of \cref{lem:targetedMGGMadjoint} avoids forming the full tensor and directly produces the adjoint matrices needed for $\gsim$. 
Finally, for $k=1$ one has $d_1=n$, so \cref{th:HWstructureconstants-dk3} implies an $\mathcal{O}(n^3)$ primitive for generating the full adjoint data of $\mathfrak{u}(n)$ in this basis, which is asymptotically cheaper than Pauli string preprocessing in quadratically-scaling Pauli Lie algebras (cf.\ \cref{ssec:FF_Paulialgebras}).

Appendix~\ref{appssec:BMstructureconstantsHW} presents numerical benchmarks for computing adjoint data in the MGGM basis for $k=1,2,3$, illustrating the practical performance of $\gsim$ in low-HW regimes.

\begin{corollary}[Efficient $\gsim$ for bounded-HW-preserving circuits]
\label{cor:efficientGSIMforHWcircuits}
Let $U(\boldsymbol{\theta})$ be a HW-preserving circuit acting on a fixed sector $\mathcal{H}_k$ with $k=\mathcal{O}(1)$, and suppose its dynamical generators are finite sums of two-qubit HW-preserving terms of the form~\eqref{eq:HWgenerator}. Then, after restriction to $\mathcal{H}_k$, the dynamics are contained in a subalgebra of $\mathfrak{u}(d_k)$ with $d_k=\tbinom{n}{k}=\mathcal{O}(\poly(n))$, and the $\gsim$ preprocessing and simulation primitives can be implemented in time polynomial in $n$ and in the circuit size, provided the input states and observables are specified efficiently in the MGGM representation. More explicitly, the MGGM basis~\eqref{eq:MGGMbasis} gives an explicit analytic basis for the full sector algebra $\mathfrak{u}(d_k)$, eliminating the need of computations in a dense Hilbert-space representation. The generators can be expanded using~\cref{prop:HWgeneratosinMGGMbasis}. In the most general case, all nonzero structure constants of $\mathfrak u(d_k)$ can be generated in $\mathcal{O}(d_k^3)=\mathcal{O}(n^{3k})$ time by \cref{th:HWstructureconstants-dk3}. Whenever the relevant restricted generators have sparse MGGM support, the targeted construction of \cref{lem:targetedMGGMadjoint} can instead be used to assemble only the adjoint matrices required by the circuit.
\end{corollary}

This corollary makes the Lie-algebraic simulability of bounded-HW circuits constructive. Previous discussions of RBS architectures note that, although fixed-HW sectors of constant weight are directly simulable, a more general Lie-algebraic treatment requires an explicit basis of the generated algebra and this basis may be costly to obtain~\cite{Monbroussou.2025-TrainabilityExpressivityHammingweight}. Here this potential obstruction is removed for bounded-HW sectors: the MGGM basis provides a fixed polynomial-size ambient algebra in which arbitrary two-qubit HW-preserving generator sets and connectivity patterns can be represented. Thus, for $k=\mathcal{O}(1)$ (and likewise for $n-k=\mathcal{O}(1)$), HW-preserving architectures are efficiently simulable via $\gsim$ on the corresponding sector. In fact, since the full subspace dimension $d_k=\binom{n}{k}$ is itself polynomial, this setting is stronger than generic expectation value $\gsim$: by working in the full subspace algebra $\mathfrak u(d_k)$, one can also recover arbitrary computational basis probabilities, and for pure inputs the sector amplitudes up to a global phase, in polynomial time, making $\gsim$ a strong simulator w.r.t.~\cref{ssec:notionSimulability} for computational-basis measurements restricted to the fixed sector.

\begin{example}[Quantum chemistry in a two electron active space]
Electronic-structure Hamiltonians conserve particle number. Under standard fermion-to-qubit encodings (e.g.\ Jordan--Wigner), the resulting qubit dynamics commute with the excitation-number operator~\eqref{eq:excitation_number} and hence preserve Hamming weight, i.e.\ they act block-diagonally on the fixed-weight subspaces $\mathcal H_k$.
A widely used chemistry-motivated VQE ansatz is the unitary coupled cluster (UCC)~\cite{Peruzzo.2014-VariationalEigenvalueSolver,Anand.2022-QCviewUnitaryCoupledClusterTheory}, with ${U=\exp(T-T^\dagger)}$ where each excitation term has equal numbers of creation and annihilation operators and therefore preserves particle number. In the common UCCSD truncation one retains only single and double excitations~\cite{Romero.2018-StrategiesforMolecularEnergiesUCCAnsatz}, so that a $k$-electron reference determinant (e.g.\ Hartree--Fock) remains supported entirely on $\mathcal H_k$.
For constant active electron number $k=O(1)$, the effective dynamics are confined to $\mathfrak u(d_k)$ with $d_k=\binom{n}{k}$ (see Proposition~\ref{prop:bounded_hw_dla}). In particular, for $k=2$ one has $d_2=\binom{n}{2}=O(n^2)$ and hence $\dim(\g)\le d_2^2=O(n^4)$. Thus, even though the circuit may be highly expressive, the Lie-algebraic simulation primitives remain polynomial in $n$ within $\mathcal H_2$. Typical examples include $\mathrm{H}_2$ (in an extended orbital basis), $\mathrm{HeH}^+$, or effective two-electron active-space models such as $\mathrm{LiH}$ with a frozen core.
\end{example}

\begin{example}[Quantum data encoders in fixed-HW subspaces]
Fixed-HW subspaces also arise in practical state-preparation and QML protocols~\cite{Cherrat.2023-QuantumDeepHedging,Farias.2025-QuantumEncoderFixedhammingweight}, where one deliberately restricts the dynamics to a sector $\mathcal{H}_k$ of fixed Hamming weight~\eqref{eq:k-weight-subspace-dimension}. In particular, Ref.~\cite{Farias.2025-QuantumEncoderFixedhammingweight} proposes a gate-optimal amplitude encoder acting on $\mathcal{H}_k$, which loads a real vector of length $d_k=\binom{n}{k}$ using exactly $d_k-1$ HW-preserving gates and thereby achieves polynomial space compression of degree $k$.
The elementary building blocks of this encoder are uncontrolled and controlled reconfigurable beam-splitter (RBS) gates~\cite{Gard.2020-EfficientSymmetrypreservingState}. On two qubits, an RBS gate is a planar rotation on the subspace $\mathrm{span}\{\ket{01},\ket{10}\}$, $\operatorname{RBS}(\theta)=e^{\,i\theta H_{\mathrm{RBS}}}$ with
\begin{equation}
H_{\mathrm{RBS},ij}=
\begin{pmatrix}
0&0&0&0\\
0&0&-i&0\\
0&i&0&0\\
0&0&0&0
\end{pmatrix}_{ij}
= -\,\mathbf{J}_{ij},
\end{equation}
so, up to our sign convention in~\cref{eq:HWgeneratorDECOMPOSED}, it is generated precisely by the HW-preserving term $\mathbf{J}_{ij}$. By~\cref{prop:HWgeneratosinMGGMbasis}, the restriction of such gates to a fixed sector $\mathcal{H}_k$ is represented entirely by antisymmetric MGGM basis elements $\Lambda_A^{ab}$~\eqref{eq:MGGM_offdiag}. Consequently, the real-valued encoder evolves inside the subalgebra $\mathfrak{so}(d_k)\subseteq \mathfrak{u}(d_k)$. The resulting dynamics therefore fit naturally into the MGGM framework developed above and are exactly simulable with $\gsim$; see also~\cref{ssec:HWencoder}.
\end{example}

\section{Numerical experiments}
\label{sec:NumericalExamples}
We now illustrate the practical use of the Lie-algebraic simulation framework for $\poly$-DLA systems developed above through a set of representative numerical examples. In each subsection, we revisit a concrete application from the literature whose numerical study was previously limited to relatively small system sizes by conventional state-vector simulation, and show how the same setup can be reproduced and extended to substantially larger sizes using $\gsim$ together with the algebra-specific primitives introduced in this work.

The goal of these experiments is mainly methodological: they are intended to demonstrate the computational feasibility of our theory across the three classes of systems considered in this paper. To this end, we include one representative example for each setting, namely a free-fermionic benchmark for the $1D$ transverse-field Ising model in Section~\ref{ssec:tfimopt}, a permutation-equivariant quantum neural network for graph classification in Section~\ref{ssec:peQNN}, and a fixed-Hamming-weight amplitude encoder for $q$-Gaussian state preparation in Section~\ref{ssec:HWencoder}. Rather than providing exhaustive application studies, these examples serve as proof-of-principle demonstrations that the Lie-algebraic formalism can be turned into concrete large-scale simulation routines. Broader perspectives on potentially impactful applications of $\gsim$ are discussed in \cref{sec:outlook}. The code for all experiments is included in the accompanying software package~\cite{github_repo}; see the Code Availability statement for details.

\subsection{A provably well-behaved VQA for TFIM ground states in a free-fermion DLA}
\label{ssec:tfimopt}
A useful aspect of the Lie-algebraic viewpoint is that it provides a constructive recipe for variational quantum problems whose optimization landscape is well-behaved in the noiseless setting. Concretely, both barren plateau behavior and spurious suboptimal minima can be avoided by choosing a setup in which (i) the input state $\rho_{\mathrm{in}}$, observable $O$, and ansatz generators are all supported on a polynomially growing DLA $\g$, and (ii) the circuit is sufficiently overparameterized, in the sense that the number of independent parameters is of order $\dim(\g)$ (cf.\ \cref{ssec:LieCharVQAs}). While this direction can be viewed as complementary to the (still open) conjectured implication “provable absence of barren plateaus $\Rightarrow$ efficient classical simulation” \cite{Cerezo.2025-DoesProvableAbsence}, here we use it in the opposite direction: we deliberately design a trainable instance and verify that $\gsim$ reproduces the expected favorable optimization behavior at system sizes far beyond state-vector simulation.

As a minimal nontrivial test case we consider the transverse-field Ising model (TFIM) on a path graph and a Hamiltonian variational ansatz (HVA)~\cite{Wecker.2015-ProgressPracticalQuantum,Wiersema.2020-ExploringEntanglementOptimization}. We use the freely-parameterized TFIM generator set $G_{\mathrm{TFIM}}^{\mathrm{free}}$~\eqref{eq:freeTFIMgenerators} which generates a quadratic-dimensional DLA with $\dim(\g)=n(2n-1)$ for open boundaries~\cite{Kazi.2025-AnalyzingQuantumApproximate} (cf.~\cref{ex:TFIM}).
Our ansatz consists of $L$ layers, each containing one independently parameterized gate for every generator in $G_{\mathrm{TFIM}}^{\mathrm{free}}$, i.e.
\begin{equation}
U_\ell(\boldsymbol{\theta}_\ell)
=
\prod_{i=1}^{2n-1} e^{-i\theta_{\ell,i} H_i},
\qquad
H_i\in G_{\mathrm{TFIM}}^{\mathrm{free}},
\label{eq:tfim_free_layer}
\end{equation}
so that the total parameter count is $N_{\mathrm{params}}=(2n-1)L$.
To enforce the overparameterized regime we choose $L=n$, yielding $N_{\mathrm{params}}=\dim(\g)$.

We minimize the energy expectation of the TFIM Hamiltonian (equivalently, we set $O=H_{\mathrm{TFIM}}$ in \cref{eq:expectationvalue}). As initial state we take $\rho_{\mathrm{in}}=(\ket{+}\!\bra{+})^{\otimes n}$. This choice is compatible with the Lie-algebraic simulation pipeline in the Pauli string basis: $\rho_{\mathrm{in}}$ has a sparse Pauli expansion supported only on tensor products of $I$ and $X$, so the corresponding input vector $\mathbf e^{(\mathrm{in})}$ can be assembled efficiently by evaluating \cref{eq:O_and_ein} over the chosen Pauli basis elements $B_\alpha\in i\g$ (cf.\ \cref{ssec:gsimTheory,ssec:FF_Paulialgebras}).

We simulate this optimization problem with $\gsim$ up to $n=200$ and perform $100$ independent optimization runs per system size, using L-BFGS-B~\cite{Byrd.1995-LimitedMemoryAlgorithm} from uniformly random initial parameters. The exact TFIM ground-state energy $E_0$ is computed via standard free-fermion diagonalization (Jordan--Wigner / Bogoliubov transformation), which is efficient in this regime and serves only as a reference. The results in \cref{fig:tfimopt} show that the optimized energies concentrate tightly around $E_0$ across all tested system sizes, consistent with the expectation that the overparameterized, DLA-confined setup eliminates spurious minima in the noiseless setting and yields a well-behaved optimization task.

\begin{figure}[htbp]
\includegraphics[width=1\linewidth]{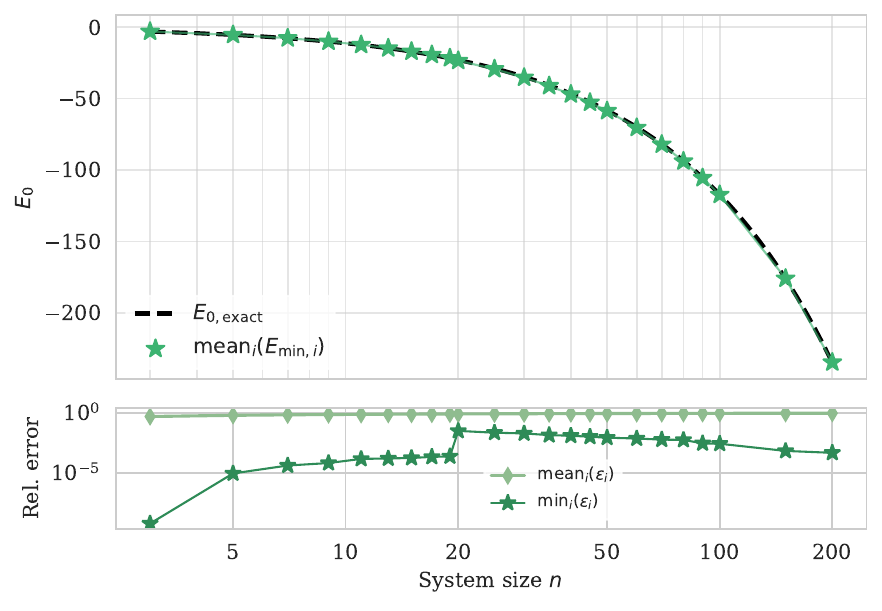}
\caption{\label{fig:tfimopt}Overparameterized TFIM-HVA optimization in a quadratic (free-fermion) DLA, simulated with $\gsim$. Top: exact TFIM ground state energy $E_{0,\mathrm{exact}}$ (dashed) versus the mean optimized energy $\mathrm{mean}(E_{\min})$ over $100$ independent L-BFGS-B runs (stars) for each system size $n$. Bottom: final relative energy errors $\epsilon_i=\epsilon(E_{\min}^{(i)},E_{0,\mathrm{exact}})$ across $i\in[100]$ runs, shown as the mean over runs, $\operatorname{mean}_i(\epsilon_i)$, and the best-run error, $\operatorname{min}_i(\epsilon_i)$. Each run uses the freely parameterized generator set $G_{\mathrm{TFIM}}^{\mathrm{free}}$~\eqref{eq:freeTFIMgenerators} with $L{=}n$ layers, giving $N_{\mathrm{params}}{=}\dim(\g)$. Exact reference energies are obtained via standard free-fermion diagonalization. The standard deviation of $E_{\min}$ across runs is smaller than the marker size over the plotted range.}
\end{figure}

A practically relevant caveat is that this benign behavior is fragile under stochastic perturbations of the objective, as encountered with finite-shot estimation~\cite{Scriva.2024-ChallengesVariationalQuantum,Barligea.2025-ScalabilityChallengesVariational} or device noise~\cite{Wang.2021-NoiseinducedBarrenPlateaus}. Following the methodology of \cite{Barligea.2025-ScalabilityChallengesVariational}, one can add i.i.d.\ Gaussian noise to each queried loss value and gradient and study the resulting degradation in the achieved relative energy error as a function of $n$ and the noise scale (see \cref{fig:tfimnoise}). In our experiments, the optimization becomes unreliable already at moderate system sizes once the noise level exceeds a small threshold, despite the underlying noiseless landscape being well-behaved. This highlights that provable trainability in the idealized setting does not by itself guarantee robustness of practical VQA training.

\begin{figure}[htbp]
\includegraphics[width=1\linewidth]{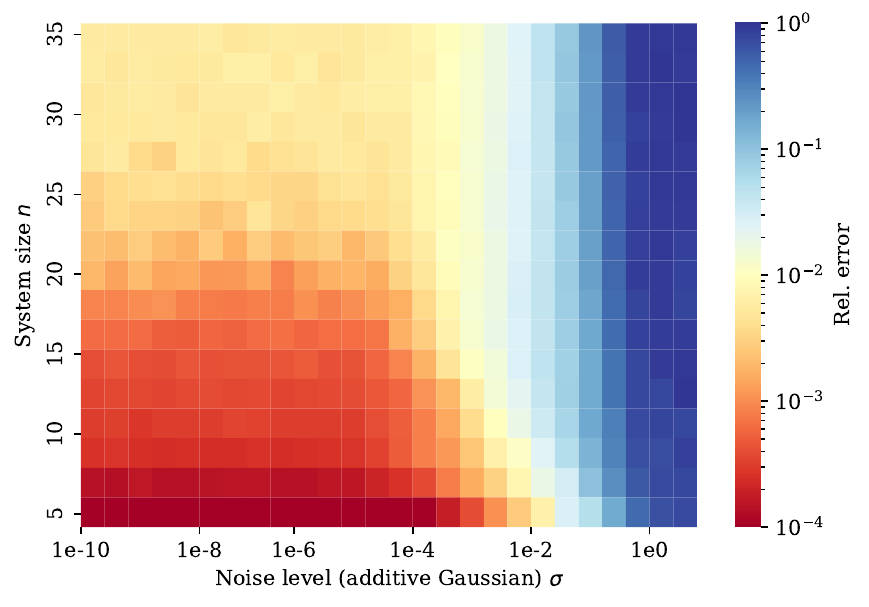}
\caption{\label{fig:tfimnoise}Noise sensitivity of overparameterized TFIM-HVA training under additive Gaussian perturbations. For each system size $n\le 35$ and noise level $\sigma$ (horizontal axis, logarithmic decades), we repeat the optimization experiment of \cref{fig:tfimopt} but add i.i.d.\ Gaussian noise with standard deviation $\sigma$ to every queried loss value (and the corresponding gradient evaluation) following the protocol of Ref.~\cite{Barligea.2025-ScalabilityChallengesVariational}. Each pixel reports the mean final relative energy error (color scale, logarithmic) with respect to the exact TFIM ground-state energy, averaged over $100$ independent optimization runs. The sharp crossover from low error to order-one error as $\sigma$ increases illustrates that even provably well-behaved noiseless landscapes can become practically untrainable under small stochastic perturbations, with the tolerable noise level decreasing as $n$ grows.}
\end{figure}

\subsection{Permutation-equivariant QNN for binary graph classification}
\label{ssec:peQNN}
As a representative application of permutation-equivariant dynamics, we consider the graphs-state classification task studied in Ref.~\cite{Schatzki.2024-TheoreticalGuaranteesPermutationequivariant}, where an $S_n$-equivariant QNN (cf.~\cref{ex:SN1}) is used to distinguish connected from disconnected graph states. Closely related permutation-equivariant architectures have also been explored for the graph isomorphism problem~\cite{Larocca.2022-GroupInvariantQML} and quantum phase classification~\cite{Nguyen.2024-TheoryEquivariantQuantum}.

Given an undirected graph $G=(V,E)$ on $n=|V|$ vertices, its graph state is defined as
\begin{equation}
\ket{G}
=
\Bigl(\prod_{(i,j)\in E} \mathrm{C}Z_{ij}\Bigr)\ket{+}^{\otimes n},
\label{eq:graph_state_def}
\end{equation}
where $\mathrm{C}Z_{ij}$ is the controlled-$Z$ gate. In the full supervised learning problem one considers a dataset
$\{(\rho_i,y_i)\}_{i=1}^M$ with $\rho_i=\ket{G_i}\!\bra{G_i}$ and labels $y_i\in\{+1,-1\}$
indicating whether $G_i$ is connected ($+1$) or disconnected ($-1$).
Since relabeling graph vertices corresponds exactly to permuting qubits, the encoding~\eqref{eq:graph_state_def} is permutation-covariant and therefore naturally matched to an $S_n$-equivariant model as discussed in~\cref{sec:Pauliorbitalgebra}.

Following Ref.~\cite{Schatzki.2024-TheoreticalGuaranteesPermutationequivariant}, we use a permutation-equivariant QNN whose layers are generated by three permutation-invariant Hamiltonians. In our Pauli orbit notation these correspond (up to fixed normalizations) to the orbit elements $B_{1,0,0}$, $B_{0,1,0}$, and $B_{0,0,2}$ (cf.~\cref{eq:peQNNgenerators}), i.e.\ collective $X$, collective $Y$, and collective $ZZ$ interactions.
Concretely, we take each layer to be
\begin{equation}
U_\ell(\alpha_\ell,\beta_\ell,\gamma_\ell)
=
e^{-\alpha_\ell B_{1,0,0}}\,
e^{-\beta_\ell B_{0,1,0}}\,
e^{-\gamma_\ell B_{0,0,2}},
\label{eq:peQNN_layer_orbit}
\end{equation}
and the full circuit is $U(\boldsymbol{\theta}){=}\prod_{\ell=1}^L U_\ell(\alpha_\ell,\beta_\ell,\gamma_\ell)$, with $\boldsymbol{\theta}=(\boldsymbol{\alpha},\boldsymbol{\beta},\boldsymbol{\gamma})$.
As measurement we take the $S_n$-invariant two-body $XX$ observable
\begin{equation}
    O = \frac{2}{n(n-1)}\sum_{1\leq i<j\leq n} X_i X_j,
\end{equation}
which is proportional to the Pauli orbit $B_{2,0,0}$. The corresponding linear classification objective is
\begin{equation}
    \mathcal{L}(\boldsymbol{\theta})=\frac{1}{M}\sum_{i=1}^My_i\,\tr[O\,U(\boldsymbol{\theta})\,\rho_i\,U^\dagger(\boldsymbol{\theta})].
\label{eq:peqnnloss}
\end{equation}

In $\gsim$, the relevant input of the adjoint-space propagation is the vector $\mathbf e^{(\mathrm{in})}$~\eqref{eq:O_and_ein}, expressed in the chosen Lie algebra basis. For graph states, these overlaps can be obtained from their stabilizer description (cf.\ Ref.~\cite{Schatzki.2024-TheoreticalGuaranteesPermutationequivariant}). While one can in principle handle general graph states efficiently by approximating their expansion in the Pauli orbit basis~\eqref{eq:Pauliorbitbasis}, we keep the present numerical demonstration deliberately simple: we restrict to disconnected graphs whose connected components have size at most $10$. In this regime, the orbit-level decomposition is straightforward to compute exactly by enumerating stabilizer elements component-wise and combining components, yielding an efficient construction of $\mathbf e^{(\mathrm{in})}$ in the Pauli orbit basis $\mathcal{B}_n$.
To place the ansatz in the overparameterized regime discussed in~\cref{ssec:gsimTheory}, we choose the depth such that the parameter count matches the algebra dimension, i.e., 
\begin{equation}
N_{\mathrm{params}}=3L=\dim(\mathfrak g)=\binom{n+3}{3}-1.
\end{equation}

\cref{fig:peqnnvar} compares the empirical variance of the permutation-equivariant QNN loss~\eqref{eq:peqnnloss} with the large-depth prediction~\eqref{eq:varianceformula}, both evaluated with $\gsim$. For each system size, the empirical variance is estimated from loss evaluations in $100$ uniformly sampled parameter settings. In the practically relevant overparameterized regime ($L=L_\mathrm{ovp}$), the variance decays polynomially with system size and is well described by the fitted law, $\operatorname{Var}_{\boldsymbol{\theta}}[\ell(\boldsymbol{\theta})]\sim n^{-2.53\pm0.11}$. This demonstrates the expected absence of exponential variance suppression in this setup, consistent with the trainability behavior reported in Ref.~\cite{Schatzki.2024-TheoreticalGuaranteesPermutationequivariant}, while extending the simulation far beyond the range of direct state-vector numerics. In our experiments, the variance is evaluated up to $n=80$, where the underlying permutation-invariant algebra already has dimension $\binom{80+3}{3}-1{=}91880$, showing that Lie-algebraic simulation remains practical even in large symmetry sectors. 

At the same time, panel~(b) of~\cref{fig:peqnnvar} shows that the empirical overparameterized-depth variance lies noticeably above the asymptotic large-depth prediction~\eqref{eq:varianceformula}, whose fitted decay is substantially steeper, $\operatorname{Var}_{\boldsymbol{\theta}}[\mathcal{L}(\boldsymbol{\theta})]
\sim n^{-5.68\pm0.10}$. To confirm that this discrepancy is a finite-depth effect rather than a failure of the Lie-algebraic prediction, we repeated the experiment in an artificially deep regime. As shown in panel~(a), the empirical variance then closely follows the asymptotic curve. This validates Eq.~\eqref{eq:varianceformula} as the appropriate large-depth prediction for variational quantum loss functions such as~\cref{eq:peqnnloss}, while simultaneously illustrating that $\gsim$ can resolve the pre-asymptotic finite-depth regime at system sizes inaccessible to conventional state-vector simulation.

\begin{figure}[htbp]
\includegraphics[width=1\linewidth]{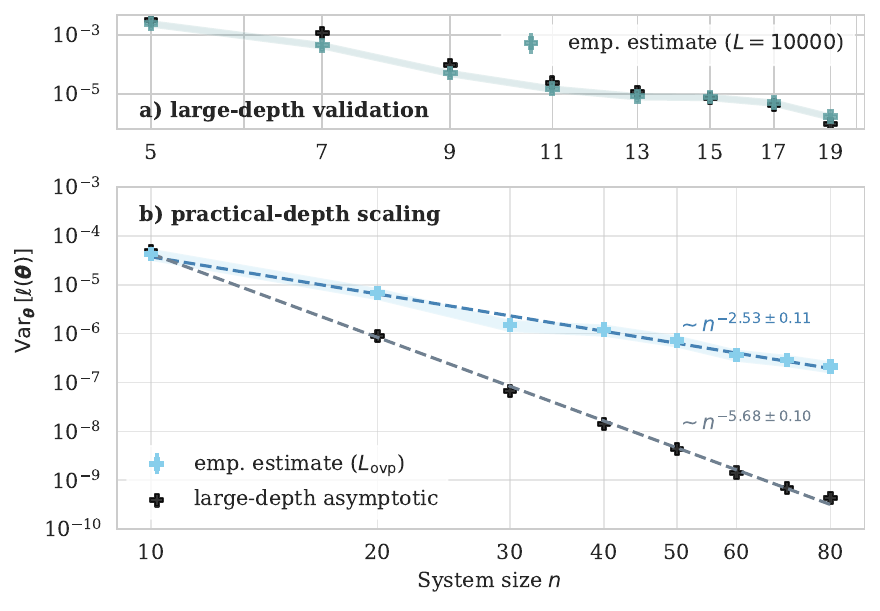}
\caption{\label{fig:peqnnvar}Variance scaling of the permutation-equivariant QNN loss~\eqref{eq:peqnnloss} for the graph-state classification task with $M=10$ disconnected training instances. For each system size $n$, the empirical variance is estimated from $100$ uniformly sampled parameter settings using $\gsim$. Panel~(a) validates the large-depth regime: when the circuit depth is taken sufficiently large, the empirical variance agrees closely with the asymptotic Lie-algebraic prediction from Eq.~\eqref{eq:varianceformula}. Panel~(b) shows the practically relevant overparameterized-depth regime $(L=L_\mathrm{ovp})$, where the variance still decays polynomially with system size, but more slowly than the large-depth asymptotic curve. Dashed lines indicate power-law fits, yielding exponents $-2.53\pm0.11$ for the practical-depth data and $-5.68\pm0.10$ in the large-depth regime. The figure illustrates both the absence of exponential variance suppression and the role of $\gsim$ in analyzing $S_n$-symmetric systems at large system sizes.}
\end{figure}

\subsection{Fixed-HW amplitude encoding for $q$-Gaussian state preparation}
\label{ssec:HWencoder}
Fixed-Hamming-weight-preserving circuits arise naturally not only in particle-number-conserving variational ansätze, but also in exact state-preparation protocols. As a concrete benchmark for our $\gsim$ framework for applied HW-preserving state evolutions, we consider the fixed-HW amplitude encoder of Ref.~\cite{Farias.2025-QuantumEncoderFixedhammingweight}. For a real amplitude vector $\mathbf a\in\mathbb R^d$, this construction prepares the corresponding state in the $k$-particle subspace $\mathcal H_k$~\eqref{eq:k-weight-subspace-dimension}, which yields polynomial space compression for fixed $k$ (i.e., $d=\mathcal{O}(n^k)$), in contrast to binary amplitude encoding, which achieves exponential compression at the price of exponentially growing gate complexity~\cite{Plesch.2011-QuantumstatePreparationUniversalGate}. In the real-valued case, the encoder of Ref.~\cite{Farias.2025-QuantumEncoderFixedhammingweight} is parameter-optimal: it uses exactly $d-1$ uncontrolled and controlled reconfigurable beam-splitter (RBS) gates~\cite{Gard.2020-EfficientSymmetrypreservingState,Tilly.2022-VariationalQuantumEigensolver} to prepare an arbitrary state in $\mathcal H_k$.

The crucial observation for our purposes is that every gate in this construction preserves Hamming weight. In the computational basis of $\mathcal H_k$, each RBS gate acts as a two-level rotation between two basis states whose Hamming distance is \(2\). Denoting these basis states by \(\ket{b_a}\) and \(\ket{b_b}\), the corresponding generator is
\begin{equation}
    G_{ab}=i\bigl(\ket{b_b}\!\bra{b_a}-\ket{b_a}\!\bra{b_b}\bigr),
\end{equation}
which is precisely an antisymmetric MGGM basis element $\Lambda_A^{ab}$ (cf.~\cref{def:MGGMbasisforHWalgebras}); in the physical qubit picture this is equivalent to the usual RBS generator $H_{\mathrm{RBS}}\propto X_iY_j-Y_iX_j$, up to our sign convention for $\mathbf J_{ij}$~\eqref{eq:HWgeneratorDECOMPOSED}. Hence, the full real-valued encoder evolves entirely inside the $\mathfrak{so}(d)$ subalgebra generated by the antisymmetric sector of the MGGM basis, and is therefore exactly amenable to simulation by $\gsim$.

Unlike variational HW-preserving encoders~\cite{Yan.2024-UniversalHammingWeight,Monbroussou.2025-TrainabilityExpressivityHammingweight}, the parameters here are fixed deterministically. If $\mathbf a=(a_1,\dots,a_d)$ denotes the target amplitude vector ordered according to the Gray-code sequence used in Ref.~\cite{Farias.2025-QuantumEncoderFixedhammingweight}, the rotation angles are given by the hyperspherical coordinates
\begin{equation}
\begin{aligned}
\theta_j & :=\operatorname{atan2} \left(\sqrt{\sum_{j^{\prime}=j+1}^d a_{j^{\prime}}^2},\ a_j\right), \quad j \in[d-2] \\
\theta_{d-1} & :=\operatorname{atan2} \left(a_d, a_{d-1}\right),
\end{aligned}
\end{equation}
where $\operatorname{atan2}(y,x)$ denotes the two-argument arctangent. In particular, the algorithm orders the HW-$k$ basis states such that consecutive bitstrings differ by Hamming distance $2$, so that every step adds exactly one new basis state by a single (possibly controlled) RBS rotation.

We apply this construction to the $q$-Gaussian loading task discussed in Ref.~\cite{Farias.2025-QuantumEncoderFixedhammingweight}. Choosing $q=\tfrac32$ and $\beta=2$, the target probability density is proportional to
\begin{equation}
    p(x)=(1+x^2)^{-2}.
\end{equation}
For a fixed $n$ and $k=2$, we discretize this function on $d=\binom{n}{2}$ equidistant grid points $x_j\in[-2,2]$, normalize the resulting probabilities, and encode the nonnegative amplitudes $a_j=\sqrt{p(x_j)}$. In our MGGM description, the input state of the protocol is simply the first computational basis state in the chosen Gray-code ordering. Consequently, the initial adjoint-space vector $\mathbf e^{(\mathrm{in})}$~\eqref{eq:O_and_ein} is supported on a single diagonal MGGM basis element $\Lambda_P^a$. After the circuit evolution, the encoded probabilities are read off directly from the diagonal sector of $\mathbf e^{(\mathrm{out})}$~\eqref{eq:eout_adjoint}, i.e. from the expectation values of the projectors $\Lambda_P^a$ onto the ordered HW-$2$ basis states.

While Ref.~\cite{Farias.2025-QuantumEncoderFixedhammingweight} validated the encoder on quantum hardware for the proof-of-principle instance $n=6$, we use $\gsim$ to verify the same fixed-HW protocol to substantially larger symmetry sectors. Figure~\ref{fig:hwencoderplot} shows the case $n=50$, $k=2$, for which the state contains $d=\binom{50}{2}=1225$ encoded amplitudes. The resulting probabilities match the target $q$-Gaussian to machine precision.

\begin{figure}[htbp]
\includegraphics[width=1\linewidth]{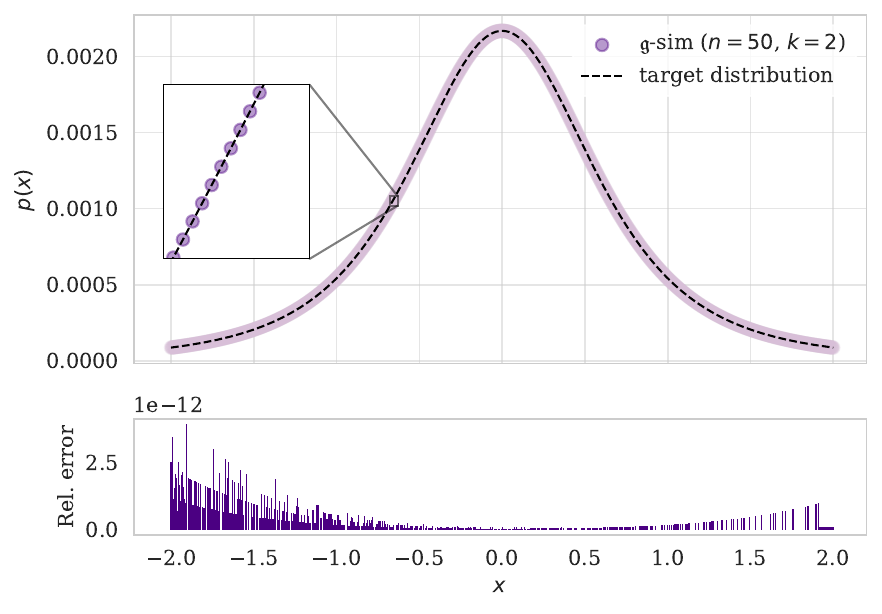}
\caption{\label{fig:hwencoderplot}Large-scale fixed-HW amplitude encoding of a $q$-Gaussian distribution with $\gsim$. The circuit of Ref.~\cite{Farias.2025-QuantumEncoderFixedhammingweight} is simulated in the $k=2$ subspace for $n=50$, corresponding to $d=\binom{50}{2}=1225$ encoded amplitudes. Top: probabilities extracted from the diagonal MGGM sector after the HW-preserving encoding circuit, compared with the target discretized $q$-Gaussian $p(x)\propto(1+x^2)^{-2}$ on $x\in[-2,2]$. The inset resolves a small portion of the curve and highlights that the smooth profile is supported on a large number of discrete points. Bottom: pointwise relative error between the $\gsim$ output and the target distribution, which remains at the level of numerical precision.}
\end{figure}

Although the angles in this state-preparation example are known analytically, the same HW-preserving architecture can also be interpreted as a variational ansatz by treating the RBS angles as trainable parameters. This viewpoint becomes relevant when the target state is specified only implicitly, for example through empirical samples from a constrained distribution or as the ground state of a particle-number-preserving Hamiltonian. In such cases, one may optimize a task-specific loss while the dynamics remain confined to the fixed-HW sector, and the same MGGM adjoint representation can be used within $\gsim$ to evaluate both the loss and its gradients efficiently.

Overall, this example illustrates that $\gsim$ applies beyond the standard setting of expectation value evaluation for variational circuits. It can also simulate structured state-preparation protocols, read out full probability distributions within symmetry sectors, and provide gradients when the same architectures are used variationally. In the deterministic loading task considered here, increasing $n$ at fixed $k=2$ directly refines the discretization of the target distribution, allowing one to study high-resolution amplitude encodings in regimes that are inaccessible to direct $2^n$-dimensional state-vector simulation.

\section{Conclusions}
\label{sec:conclusion}
As classical simulation becomes an increasingly critical tool across the quantum computing stack, it is crucial to understand which forms of structure make quantum circuits efficiently simulable in practice and to develop broadly applicable, efficient, and practical simulation frameworks for them. This work addressed this question for Lie-algebraic classical simulation. While previous practical demonstrations of $\gsim$ have centered on free-fermionic (matchgate) settings, we showed that Lie-algebraic simulability extends beyond this regime and offers a meaningfully broader paradigm. The central lesson is representation-theoretic: a polynomial-dimensional dynamical Lie algebra does not, by itself, provide an efficient simulation algorithm, but a structural promise that becomes practical only once the relevant operator dynamics are expressed in coordinates where the required simulation primitives are efficient. 

We developed such representations for three representative sources of polynomial-dimensional Lie-algebraic structure. For free-fermionic Pauli algebras and translation-invariant free-fermion equivalents, we streamlined the Pauli string description and introduced ``Pauli cycle'' representations that exploit cyclic structure to reduce preprocessing costs. For permutation-equivariant dynamics, we proposed the ``Pauli orbit'' basis, which compresses exponentially large Pauli expansions into a cubic-dimensional ambient algebra with efficient commutator and inner-product primitives. For Hamming-weight preserving dynamics on bounded fixed-weight sectors, we introduced a modified generalized Gell--Mann basis representation with closed-form commutation rules. Together with runtime benchmarks and large-scale proof-of-concept simulations, enabled by an open-source implementation~\cite{github_repo}, these results show that the proposed preprocessing primitives are not only polynomial in principle, but practical at system sizes far beyond direct state-vector simulation. 

Taken together, the work reframes $\gsim$ from a reformulation of known free-fermionic simulability into a representation-theoretic program for uncovering new efficient classical simulations. The relevant question is not only whether $\dim(\g)=\poly(n)$, nor whether a circuit belongs to a familiar tractable class, but whether its Lie-algebraic structure can be exposed in coordinates where basis construction, structure constant computation, and adjoint-space propagation are efficient. This motivates the following working conjecture.

\begin{conjecture}[Efficient representations for polynomial-dimensional DLAs]
\label{conj:gsim_from_poly_dim}
Consider a structured parametrized quantum circuit $U(\boldsymbol{\theta})$ of the form~\cref{eq:standardansatz}, and let $\g$ be its dynamical Lie algebra (DLA) as defined in~\cref{eq:DLAdefinition}, with dimension $d{=}\dim(\g)$. If $d{=}\poly(n)$, then there exists an efficiently encodable basis $\mathcal B{=}\{B_\alpha\}_{\alpha=1}^d$ of $\g$ such that the algebraic primitives required by $\gsim$ can be evaluated in time $\poly(d, n)$. More precisely, for any basis elements $B_\alpha,B_\beta\in\mathcal B$, one can compute the commutator expansion ${[B_\alpha,B_\beta]{=}\sum_{\gamma=1}^d f_{\alpha\beta}^{\gamma} B_\gamma}$ and the Hilbert--Schmidt inner products $\langle B_\alpha,B_\beta\rangle$ in $\poly(d,n)$ arithmetic time. Equivalently, the structure constants~\eqref{eq:structureconstants} defining the adjoint representation of the generators of \(U(\boldsymbol{\theta})\) are efficiently computable in this representation.
Consequently, whenever the input state and observables admit efficient coordinate descriptions in the corresponding Hermitian operator space \(i\g\) (up to the identity component when necessary), $\gsim$ can evaluate exact expectation values~\eqref{eq:expectationvalue}, correlators~\eqref{eq:correlators}, and their parameter gradients in time polynomial in \(n\), \(d\), and the circuit size.
\end{conjecture}

The results of this paper provide three concrete pieces of evidence: free-fermionic structure, collective $S_n$ symmetry, and $U(1)$ symmetry together with bounded fixed-sector restriction. In each case, polynomial-dimensional algebraic closure is accompanied by a natural representation in which the apparent exponential preprocessing barrier disappears. These examples show that polynomial Lie-algebraic structure, when described in the right representation, can serve as an actionable resource for efficient classical simulation beyond free fermions.

\subsection{Outlook}
\label{sec:outlook}
The methods developed here open several interesting directions for future work. 
A first and immediate opportunity is to revisit existing applications of $\gsim$ with the enlarged class of efficient representations now available. These include variational warm-starting and training acceleration~\cite{Goh.2025-LiealgebraicClassicalSimulations,bhowmick2025enhancing,zering2025benchmarking}, investigating privacy attacks in quantum machine learning~\cite{Heredge.2025-CharacterizingPrivacyQuantum}, quantum metrology~\cite{Lecamwasam.2024-QuantumMetrologyLinear}, and circuit synthesis~\cite{Goh.2025-LiealgebraicClassicalSimulations,goh2024protocols}. 
More broadly, the extensions of $\gsim$ developed in this work may open new applications across the quantum computing stack, especially in settings where other classical simulation paradigms have already played an enabling role~\cite{knill2008randomized,helsen2022matchgate,arute2019quantum,tang2019quantum,gilyen2018quantum,chia2018quantum,chen2023quantum,Larocca.2022-GroupInvariantQML,Larocca.2022-DiagnosingBarrenPlateaus,Larocca.2023-TheoryOverparametrizationQuantum,Larocca.2025-BarrenPlateausVariational,white1992density,vidal2003efficient,schollwock2005density,schollwock2011density,lieb1961two,baxter1985exactly,kivlichan2018quantum,google2020hartree,fomichev2024initial,ran2020encoding,holmes2020efficient,gonzalez2024efficient,rudolph2024decomposition,berry2025rapid,huggins2025efficient,rupprecht2026sparse,Huang.2020-PredictingManyProperties,elben2023randomized,zhao2021fermionic,Wan.2023-MatchgateShadowsFermionic,yen2021cartan,aaronson2004improved,kissinger2020reducing,gibbs2025deep,strikis2021learning,czarnik2021error,arute2020observation,montanaro2021error,Gottesman.1997-Stabilizer,fowler2012surface,ferris2014tensor,farrelly2021tensor,gidney2021stim}. In particular, it would be interesting to ask where $\gsim$ can strengthen, generalize, or scale existing simulation-based tools for hardware verification, benchmarking, state preparation, and measurement. 

A second direction is to use $\gsim$ as a diagnostic tool for (variational) quantum algorithms. The ability to compute exact expectation values and gradients at large system sizes enables probing of scaling phenomena that are difficult to access with state-vector simulation or noisy quantum hardware. This includes testing whether proposed barren-plateau-free architectures induce polynomial-dimensional DLAs, whether they reduce to known solvable structures, and how robust their optimization landscapes are to finite-shot noise or other stochastic perturbations. Lie closure can also reveal effective invariants or conserved quantities directly from the generator set, without first requiring a bespoke mapping such as the frustration-graph constructions used to certify free-fermion solvability~\cite{Chapman.2023-UnifiedGraphtheoreticFramework}. In this sense, polynomial-dimensional Lie-algebraic structure can serve not only as an explanation for known simulable regimes, but also as a practical guide for discovering new ones. This perspective is further supported by the growing use of Lie-algebraic methods to analyze quantum systems beyond the settings discussed here~\cite{Galitski.2011-QuantumtoclassicalCorrespondence,yen2021cartan,Patel.2024-ExtensionExactlysolvableHamiltonians,Burkat.2025-QuantumPaldusTransform,Mhiri.2026-BosonSamplingDilute,Lastres.2026-NonuniversalityConservedSuperoperators}. Applying the representation-adapted viewpoint developed in this work to such systems may provide new routes toward practical Lie-algebraic classical simulation algorithms. 

Beyond dequantization statements and diagnostics, $\gsim$ may also be useful as a classical state preparation primitive. In settings where a structured ansatz remains inside a polynomial-dimensional DLA, one can optimize expectation values exactly and without sampling noise, obtain reliable parameters classically, and then use the resulting circuit as an input state-preparation routine on a quantum device. Such states could serve as warm starts for more expressive variational procedures or as initial states for downstream algorithms such as quantum phase estimation. A central question for future work is how useful states prepared inside polynomial-dimensional Lie-algebraic models remain when used as approximations or starting points for larger, non-simulable target systems, for example, in quantum chemistry, condensed-matter models, or quantum machine learning.

On the theoretical side, a natural next step is to follow the thread of \cref{conj:gsim_from_poly_dim} and ask how broadly symmetry-adapted basis constructions extend. Beyond $S_n$-equivariance and bounded Hamming-weight sectors, one may consider translation and dihedral symmetries, stabilizer-like constraints, non-Abelian symmetries such as total-spin conservation, or combinations of locality with global constraints. A related problem is to systematically characterize polynomially growing DLAs arising from mixed generator types, including single Pauli strings, structured sums, and collective terms, and to determine when their closures escape the known polynomially-scaling regimes. More generally, it would be valuable to formulate Lie-algebraic simulation in terms of faithful representations beyond the adjoint representation whenever this yields clearer connections to specialized simulation methods or sharper complexity bounds.

Finally, an important longer-term direction is to develop and evaluate approximate Lie-algebraic simulation methods. Many physically relevant systems may not have polynomial-dimensional DLAs exactly, but may contain large polynomially simulable sectors weakly coupled to non-simulable directions. Truncation schemes, perturbative treatments around polynomial-dimensional subalgebras, or hybrid methods combining $\gsim$ with other simulation paradigms could substantially broaden the practical reach of the framework. The bounded Hamming-weight setting already suggests such a perspective: while the full symmetry-preserving algebra is generally exponential, fixed or low-excitation sectors remain efficiently tractable.

Overall, our work serves as evidence that $\gsim$ is not merely a reformulation of free-fermion simulation, but a flexible framework with substantially broader applicability. We sincerely hope that the tools introduced in this work will provide a strong foundation for the community for further creative applications of Lie-algebraic simulation across quantum computing and many-body physics. 

\section*{Code availability}
The code for implementing all basis representations, preprocessing primitives, and numerical examples reported in this work is publicly available at~\cite{github_repo}: \url{https://github.com/adelina-b/g-sim}.

\section*{Acknowledgments}
We thank the participants of QIP 2026 for many helpful discussions during the poster session, where a preliminary version of this work was presented. In particular, we thank Dylan Herman, Jens Eisert, and Renato M.\ S.\ Farias for comments and conversations that influenced several aspects of this paper. 
We are especially grateful to Lionell J.\ Dmello for guidance and discussions on representation-theoretic aspects in the early stages of this project, and to Albert H.\ Werner and Frank Pollmann for constructive feedback on the overall direction of the work. 
All numerical experiments and runtime benchmarks were carried out on local devices and the LiCCA HPC cluster of the University of Augsburg, co-funded by the Deutsche Forschungsgemeinschaft (DFG, German Research Foundation) -- Project-ID 499211671.
This research is part of the Munich Quantum Valley, which is supported by the Bavarian state government with funds from the Hightech Agenda Bavaria.

\bibliography{literature}

\onecolumngrid
\appendix

\section{Supplementary information on the complexity of $\gsim$}
\label{appsec:gsimcomplexity}

This appendix provides additional discussion and justification for the schematic complexity bounds summarized in \cref{tab:gsimcomplexities} for the $\gsim$ pipeline described in \cref{ssec:gsimTheory}. Throughout, let $\g$ be the DLA (or, more generally, the chosen invariant operator subspace) and denote $d:=\dim(\g)$.

\paragraph*{Input vector and observable decomposition.}
The adjoint-space formulas in \cref{ssec:gsimTheory} assume that the observable $O$ and the input state $\rho_{\mathrm{in}}$ are provided through their coordinates in the chosen basis $\{B_\alpha\}_{\alpha=1}^{d}$ of $i\g$; cf.\ \eqref{eq:O_and_ein}. In many applications (including all numerical examples in this paper), $O$ and $\rho_{\mathrm{in}}$ are chosen to lie directly in the algebraic subspace of interest and are specified natively in the same representation used for preprocessing, so coordinate extraction is either trivial or can be performed with the same inner-product primitive used for structure constant computation. In general, if $O$ or $\rho_{\mathrm{in}}$ are not supported on $i\g$, one must instead work in a larger invariant subspace as described in~\cite{Goh.2025-LiealgebraicClassicalSimulations}.

\paragraph*{Preprocessing: Lie closure / basis construction.}
Given a generator set $G=\{H_k\}_{k=1}^K$, Lie closure refers to constructing a basis of $\g=\langle iG\rangle_{\mathrm{Lie}}$ by repeatedly forming commutators of newly discovered basis elements and checking whether the result increases the span. In the worst case, a basis of size $d$ requires considering $\mathcal{O}(d^2)$ candidate commutators. We denote by $t_{\mathrm{com}}(n)$ the cost of evaluating a commutator primitive in the chosen representation (typically dependent on the system size $n$), and by $t_{\mathrm{li}}(d)$ the cost of testing whether a candidate commutator is linearly independent of the current basis (e.g.\ via a hash lookup for Pauli bases, or via incremental orthogonalization in a general basis). This yields the schematic preprocessing cost $\mathcal{O}\!\left(d^2\, (t_{\mathrm{com}}(n)+ t_{\mathrm{li}}(d))\right)$ matching \cref{tab:gsimcomplexities}.

\paragraph*{Preprocessing: structure constants / adjoint data.}
Once a basis $\{B_\alpha\}_{\alpha=1}^{d}$ is fixed, one may compute the adjoint generators $\Phi^{\ad}(H_k)$~\eqref{eq:structureconstants} by evaluating commutators and expressing them in the chosen basis. Although structure constants form a formal $d\times d\times d$ tensor, in practice one only needs the \emph{nonzero} coefficients in expansions of commutators
\begin{equation}
i[B_\alpha,B_\beta]=\sum_{\gamma} f_{\alpha\beta}^{\gamma}\,B_\gamma,
\end{equation}
and these can be generated by looping over the $\mathcal{O}(d^2)$ pairs $(\alpha,\beta)$, computing the commutator, and recording only the nonzero terms in its basis expansion. The same primitive $t_{\mathrm{com}}(n)$ governs this step, giving the second preprocessing line of \cref{tab:gsimcomplexities}. Moreover, since Lie closure already requires evaluating commutators, one can accumulate the nonzero structure constant data during the basis construction, reducing constant factors in practice.

\paragraph*{Main simulation: matrix exponentials in adjoint space.}
During simulation, each gate $e^{-i\theta_{l,k}H_k}$ acts on the coefficient vector via the adjoint exponential $\exp(\theta_{l,k}\Phi^{\ad}(H_k))$ (cf.\ \eqref{eq:adjointidentity}). A naive dense matrix exponential costs $\mathcal{O}(d^3)$, but one may exploit that the adjoint generators are fixed once preprocessing is complete: one can precompute an eigendecomposition (or another suitable factorization) of $\Phi^{\ad}(H_k)$ once per generator and then apply the exponential to vectors in $\mathcal{O}(d^2)$ time per gate (dominated by matrix--vector products), as discussed in~\cite{Goh.2025-LiealgebraicClassicalSimulations}. This yields an overall $\mathcal{O}(K d^2)$ preprocessing/evaluation cost for preparing the per-generator exponential action, consistent with \cref{tab:gsimcomplexities}. In many Pauli-basis settings, the adjoint generators are additionally extremely sparse, further improving constant factors.

\paragraph*{Main simulation: forward propagation and gradients.}
Given the per-gate adjoint actions, computing $\mathbf{e}^{(\mathrm{out})}$ in \eqref{eq:eout_product} requires $LK$ adjoint-space updates. In the worst case, each update is a dense $d\times d$ matrix applied to a $d$-vector, costing $\mathcal{O}(d^2)$, so total propagation cost is $\mathcal{O}(LK\,d^2)$. The final dot product with the observable coefficients $\boldsymbol{w}$ adds $\mathcal{O}(d)$ and is typically negligible at scale.

Gradients with respect to all $M=LK$ parameters can be computed with the same asymptotic complexity as function evaluation by using \emph{reverse-mode} differentiation (backpropagation) through the adjoint-space propagation, as described in~\cite{Goh.2025-LiealgebraicClassicalSimulations}. This yields the final row of \cref{tab:gsimcomplexities}.

\section{Details and proofs for translation-invariant algebras}
\label{prf:FFalgebra}

The translation-invariant operator algebra $\g^{\mathrm{TI}}\subseteq\mathfrak{u}(2^n)$~\eqref{eq:gTIdefinition} is the commutant of the cyclic translation action $C_n$ on $n$ qubits, i.e.\ the fixed-point subspace of conjugation by the shift unitaries $\{U_\tau^k\}_{k=0}^{n-1}$. In \cref{sssec:TI-systems} we focus on translation-invariant generator families that, in addition, admit a free-fermion description. Such models generate DLAs contained in the intersection $\g^{\mathrm{FF}}\cap \g^{\mathrm{TI}}$, and are naturally expressed as translation-invariant sums of Pauli strings.
Rather than working with explicit Pauli expansions of these sums, we introduce \emph{Pauli cycles} as a symmetry-adapted basis obtained by twirling Pauli strings with the Reynolds projector~\eqref{eq:cyclic_twirl}. This representation supports more efficient Lie-algebraic primitives: linear-independence tests reduce to comparisons of orbit representatives, and commutators can be evaluated with improved asymptotic cost, which directly accelerates the preprocessing stage of $\gsim$ in translation-invariant settings.

This appendix collects the supporting material for \cref{sssec:TI-systems}. In particular, we provide (i) a derivation of the Pauli cycle commutator identity \cref{th:PCcommutator}, (ii) a proof of the linear-time improvement for bounded-weight cycles in \cref{cor:bounded_weight_PC_commutator}, and (iii) numerical runtime benchmarks demonstrating the practical advantage of Pauli cycles for free-fermion algebras within $\g^{\mathrm{FF}}\cap\g^{\mathrm{TI}}$.

\subsection{Proof of \cref{th:PCcommutator}}
\label{prf:PCcommutator}

\begin{proof}
By \cref{def:Paulicycle}, the Pauli cycles $O_P, O_{P'}$ are given by $i\mathcal{T}(P)$ and $i\mathcal{T}(P')$, where $\mathcal{T}(\cdot)$ is the cyclic twirling map~\eqref{eq:cyclic_twirl}. Expanding the commutator yields a double sum over the translation indices:
\begin{align}
[O_P, O_{P'}] 
&= \left[ i\mathcal{T}(P), i\mathcal{T}(P') \right] \nonumber\\
&= -\frac{1}{n^2} \sum_{j=0}^{n-1} \sum_{l=0}^{n-1} \left[ U_\tau^j P (U_\tau^\dagger)^j,\, U_\tau^l P' (U_\tau^\dagger)^l \right].
\end{align}
Because $U_\tau$ is a unitary representation of the cyclic group, we can pull the global conjugation by $U_\tau^j$ outside of the commutator:
\begin{equation}
[O_P, O_{P'}] = -\frac{1}{n^2} \sum_{j=0}^{n-1} U_\tau^j \left( \sum_{l=0}^{n-1} \left[ P,\, U_\tau^{l-j} P' (U_\tau^\dagger)^{l-j} \right] \right) (U_\tau^\dagger)^j.
\end{equation}
We now change the summation variable in the inner sum to $k = (l-j) \bmod n$. Since the summation runs over a full period of the cyclic group, the set of indices is invariant, and the inner sum over $k$ becomes independent of $j$:
\begin{equation}
[O_P, O_{P'}] = -\frac{1}{n^2} \sum_{j=0}^{n-1} U_\tau^j \left( \sum_{k=0}^{n-1} \left[ P,\, P'^{(k)} \right] \right) (U_\tau^\dagger)^j,
\end{equation}
where we substituted $P'^{(k)} \equiv U_\tau^k P' (U_\tau^\dagger)^k$. 
Because the inner sum over $k$ is now independent of $j$, we can identify the outer sum over $j$ as simply the definition of the twirling map $\mathcal{T}$:
\begin{align}
[O_P, O_{P'}] 
&= -\frac{1}{n} \sum_{k=0}^{n-1} \mathcal{T}\left( [P, P'^{(k)}] \right). 
\label{eq:PC_proof_twirl_step}
\end{align}
Finally, we rewrite this expression in terms of the cycle basis. Recalling that $O_A = i\mathcal{T}(A)$, we have $\mathcal{T}(A) = -i O_A$. Choosing $A = [P, P'^{(k)}]$, we get
\begin{equation}
-\mathcal{T}\left( [P, P'^{(k)}] \right) = i O_{[P, P'^{(k)}]} = O_{i[P, P'^{(k)}]}.
\end{equation}
Substituting this back into \eqref{eq:PC_proof_twirl_step} yields the desired identity:
\begin{equation}
[O_P, O_{P'}] = \frac{1}{n} \sum_{k=0}^{n-1} O_{i[P, P'^{(k)}]}.
\end{equation}

To establish the algorithmic complexity, note that the sum contains exactly $n$ terms. For each $k \in \{0, \dots, n-1\}$, evaluating the standard Pauli string commutator $[P, P'^{(k)}]$ requires $\mathcal{O}(n)$ operations (e.g., via symplectic inner products or phase tracking). Since there are $n$ such commutators to evaluate, computing the full cycle commutator requires $\mathcal{O}(n^2)$ time, reducing the naive $\mathcal{O}(n^3)$ cost asymptotically by a factor of $n$.
\end{proof}

\subsection{Proof of \cref{cor:bounded_weight_PC_commutator}}
\label{prf:bounded_weight_PC_commutator}

\begin{proof}
By \cref{th:PCcommutator}, the cycle commutator requires evaluating the sum $\frac{1}{n}\sum_{k=0}^{n-1} O_{i[P,P'^{(k)}]}$. Because $P$ and $P'^{(k)}$ are tensor products of single-qubit Pauli matrices, they commute if and only if they commute on every individual qubit. A necessary condition for non-commutation is that their non-identity supports on the lattice intersect. 

Let $\mathrm{supp}(P) \subset \{0, \dots, n-1\}$ be the set of sites where $P$ acts non-trivially, with $|\mathrm{supp}(P)| = w_P$. The shifted string $P'^{(k)}$ has support $\mathrm{supp}(P'^{(k)}) = \{ (y + k) \bmod n \mid y \in \mathrm{supp}(P') \}$. 
The intersection $\mathrm{supp}(P) \cap \mathrm{supp}(P'^{(k)})$ is non-empty if and only if there exists $x \in \mathrm{supp}(P)$ and $y \in \mathrm{supp}(P')$ such that
\begin{equation}
    x \equiv y + k \pmod n \iff k \equiv x - y \pmod n.
\end{equation}
Since there are exactly $w_P$ choices for $x$ and $w_{P'}$ choices for $y$, there are at most $w_P w_{P'}$ distinct relative shifts $k \in \{0, \dots, n-1\}$ for which the supports overlap. For all other shifts, $[P, P'^{(k)}] = 0$.

Consequently, the summation over $n$ shifts reduces to at most $w_P w_{P'}$ non-zero terms. Computing the active shifts $k$ requires $\mathcal{O}(w_P w_{P'})$ integer operations. For each active shift, evaluating the standard Pauli string commutator $[P, P'^{(k)}]$ and extracting the new Pauli cycle representative takes $\mathcal{O}(n)$ time in the standard dense representation. 

The total time to compute the cycle commutator is thus bounded by $\mathcal{O}(w_P w_{P'} \cdot n)$. Under the assumption that the local weights are bounded independent of system size, $w_P, w_{P'} = \mathcal{O}(1)$, the overall asymptotic complexity reduces strictly to $\mathcal{O}(n)$.
\end{proof}

\subsection{Numerical runtime benchmarks for translation-invariant free-fermion algebras}
\label{appssec:BMforFFalgebrasTI}

The code released alongside this paper implements optimized preprocessing routines for $\gsim$, covering (i) Pauli string arithmetic for both ``free'' (individual Pauli strings) and ``sum'' (linear combinations of Pauli strings) generator sets, and (ii) Pauli cycle primitives for translationally invariant sums as introduced in \cref{sssec:TI-systems}. All mentioned implementations are based on the binary symplectic representation of Pauli strings (cf.\ \cref{ssec:FF_Paulialgebras}) as provided by \texttt{PauliEngine}~\cite{Muller.2026-PauliEngineHighperformantSymbolic} used as a \texttt{BLAS}-like backend for Pauli strings.

We benchmark the $\gsim$ preprocessing stage (Lie closure and structure constant computation) for the transverse-field Ising model (TFIM) generator families described in \cref{ex:TFIM}. Concretely, we consider three input formats:
(i) a freely parameterized Pauli string set (``free''), 
(ii) translation-invariant summed generators expressed in the Pauli string basis (``sums''), and
(iii) the same translation-invariant generators expressed in the symmetry-adapted Pauli cycle basis (``cycles'').
For the commutator primitive, we use the asymptotic costs established earlier: Pauli string commutators scale as $t_\mathrm{com}=\mathcal{O}(n)$ (cf.\ \cref{cor:PSlineartimecommutators}); naive commutators between translation-invariant sums scale as $t_\mathrm{com}=\mathcal{O}(n^3)$ (pairwise commutators of $n$ summands, each costing $\mathcal{O}(n)$); and Pauli cycle commutators scale as $t_\mathrm{com}=\mathcal{O}(n^2)$ in the worst case (cf.\ \cref{th:PCcommutator}), with a further reduction to $\mathcal{O}(n)$ for bounded Pauli weight (cf.\ \cref{cor:bounded_weight_PC_commutator}).

Using the preprocessing complexity bound $\mathcal{O}(\dim(\mathfrak{g})^2\,t_\mathrm{com})$ (cf.\ \cref{tab:gsimcomplexities}) together with the TFIM DLA dimensions from \cref{ex:TFIM}, we obtain the pessimistic worst-case scaling summarized in \cref{tab:TFIMpreprocessingscaling}.

\begin{table}[htbp]
\centering
\begin{tabular}{c|c|c|c}
$\mathcal{O}(\dim(\mathfrak{g})^2\,t_\mathrm{com})$
& free
& sums
& cycles
\\ \hline
open boundaries
& $\mathcal{O}(n^5)$
& $\mathcal{O}(n^7)$
& $\mathcal{O}(n^6)$
\\
periodic boundaries
& $\mathcal{O}(n^5)$
& $\mathcal{O}(n^5)$
& $\mathcal{O}(n^4)$
\end{tabular}
\caption{\label{tab:TFIMpreprocessingscaling}
Runtime upper bound scalings of $\gsim$ preprocessing for TFIM generator families under different input representations. Here ``free'' refers to the edge-local Pauli string generators~\eqref{eq:freeTFIMgenerators}, ``sums'' to translation-invariant summed generators in the Pauli string basis~\eqref{eq:stdTFIMgenerators}, and ``cycles'' to the same summed generators represented in the Pauli cycle basis (cf.\ \cref{sssec:TI-systems}). The bounds combine $\mathcal{O}(\dim(\mathfrak{g})^2\,t_\mathrm{com})$ with the relevant TFIM DLA dimensions and worst-case commutator costs $t_\mathrm{com}$.}
\end{table}

The measured wall-clock runtimes are shown in \cref{fig:BMstructureconstantsCn}. Across all settings we observe substantially better average-case scaling than the pessimistic bounds in \cref{tab:TFIMpreprocessingscaling}, and we see a clear practical advantage of the symmetry-adapted Pauli cycle basis for translation-invariant generators in the periodic setting. In particular, for periodic boundaries the cycle representation yields the smallest fitted exponent among the three variants, whereas for open boundaries the freely parameterized Pauli string representation performs best. Overall, these results indicate that the implemented preprocessing routines remain efficient over the tested range and that symmetry-adapted bases can yield tangible runtime improvements in translation-invariant settings.

\begin{figure}[htbp]
\includegraphics[width=1\linewidth]{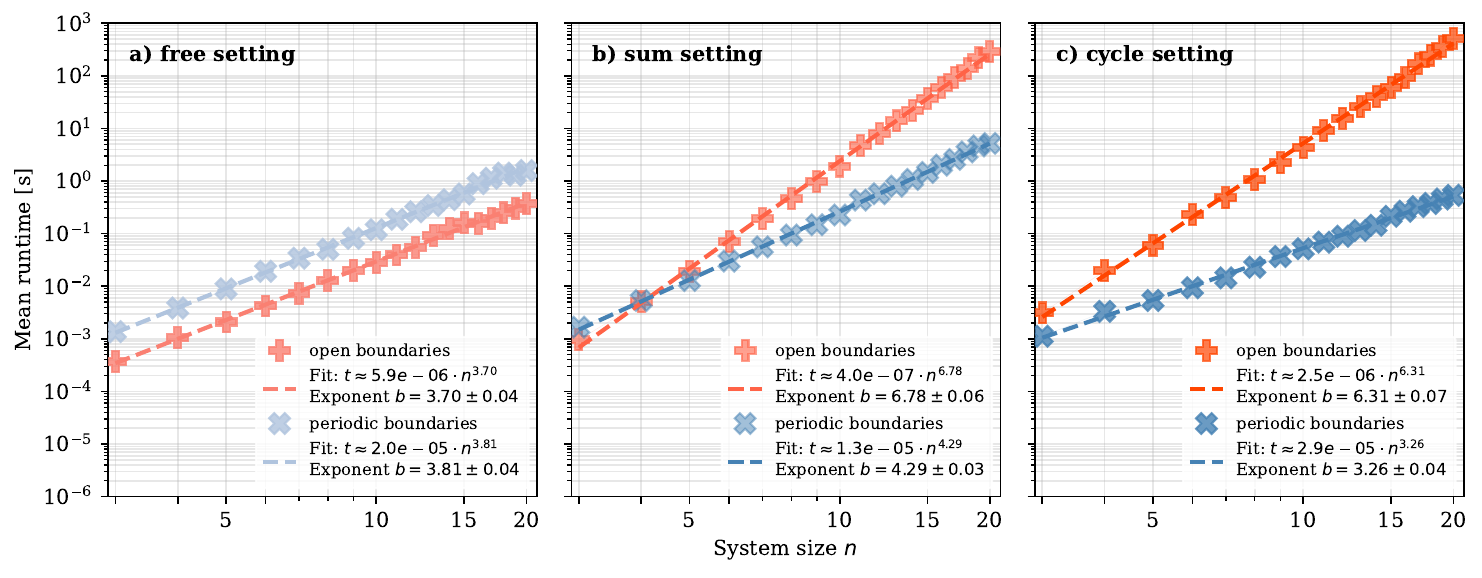}
\caption{\label{fig:BMstructureconstantsCn}
Runtime benchmarks for $\gsim$ preprocessing (Lie closure and structure constant computation) for TFIM-type free-fermion algebras under three input representations. Panel (a): freely parameterized Pauli string generators \eqref{eq:freeTFIMgenerators}. Panel (b): translation-invariant summed generators in the Pauli string basis. Panel (c): the same summed generators represented in the Pauli cycle basis (see~\cref{def:Paulicycle}). Points show mean wall-clock time versus system size $n$ (averaged over $100$ independent runs); dashed lines are least-squares fits of the form $t=a\,n^{b}$ with fitted exponents.
The fitted scaling exponents are: (a) $b=3.70\pm 0.04$ (open), $b=3.81\pm 0.04$ (periodic); (b) $b=6.78\pm 0.06$ (open), $b=4.29\pm 0.03$ (periodic); (c) $b=6.31\pm 0.07$ (open), $b=3.26\pm 0.04$ (periodic). All benchmarks were run on a single CPU node with AMD EPYC 7713 processors (128 cores total) and 1\,TiB RAM per node; no GPUs were used.}
\end{figure}

\section{Details and proofs for permutation-invariant algebras}
\label{prf:pauliorbitalgebra}
The permutation-invariant algebra $\g^\mathrm{PI}\subseteq \mathfrak{u}(2^n)$~\eqref{eq:gPIdefinition} consists of all skew-Hermitian operators that commute with the natural action of the symmetric group $S_n$ on $n$ qubits, i.e.\ with all tensor-factor permutations. In Section~\ref{sec:Pauliorbitalgebra}, we propose the \emph{Pauli orbit} formalism as a symmetry-adapted representation of $\g^\mathrm{PI}$ that avoids explicit expansions into exponentially many Pauli strings and enables efficient algebraic primitives required by $\gsim$.
This appendix collects the supporting material for that section. In particular, we provide: (i) a derivation of the Schur--Weyl decomposition underlying Proposition~\ref{prop:schur_weyl_commutant}, (ii) a proof of the Pauli orbit basis and dimensionality, (iii) algorithmic details and proofs for efficient commutator computations (see \cref{th:POcommutator,prop:POdetailedcommutator}) in the Pauli orbit basis, along with illustrative examples, and (iv) a conceptual comparison of the Lie-Algebraic simulation of $S_n$-invariant systems with the Schur-basis method proposed in Ref.~\cite{Anschuetz.2023-EfficientClassicalAlgorithms}.

\subsection{Proof of \cref{prop:schur_weyl_commutant}}
\label{prf:schur_weyl_commutant}
\begin{proof}
Consider the Hilbert space $\mathcal{H}=(\mathbb{C}^{2})^{\otimes n}$.
Let $S_n$ act on $\mathcal{H}$ by permuting tensor factors, $\pi\mapsto U_\pi$, and let $U(2)$ act diagonally by $V\mapsto V^{\otimes n}$. These two actions commute.
By Schur--Weyl duality, $\mathcal{H}$ admits a decomposition into irreducible $U(2)\times S_n$-modules,
\begin{equation}
\mathcal{H}\cong \bigoplus_{\lambda\vdash n,\ \ell(\lambda)\le 2}\ \mathcal{V}_\lambda\otimes \mathcal{S}_\lambda,
\label{eq:Schur-Weyl-basis}
\end{equation}
where $\lambda=(\lambda_1,\lambda_2)$ ranges over partitions of $n$ with at most two parts, $\mathcal{V}_\lambda$ is an irrep of $U(2)$, and $\mathcal{S}_\lambda$ is the corresponding Specht module (irrep of $S_n$).

The double-commutant form of Schur--Weyl duality implies that the operators commuting with all permutations are exactly those that act nontrivially only on the $U(2)$ factors. Thus
\begin{equation}
\{A\in \mathcal{L}(\mathcal{H}) : [A,U_\pi]=0\ \forall \pi\in S_n\}
\cong
\bigoplus_{\lambda}\mathcal{L}(\mathcal{V}_\lambda),
\end{equation}
where, in the decomposition~\eqref{eq:Schur-Weyl-basis}, any such operator has the block form
$A = \bigoplus_{\lambda} \left(M_\lambda \otimes I_{\mathcal{S}_\lambda}\right)$ with $M_\lambda\in \mathrm{End}(\mathcal{V}_\lambda)$. Restricting to skew-Hermitian operators therefore yields the Lie algebra commutant
\begin{equation}
\g^{\mathrm{PI}}
=
\{A\in \mathfrak{u}(2^n): [A,U_\pi]=0\ \forall \pi\in S_n\}
\cong
\bigoplus_{\lambda}\mathfrak{u}(\dim(\mathcal V_\lambda)).
\label{eq:gPIisomorphism}
\end{equation}

For qubits, the relevant partitions have at most two parts and may be written as $\lambda=(n-k,k)$ with $0\le k\le \lfloor n/2\rfloor$. The corresponding $U(2)$-irrep is the spin-$J$ representation with $J=\frac{n}{2}-k$, hence
\begin{equation}
\dim(\mathcal{V}_{(n-k,k)}) = 2J+1 = n-2k+1.
\end{equation}
Substituting both the new summation index $k$ and the dimension $\dim(\mathcal{V_\lambda})$ into the isomorphism~\eqref{eq:gPIisomorphism} yields exactly the decomposition of Eq.~\eqref{eq:PIalgebradecomposition}. 
\end{proof}

Similar considerations about the theory of permutation-invariant Lie algebras appear in Refs.~\cite{Albertini.2018-ControllabilitySymmetricSpin,Anschuetz.2023-EfficientClassicalAlgorithms,Nguyen.2024-TheoryEquivariantQuantum,Mancinska.2025-Classificationipermutationinvariant,Bastin.2025-Permutationinvariantprocesses}.

\subsection{Proof of \cref{prop:orbit_basis}}
\label{prf:Pauliorbitbasis}
\begin{proof}
Let $\mathcal{P}_n$ denote the set of all $n$-qubit Pauli strings. The set $\{i P \mid P \in \mathcal{P}_n\}$ forms a basis for $\mathfrak{u}(2^n)$ over $\mathbb{R}$. Because the Reynolds operator $\mathcal{S}: \mathfrak{u}(2^n) \to \g^{\mathrm{PI}}$ is a surjective linear projection, the image of this basis spans $\g^{\mathrm{PI}}$. 

Conjugation by a permutation $\pi \in S_n$ acts on a Pauli string $P \in \mathcal{P}_n$ purely by permuting the lattice indices of its single-qubit factors. This action preserves the total count of $X$, $Y$, and $Z$ matrices within the string. Thus, the equivalence classes (orbits) of $\mathcal{P}_n$ under $S_n$ conjugation are uniquely and fully classified by the tuple $(p,q,r)$ specifying the number of $X$, $Y$, and $Z$ operators, respectively, where $p,q,r \ge 0$ and $p+q+r \le n$. 

By selecting the canonical representative $P_{p,q,r}$ for each orbit, we capture exactly one element from every equivalence class. Therefore, the set of projected representatives $\mathcal{B}_n = \{i\mathcal{S}(P_{p,q,r})\}$ spans $\g^{\mathrm{PI}}$.

To establish linear independence, we note that the orbit $B_{p,q,r}$ is a uniform linear combination of Pauli strings with exact counts $(p,q,r)$. For any distinct tuples $(p,q,r) \neq (p',q',r')$, the constituent Pauli strings belong to disjoint equivalence classes. Because distinct Pauli strings are strictly orthogonal with respect to the Hilbert-Schmidt inner product $\langle A, B \rangle = \mathrm{Tr}(A^\dagger B)$, operators supported on disjoint subsets of $\mathcal{P}_n$ are also orthogonal. Consequently, $\langle B_{p,q,r}, B_{p',q',r'} \rangle = 0$. Since the elements of $\mathcal{B}_n$ are mutually orthogonal and non-zero, they are linearly independent and thus form a basis for $\g^{\mathrm{PI}}$.

Finally, the cardinality of $\mathcal{B}_n$ is given by the number of valid tuples $(p,q,r)$. Letting $s = n - p - q - r \ge 0$ denote the number of identity operators, the size of $\mathcal{B}_n$ is the number of weak integer compositions of $n$ into exactly $4$ parts ($p+q+r+s = n$). By a standard stars-and-bars combinatorial argument, this is exactly $\binom{n+4-1}{4-1} = \binom{n+3}{3}$.
\end{proof}

\subsection{Supplementary information regarding \cref{th:POcommutator}}
\label{prf:POcommutator}

This section provides an explicit coefficient formula for the commutator of two Pauli orbits in terms of contingency tables (see \cref{appssec:POcommutatordetailedformula}) and a polynomial-time evaluation procedure based on enumerating feasible overlap patterns (see \cref{appssec:POcommutatorevaluationpoly}).

\subsubsection{Detailed commutator formula}
\label{appssec:POcommutatordetailedformula}

\begin{proposition}[Commutator of Pauli orbits via contingency tables]
For two Pauli orbits $B_{p,q,r}$ and $B_{p',q',r'}$ defined as in \cref{def:Pauliorbit}, their commutator admits an expansion
\begin{equation}
[B_{p,q,r}, B_{p',q',r'}] = 
\sum_{(\tilde{p}, \tilde{q}, \tilde{r})}\: c_{\tilde{p}, \tilde{q}, \tilde{r}}\: B_{\tilde{p},\tilde{q}, \tilde{r}},
\label{eq:POdetailedcommutator_basis}
\end{equation}
where the coefficients are given by
\begin{equation}
c_{\tilde{p}, \tilde{q}, \tilde{r}}= \frac{2}{N_{\mathrm{terms}}(p,q,r;n)N_{\mathrm{terms}}(p',q',r';n)}
\sum_{\substack{\eta \in \mathcal{F}_{(p,q,r;p',q',r';n)}\\(\tilde{p}(\eta), \tilde{q}(\eta), \tilde{r}(\eta))=(\tilde{p}, \tilde{q}, \tilde{r})}}\sgn(\eta)\: W(\eta).
\label{eq:POdetailedcommutatorC}
\end{equation}
Here, $\eta\in \mathbb{N}_0^{4\times 4}$ is a contingency table indexed by ${A,B\in\{X,Y,Z,I\}}$ encoding the overlap pattern between Pauli strings drawn from the two orbits, and $N_{\mathrm{terms}}(p,q,r;n)$~\eqref{eq:PStermsinPO} denotes the number of distinct Pauli strings in the orbit labeled by $(p,q,r)$.
The feasible set $\mathcal{F}_{(p,q,r;p',q',r';n)}$ consists of all tables $\eta$ satisfying the row-sum constraints
\begin{equation}
\sum_B \eta_{XB} = p,\qquad \sum_B \eta_{YB} = q,\qquad \sum_B \eta_{ZB} = r,\qquad \sum_B \eta_{IB} = n-p-q-r,
\label{eq:POoverlaprowsums}
\end{equation}
the column-sum constraints
\begin{equation}
\sum_A \eta_{AX} = p',\qquad \sum_A \eta_{AY} = q',\qquad \sum_A \eta_{AZ} = r',\qquad \sum_A \eta_{AI} = n-p'-q'-r',
\label{eq:POoverlapcolsums}
\end{equation}
and the parity constraint selecting only nonzero Pauli string commutators,
\begin{equation}
D(\eta):= \sum_{\substack{A\neq B\\ A,B\neq I}} \eta_{AB}=1\pmod{2}.
\label{eq:POoddcommutatorsites}
\end{equation}
Equivalently,
\begin{equation}
    \mathcal{F}_{(p,q,r;p',q',r';n)}=\left\{\eta\in \mathbb{N}_0^{4\times 4}:\ \eqref{eq:POoverlaprowsums},\ \eqref{eq:POoverlapcolsums},\  \eqref{eq:POoddcommutatorsites}\ \text{hold} \right \}.
\label{eq:POcontingencyfeasibleset}
\end{equation}
Each feasible $\eta$ contributes to exactly one output orbit label $(\tilde{p},\tilde{q},\tilde{r})$ determined by the sitewise Pauli product rule,
\begin{equation}
\begin{aligned}
& \tilde{p}(\eta)=\left(\eta_{XI}+\eta_{IX}\right)+\left(\eta_{YZ}+\eta_{ZY}\right), \\
& \tilde{q}(\eta)=\left(\eta_{YI}+\eta_{IY}\right)+\left(\eta_{ZX}+\eta_{XZ}\right), \\
& \tilde{r}(\eta)=\left(\eta_{ZI}+\eta_{IZ}\right)+\left(\eta_{XY}+\eta_{YX}\right),
\end{aligned}
\label{eq:POneworbit}
\end{equation}
and for a fixed overlap pattern $\eta$, the number of ordered pairs of Pauli strings realizing $\eta$ equals the multinomial coefficient
\begin{equation}
W(\eta)=\frac{n!}{\prod_{A,B} \eta_{A,B}!}.
\label{eq:POpatternweight}
\end{equation}
Finally, the relative sign of the commutator contribution is determined by
\begin{equation}
    \sgn(\eta)=(-1)^{E(\eta)+\frac{D(\eta)-1}{2}}\quad\text{with}\quad E(\eta)=\eta_{YX}+\eta_{ZY}+\eta_{XZ}.
\label{eq:POcommutatorSIGN}
\end{equation}
Together, Eqs.~\eqref{eq:POcontingencyfeasibleset}--\eqref{eq:POcommutatorSIGN} specify an explicit, finite, and purely combinatorial rule for the coefficients $c_{\tilde p,\tilde q,\tilde r}$ in \eqref{eq:POdetailedcommutatorC}, and hence determine the full Pauli orbit expansion \eqref{eq:POdetailedcommutator_basis} of the commutator.
\label{prop:POdetailedcommutator}
\end{proposition}

\begin{proof}
For brevity, write
\[
N:=N_{\mathrm{terms}}(p,q,r;n),
\qquad
N':=N_{\mathrm{terms}}(p',q',r';n).
\]
By Definition~\ref{def:Pauliorbit}, the two Pauli orbits are
\[
B_{p,q,r}= \frac{i}{N}\sum_{P\in\Omega_{p,q,r}} P,
\qquad
B_{p',q',r'}= \frac{i}{N'}\sum_{Q\in\Omega_{p',q',r'}} Q.
\]
Hence, by bilinearity,
\begin{equation}
[B_{p,q,r},B_{p',q',r'}]
=
\frac{1}{N\,N'}
\sum_{P\in\Omega_{p,q,r}}\ \sum_{Q\in\Omega_{p',q',r'}}
[iP,iQ].
\label{eq:PO_expand_pairs_short}
\end{equation}

For fixed Pauli strings \(P,Q\in\{I,X,Y,Z\}^{\otimes n}\), define their overlap table
\(\eta=\eta(P,Q)\in\mathbb{N}_0^{4\times 4}\) by
\[
\eta_{AB}=\bigl|\{t:(P_t,Q_t)=(A,B)\}\bigr|.
\]
The row and column sums of \(\eta\) are determined by the orbit labels, so
\(\eta(P,Q)\in\mathcal{F}_{(p,q,r;p',q',r';n)}\).

It is standard that the commutator of two Pauli strings depends only on the local pair types \((P_t,Q_t)\), and hence only on the table \(\eta\): it vanishes unless \(D(\eta)\) is odd, and in that case
\begin{equation}
[iP,iQ]=2\,\sgn(\eta)\, i\,R(P,Q),
\label{eq:PO_string_commutator_short}
\end{equation}
where \(\sgn(\eta)\) is given by \eqref{eq:POcommutatorSIGN} and \(R(P,Q)\) denotes the resulting Pauli string obtained from the sitewise Pauli product rule (up to the overall sign already encoded in \(\sgn(\eta)\)). Moreover, the orbit label of \(R(P,Q)\) is determined by \(\eta\) via \eqref{eq:POneworbit}; denote it by
\((\tilde p(\eta),\tilde q(\eta),\tilde r(\eta))\).

We now group the double sum \eqref{eq:PO_expand_pairs_short} by overlap pattern. For a fixed feasible \(\eta\), the number of ordered pairs
\((P,Q)\in \Omega_{p,q,r}\times\Omega_{p',q',r'}\) with \(\eta(P,Q)=\eta\) is the multinomial coefficient \(W(\eta)\) from \eqref{eq:POpatternweight}. Among these pairs, the multiset of outputs \(R(P,Q)\) is invariant under simultaneous qubit permutations. Since \(S_n\) acts transitively on the orbit \(\Omega_{\tilde p(\eta),\tilde q(\eta),\tilde r(\eta)}\), each
\(R\in\Omega_{\tilde p(\eta),\tilde q(\eta),\tilde r(\eta)}\) occurs equally often. Therefore
\begin{equation}
\sum_{\substack{(P,Q):\,\eta(P,Q)=\eta}} iR(P,Q)
=
W(\eta)\, B_{\tilde p(\eta),\tilde q(\eta),\tilde r(\eta)}.
\label{eq:PO_uniformity_short}
\end{equation}
Indeed, if
\(
\widetilde N
=
N_{\mathrm{terms}}(\tilde p(\eta),\tilde q(\eta),\tilde r(\eta);n)
\),
then each output string in the orbit appears exactly \(W(\eta)/\widetilde N\) times, and
\[
\sum_{\substack{(P,Q):\,\eta(P,Q)=\eta}} iR(P,Q)
=
\frac{W(\eta)}{\widetilde N}
\sum_{R\in\Omega_{\tilde p(\eta),\tilde q(\eta),\tilde r(\eta)}} iR
=
W(\eta)\, B_{\tilde p(\eta),\tilde q(\eta),\tilde r(\eta)}.
\]

Substituting \eqref{eq:PO_string_commutator_short} and \eqref{eq:PO_uniformity_short} into \eqref{eq:PO_expand_pairs_short} gives
\[
[B_{p,q,r},B_{p',q',r'}]
=
\frac{2}{N\,N'}
\sum_{\eta\in\mathcal{F}_{(p,q,r;p',q',r';n)}}
\sgn(\eta)\,W(\eta)\,
B_{\tilde p(\eta),\tilde q(\eta),\tilde r(\eta)}.
\]
Finally, collecting all terms with the same output label
\((\tilde p,\tilde q,\tilde r)\) yields
\eqref{eq:POdetailedcommutator_basis} and
\eqref{eq:POdetailedcommutatorC}.
\end{proof}

\subsubsection{Polynomial-time evaluation and an enumeration algorithm}
\label{appssec:POcommutatorevaluationpoly}

Proposition~\ref{prop:POdetailedcommutator} expresses the commutator coefficients between two Pauli orbits as finite sums over feasible contingency tables. We now make explicit why these sums can be evaluated in time polynomial in $n$, and propose an efficient algorithm with explicit upper bounds on its runtime.

A contingency table $\eta\in\mathbb{N}_{0}^{4\times 4}$ has 16 entries, but the row and column sums are fixed by the orbit labels, imposing eight linear constraints~\eqref{eq:POoverlaprowsums}--\eqref{eq:POoverlapcolsums}. Since the total sum of all entries equals $n$, these constraints have one redundancy, effectively reducing the degrees of freedom to $16-(8-1)=9$. A convenient parametrization is to treat the $3\times 3$ block $\{\eta_{AB}\}_{A, B\in\{X,Y,Z\}}$ as free variables. Once this block is fixed, all entries involving $I$ are uniquely determined by the marginals of the table,
\begin{equation}
\eta_{AI}= \mathrm{row}_A-\sum_{B\in\{X,Y,Z\}}\eta_{AB},
\qquad
\eta_{IB}= \mathrm{col}_B-\sum_{A\in\{X,Y,Z\}}\eta_{AB},
\qquad
\eta_{II}= n-\sum_{(A,B)\neq (I,I)}\eta_{AB},
\label{eq:POcommarginals}
\end{equation}
and the feasibility of a candidate block reduces to checking that these forced entries are nonnegative integers.

A straightforward enumeration proceeds row-by-row over $A\in\{X,Y,Z\}$. For a given row $A$, we must choose three nonnegative integers $\eta_{AX},\eta_{AY},\eta_{AZ}$, while the remaining slack variable $\eta_{AI}$ is fixed according to Eq.~\eqref{eq:POcommarginals}. This corresponds to splitting an integer of size $\mathcal{O}(n)$ into four non-negative parts (three explicit entries plus the slack), which has $\binom{\mathrm{row}_A+3}{3}=\mathcal{O}(n^3)$ possibilities. Choosing such triples independently for the three rows yields at most $O(n^9)$ candidate $3\times 3$ blocks in the worst case. For each candidate, the remaining $I$-entries are computed using ~\eqref{eq:POcommarginals}, column feasibility is checked, and the parity condition~\eqref{eq:POoddcommutatorsites} is enforced. The final contribution is then accumulated into the coefficient of the corresponding output orbit label. All of these per-candidate operations take $O(1)$ time once factorials or multinomials are precomputed. Therefore, the overall worst-case runtime for evaluating a full commutator between two orbits by direct enumeration is $t_\mathrm{com}=O(n^9)$. Algorithm~\ref{alg:fast_commutator} instructively implements this strategy. 

\begin{myalgorithm}[Full evaluation of Pauli orbit commutator via pruned enumeration]
\label{alg:fast_commutator}
\begin{algorithmic}[1]
\Require $n, (p,q,r), (p',q',r'), F=\{k!\mid k \in \{0,\dots,n\}\}$
\Comment{system size, input orbits, factorials}
\Ensure $S_{out}: (\tilde{p}, \tilde{q}, \tilde{r}) \to c_{\tilde p,\tilde q,\tilde r}$
\Comment{commutator coefficients}

\State Define $s \gets n-p-q-r$ and $s' \gets n-p'-q'-r'$
\If{$s < 0$ or $s' < 0$}
    \State \Return $\emptyset$
\EndIf
\State Initialize $S_{out} \gets \emptyset$, and initial column capacities $c^{(0)} \gets \{p',q',r',s'\}$

\Statex \Comment{enumerate $3\times 3$ free variable block $\eta_{A,B}$ for $A,B\in\{X,Y,Z\}$}

\For{$\eta_{XX}, \eta_{XY}, \eta_{XZ}$ satisfying partial row $X$ and column capacities $c^{(0)}$}
    \State $\eta_{XI}\gets p - \eta_{XX} - \eta_{XY} - \eta_{XZ}$
    \If{$\eta_{XI}>c_I^{(0)}$}
        \Continue \Comment{pruning}
    \EndIf
    \State $c^{(1)} \gets c^{(0)} - \eta_{X\{X,Y,Z,I\}}$ \Comment{update column capacities}

    \For{$\eta_{YX}, \eta_{YY}, \eta_{YZ}$ satisfying partial row $Y$ and column capacities $c^{(1)}$}
        \State $\eta_{YI}\gets q - \eta_{YX} - \eta_{YY} - \eta_{YZ}$
        \If{$\eta_{YI}>c_I^{(1)}$}
            \Continue \Comment{pruning}
        \EndIf
        \State $c^{(2)} \gets c^{(1)} - \eta_{Y\{X,Y,Z,I\}}$ \Comment{update column capacities}

        \For{$\eta_{ZX}, \eta_{ZY}, \eta_{ZZ}$ satisfying partial row $Z$ and column capacities $c^{(2)}$}
            \State $\eta_{ZI}\gets r - \eta_{ZX} - \eta_{ZY} - \eta_{ZZ}$
            \If{$\eta_{ZI}>c_I^{(2)}$}
                \Continue \Comment{pruning}
            \EndIf

            \State $\eta_{I\{X,Y,Z,I\}} \gets c^{(2)} - \eta_{Z\{X,Y,Z,I\}}$ \Comment{force final row}

            \If{$\min(\eta_{IX}, \eta_{IY}, \eta_{IZ}, \eta_{II})<0$}
                \Continue \Comment{ensure non-negativity}
            \EndIf

            \If{$\eta_{IX}+\eta_{IY}+\eta_{IZ}+\eta_{II}\neq s$}
                \Continue \Comment{pruning}
            \EndIf

            \Statex \Comment{verify parity condition~\eqref{eq:POoddcommutatorsites}}
            \If{$(\eta_{XY} + \eta_{XZ} + \eta_{YX} + \eta_{YZ} + \eta_{ZX} + \eta_{ZY}) \bmod 2 = 0$}
                \Continue
            \EndIf

            \State $(\tilde{p}, \tilde{q}, \tilde{r}) \gets$ \eqref{eq:POneworbit}
            \Comment{compute new orbit labels}

            \State $W(\eta)\gets$ \eqref{eq:POpatternweight} using $F$
            \Comment{compute weight of overlap pattern}

            \State $\sgn(\eta)\gets$ \eqref{eq:POcommutatorSIGN}
            \Comment{compute sign of overlap pattern}

            \State $S_{out}[\tilde{p}, \tilde{q}, \tilde{r}] \gets S_{out}[\tilde{p}, \tilde{q}, \tilde{r}] + \sgn(\eta)\cdot W(\eta)$
        \EndFor
    \EndFor
\EndFor

\State \Return $S_{out}$
\end{algorithmic}
\end{myalgorithm}

It is crucial to emphasize that while the $\mathcal{O}(n^9)$ complexity seems highly impractical it is only a \emph{worst-case} bound over all possible orbits. It reflects the maximal dimension of the corresponding transportation polytope of feasible contingency tables, and therefore applies only when the overlap constraints admit the largest possible search space, i.e., when the marginals are balanced such that $p,q,r,p',q',r'\approx n/4=\Theta(n)$. This corresponds to the center of the polytope with permutation symmetric sums of high-weight Pauli strings. However, most Pauli orbits in $\g^\mathrm{PI}$ are low weight and therefore closer to the boundaries of the polytope, where the number of feasible patterns becomes much smaller than $\mathcal{O}(n^9)$. This is well illustrated by Figure~\ref{fig:SNcommoutputorbits} and in the favorable runtime benchmarks presented in \cref{appssec:PObenchmarksstructureconstants}.

\begin{figure}[htbp]
\includegraphics[width=0.85\linewidth]{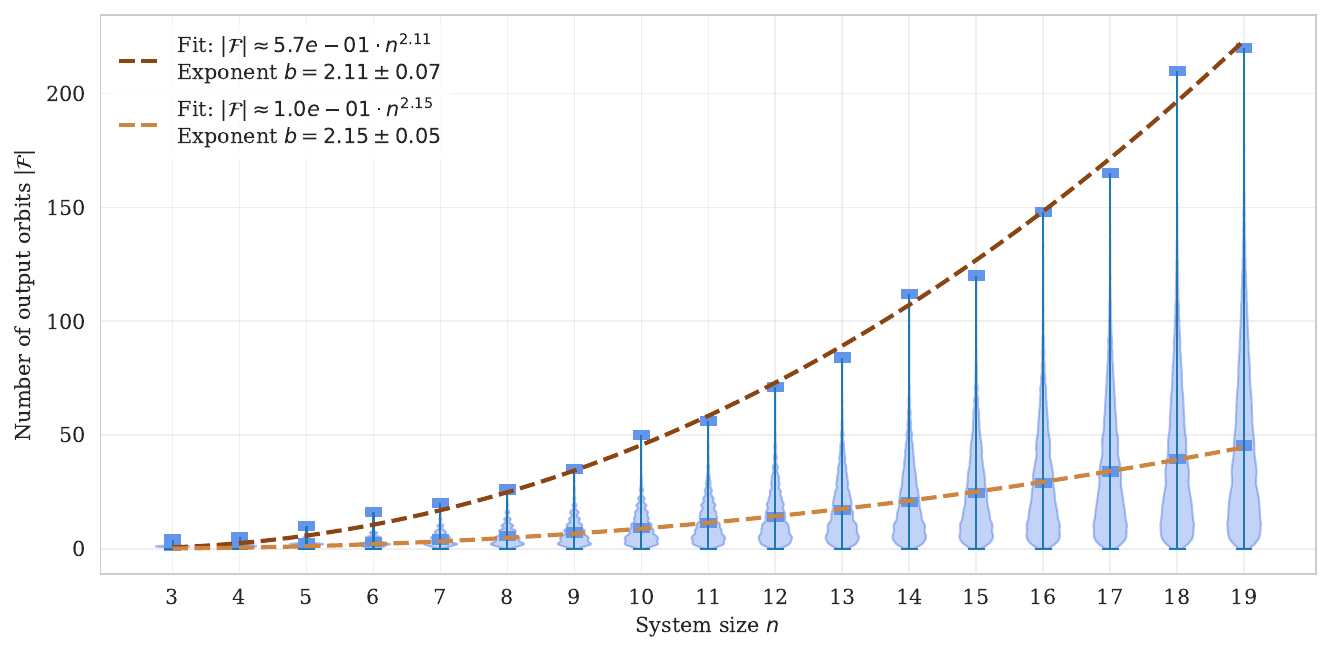}
\caption{\label{fig:SNcommoutputorbits}
Empirical output support sizes of Pauli orbit commutators in $\g^{\mathrm{PI}}$. 
For each system size $n$, we compute the commutator $[B_{p,q,r},B_{p',q',r'}]$ for all pairs of Pauli orbit basis elements $B_{p,q,r},B_{p',q',r'}\in\mathcal{B}_n$ using Algorithm~\ref{alg:fast_commutator}, and record the number of distinct output orbit labels $(\tilde p,\tilde q,\tilde r)$ with nonzero coefficient in the expansion~\eqref{eq:POcommutator}. 
Violin plots show the resulting distribution over commutator pairs; markers indicate the mean and the maximum, and the dashed curves are power-law fits to these two statistics (exponents $b\approx 2.1$). 
Across the tested range, the typical commutator output support grows approximately quadratically with $n$, illustrating that the worst-case $O(n^9)$ enumeration bound is conservative for generic orbit pairs.}
\end{figure}

\subsubsection{Constant-time evaluation for most practical settings}
\label{appssec:POcommutatorevaluationtargeted}
In most practical applications of $\gsim$, computing the full commutator expansion~\eqref{eq:POdetailedcommutator_basis} is unnecessary. The primary preprocessing cost arises from evaluating the structure constants~\eqref{eq:structureconstants}, or equivalently, from constructing the adjoint representation of the set of generators. For typical variational ansätze of the form~\eqref{eq:standardansatz}, the generator set size is independent of the system size, so only a constant number of commutator coefficients is required.

If one fixes a target output label $(\tilde p,\tilde q,\tilde r)$ and seeks only the coefficient $c_{\tilde p,\tilde q,\tilde r}$~\eqref{eq:POdetailedcommutatorC}, then the relations~\eqref{eq:POneworbit} impose three additional linear constraints on the contingency table variables. Starting from the $9$ degrees of freedom of feasible tables determined by the row and column marginals (cf.~\cref{appssec:POcommutatorevaluationpoly}), this reduces the worst-case enumeration dimension to $6$, yielding an $O(n^6)$ worst-case bound for computing a \emph{single} targeted coefficient by direct enumeration.

In the regimes most relevant to $\gsim$, the targeted output labels correspond to few-body generators and their adjoint-action images (cf.~\cref{ex:SN1,ex:SN2}, and hence have bounded Pauli weight $w:=\tilde p+\tilde q+\tilde r = \mathcal{O}(1)$. In this case, the target constraints~\eqref{eq:POneworbit} bound all off-diagonal and $I$-cross entries that can contribute to $(\tilde p,\tilde q,\tilde r)$. Explicitly, this implies
\begin{equation}
\eta_{XI},\eta_{IX},\eta_{YZ},\eta_{ZY}\le \tilde p,\qquad
\eta_{YI},\eta_{IY},\eta_{ZX},\eta_{XZ}\le \tilde q,\qquad
\eta_{ZI},\eta_{IZ},\eta_{XY},\eta_{YX}\le \tilde r.
\end{equation}
Thus, when $w=\mathcal{O}(1)$, these twelve variables range over sets of size independent of $n$. Once they are fixed, the diagonal entries $\eta_{XX},\eta_{YY},\eta_{ZZ},\eta_{II}$ are forced by the row marginals \eqref{eq:POoverlaprowsums}, and feasibility reduces to constant-time checks of nonnegativity, the column constraints \eqref{eq:POoverlapcolsums}, and the parity condition \eqref{eq:POoddcommutatorsites}. Consequently, the number of candidate overlap patterns that must be tested depends only on $(\tilde p,\tilde q,\tilde r)$ and not on $n$.

A simple upper bound on the number of candidates explored by the output-partitioning strategy is obtained by counting weak compositions into four parts:
\begin{equation}
\#\mathrm{candidates}(\tilde p,\tilde q,\tilde r)
\;\le\;
\binom{\tilde p+3}{3}\binom{\tilde q+3}{3}\binom{\tilde r+3}{3},
\end{equation}
which is $\mathcal{O}(1)$ in $n$ whenever the weight $w$ is bounded.

Algorithm~\ref{alg:fast_targeted_commutator} implements this output-partitioning strategy for computing a fixed, finite set of targeted coefficients. The subsequent benchmarking section confirms that, for the bounded-weight generator families relevant to our $\gsim$ experiments, the observed runtime is effectively independent of $n$ per targeted coefficient (up to the cost of integer arithmetic on factorial-derived weights). This provides the algorithmic justification for Corollary~\ref{cor:POcommutator_constant_time} in the main text.

\begin{myalgorithm}[Targeted commutator evaluation for Pauli orbits via output partitioning]
\label{alg:fast_targeted_commutator}
\begin{algorithmic}[1]
\Require $n, (p,q,r), (p',q',r'), T=\{(\tilde p,\tilde q, \tilde r)\},F=\{k!\ |\ k \in \{0, \dots, n\}\}$
\Comment{system size, input orbits, target orbits, factorials}
\Ensure $S_{out}: (\tilde{p}, \tilde{q}, \tilde{r}) \to c_{\tilde p,\tilde q,\tilde r}$ for $(\tilde p,\tilde q, \tilde r)\in T$ 
\Comment{targeted commutator coefficients}

\State Define $s \gets n-p-q-r$, \quad $s' \gets n-p'-q'-r'$
\State Initialize $S_{out} \gets \emptyset$

\For{each target orbit $(\tilde p,\tilde q, \tilde r)\in T$}
\If{$\tilde p+\tilde q+ \tilde r>n$} \textbf{continue} \Comment{skip invalid targets} \EndIf

\State \Comment{enumerate partitions of the target labels according to~\eqref{eq:POneworbit}}
\For{$(\eta_{XI}, \eta_{IX}, \eta_{YZ}, \eta_{ZY})$ partitioning $\tilde{p}$ satisfying capacities~\eqref{eq:POoverlaprowsums}--\eqref{eq:POoverlapcolsums}}
    \For{$(\eta_{YI}, \eta_{IY}, \eta_{ZX}, \eta_{XZ})$ partitioning $\tilde{q}$ satisfying capacities~\eqref{eq:POoverlaprowsums}--\eqref{eq:POoverlapcolsums}}
        \For{$(\eta_{ZI}, \eta_{IZ}, \eta_{XY}, \eta_{YX})$ partitioning $\tilde{r}$ satisfying capacities~\eqref{eq:POoverlaprowsums}--\eqref{eq:POoverlapcolsums}}
            
            \State \Comment{derive diagonal entries from row capacities~\eqref{eq:POoverlaprowsums}}
            \State $\eta_{XX} \gets p - \eta_{XI} - \eta_{XY} - \eta_{XZ}$
            \State $\eta_{YY} \gets q - \eta_{YI} - \eta_{YX} - \eta_{YZ}$
            \State $\eta_{ZZ} \gets r - \eta_{ZI} - \eta_{ZX} - \eta_{ZY}$
            \State $\eta_{II} \gets s - \eta_{IX} - \eta_{IY} - \eta_{IZ}$

            \If{$\min(\eta_{XX}, \eta_{YY}, \eta_{ZZ}, \eta_{II}) < 0$} \textbf{continue} \Comment{pruning} \EndIf
            
            \State \Comment{verify column capacities~\eqref{eq:POoverlapcolsums}}
            \If{$\eta_{X\{X,Y,Z\}} + \eta_{Y\{X,Y,Z\}} + \eta_{Z\{X,Y,Z\}} + \eta_{I\{X,Y,Z\}} \neq \{p',q',r'\}$} \textbf{continue} \EndIf

            \State \Comment{verify parity condition~\eqref{eq:POoddcommutatorsites}}
            \If{$(\eta_{XY} + \eta_{XZ} + \eta_{YX} + \eta_{YZ} + \eta_{ZX} + \eta_{ZY}) \pmod 2 = 0$} \textbf{continue}  \EndIf

            \State $W(\eta)\gets$~\eqref{eq:POpatternweight} using $F$ \Comment{compute weight of overlap pattern}
            \State $\sgn(\eta)\gets$~\eqref{eq:POcommutatorSIGN} \Comment{compute sign of overlap pattern}
            \State $S_{out}[\tilde{p}, \tilde{q}, \tilde{r}] \gets S_{out}[\tilde{p}, \tilde{q}, \tilde{r}] + \sgn(\eta) \cdot W(\eta)$

        \EndFor
    \EndFor
\EndFor
\EndFor
\State \Return $S_{out}$
\end{algorithmic}
\end{myalgorithm}

\subsubsection{Numerical benchmarks for structure constants in $\g^\mathrm{PI}$}
\label{appssec:PObenchmarksstructureconstants}

Optimized versions of Algorithms~\ref{alg:fast_commutator} and \ref{alg:fast_targeted_commutator} are implemented in Python and available within the code accompanying this paper. As an end-to-end benchmark of the preprocessing bottleneck in $\gsim$, we measure the wall-clock time required to compute structure constants~\eqref{eq:structureconstants} for the Pauli orbit basis $\mathcal{B}_n$ of $\g^{\mathrm{PI}}$ (cf.~\cref{prop:orbit_basis}). Since $\dim(\g^{\mathrm{PI}})=|\mathcal{B}_n|=\binom{n+3}{3}=O(n^3)$, naive pairwise commutator evaluation scales as $O(|\mathcal{B}_n|^2\,t_\mathrm{com})=O(n^6\,t_\mathrm{com})$, where $t_\mathrm{com}$ is the cost of the commutator primitive.

We consider two preprocessing modes:
(i) \emph{full structure constants}, where commutators are evaluated via the full contingency-table enumeration (Algorithm~\ref{alg:fast_commutator}); and
(ii) \emph{targeted structure constants}, where only the coefficients required for the adjoint representations of a fixed, finite generator set are computed via output-partitioning (Algorithm~\ref{alg:fast_targeted_commutator}). In the bounded-weight regime of Corollary~\ref{cor:POcommutator_constant_time}, the latter has $n$-independent enumeration cost per targeted coefficient, so one expects an overall scaling close to the ambient $|\mathcal{B}_n|^2=O(n^6)$ dependence coming from iterating over basis pairs.

Both algorithms are dominated by integer arithmetic inside deeply nested, pruned loops. We therefore benchmark two implementations: a pure Python version and a Just-In-Time compiled version using Numba~\cite{Lam.2015-numba}. In the JIT implementation, orbit labels are packed into integers and the accumulation arrays are stored in contiguous memory, avoiding Python dictionary overhead within the inner loops.

\cref{fig:BMstructureconstantsSn} summarizes the results. For full structure constants we observe an empirical scaling exponent below $7$ over the tested range, indicating that the average commutator workload grows far more slowly than the worst-case $O(n^9)$ bound from Appendix~\ref{appssec:POcommutatorevaluationpoly}. For targeted structure constants the fitted exponent is close to $6$, consistent with the $n$-independent commutator-enumeration cost established in Appendix~\ref{appssec:POcommutatorevaluationtargeted}. Extrapolating the fitted power laws suggests that the one-time preprocessing remains practical well beyond the range tested here; we defer such extrapolations and larger-scale runs to~\cref{sec:NumericalExamples}.

\begin{figure}[htbp]
\includegraphics[width=1\linewidth]{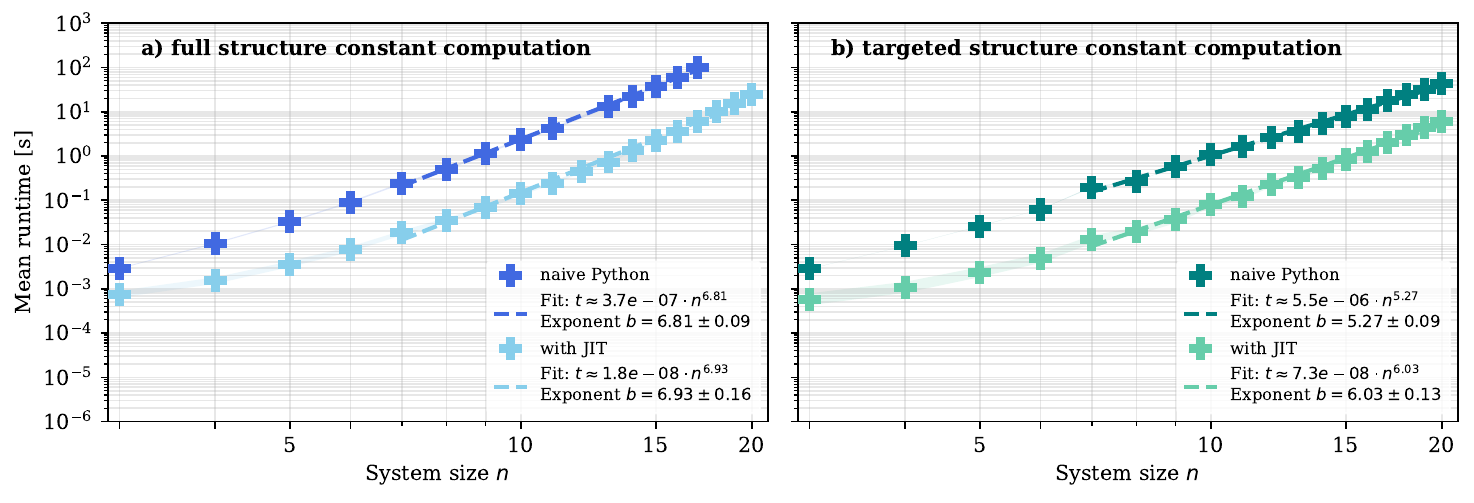}
\caption{\label{fig:BMstructureconstantsSn}
Runtime benchmarks for computing structure constants in $\g^{\mathrm{PI}}$ using Pauli orbits. Panel (a): full structure constant computation using Algorithm~\ref{alg:fast_commutator}. Panel (b): targeted computation using Algorithm~\ref{alg:fast_targeted_commutator} for a fixed bounded-weight generator set (here:~\cref{ex:SN1}). Points show mean wall-clock times versus system size $n$ on a log scale computed as averages over $100$ independent runs; dashed lines are least-squares fits of the form $t=a\,n^{b}$ performed for $n\ge 7$. The fitted exponents are $b=6.81\pm 0.09$ (full, Python) and $b=6.93\pm 0.16$ (full, JIT), and $b=5.27\pm 0.09$ (targeted, Python) and $b=6.03\pm 0.13$ (targeted, JIT). These trends are consistent with $\mathcal{O}(n^6\,t_\mathrm{com})$ preprocessing, where $t_\mathrm{com}$ is close to constant in the targeted bounded-weight regime, while remaining far below the worst-case $\mathcal{O}(n^9)$ commutator bound for generic orbit pairs. All benchmarks were run on a single CPU node with AMD EPYC 7713 processors (128 cores total) and 1\,TiB RAM per node; no GPUs were used.}
\end{figure}

Finally, several orthogonal optimizations could further reduce constants in the preprocessing time, including parallelization of the independent basis-pair commutators, heuristic reordering of loop variables (e.g., branching first on the tightest marginal constraints), caching size-independent coefficients valid once $w_1+w_2\le n$, and early filtering of cases with trivially sparse outputs. Since the current JIT-based implementation already yields practical preprocessing times, we leave these refinements to future work.

\subsection{Comparison with Schur-basis methods for simulating permutation-equivariant dynamics}
\label{appsec:schurcomparison}

Permutation symmetry can be exploited in several, ultimately equivalent, ways to obtain polynomial-time classical simulation algorithms. Here we compare the Schur-basis approaches of Refs.~\cite{Anschuetz.2023-EfficientClassicalAlgorithms,chang2026practicalframeworkpermutationequivariant} with the Pauli orbit formalism developed in this work, and clarify how each viewpoint connects to the Lie-algebraic simulation pipeline of $\gsim$.

All three approaches start from the same commutant algebra $\mathfrak{g}^\mathrm{PI}$ and, up to normalization conventions and the global factor of $i$, from the same symmetrized Pauli operators obtained by twirling Pauli strings over $S_n$. The difference lies in the representation used for computation. The Schur-basis methods map these operators to block-diagonal matrices indexed by irreducible sectors $\lambda$, whereas our approach remains in the Pauli orbit basis and computes the Lie-algebraic data needed for adjoint propagation directly at this orbit level.

Ref.~\cite{Anschuetz.2023-EfficientClassicalAlgorithms} gives a general polynomial-time Schur-basis framework in which preprocessing is dominated by the computation of the reduced matrix elements $F^{i,\lambda}_{q_\lambda,q'_\lambda}$. For qubits, this yields an $\mathcal O(n^7)$ preprocessing cost via tensor-network contraction, together with an alternative combinatorial construction of the same data. Ref.~\cite{chang2026practicalframeworkpermutationequivariant} recently improved this Schur-basis route in the practically important setting of $S_n$-equivariant circuits generated by constant-$k$-local Pauli operators: by deriving explicit sparse formulas for the Schur blocks of such generators and an $\mathcal O(n^2)$ routine for constructing them, they obtain a total simulation cost of $\mathcal O(n^3 + L n^{\omega+1})$, where $\omega$ denotes the matrix multiplication exponent, and hence $\mathcal O(n^{\omega+1})$ for constant-depth circuits.

Our contribution addresses a different task. Rather than seeking the fastest direct simulator for this single symmetry class, we show that permutation-equivariant dynamics form a natural instance of Lie-algebraic simulation via $\gsim$ by deriving a compact basis representation with efficient inner product and commutator primitives for Lie closure and structure constant generation. This makes the adjoint representation explicit and directly supports the core $\gsim$ outputs of exact expectation values, gradients, and DLA-level structural analysis within the same uniform machinery that also treats the free-fermionic and bounded-Hamming-weight families studied in this paper. For the specific task of direct simulation of constant-local $S_n$-equivariant circuits, Ref.~\cite{chang2026practicalframeworkpermutationequivariant} currently gives sharper asymptotic bounds; the value of the present construction is instead that it brings the permutation-invariant setting into the broader representation-sensitive paradigm of Lie-algebraic simulation.

\section{Details and proofs for bounded Hamming weight algebras}
\label{appsec:prfHWalgebras}

The Hamming-weight (HW) preserving operator algebra $\g^\mathrm{HW}\subseteq\mathfrak{u}(2^n)$~\eqref{eq:HWpreservingALGEBRA} is the commutant of the global $U(1)$ action generated by the total magnetization $S_Z$~\eqref{eq:totalmagnetization}, i.e., it consists of all (skew-Hermitian) operators $A$ satisfying $[A,S_Z]=0$. Equivalently, $\g^\mathrm{HW}$ captures exactly those dynamics that conserve excitation number and therefore act block-diagonally with respect to the Hamming-weight decomposition $\mathcal{H}=\bigoplus_{k=0}^n \mathcal{H}_k$, where $\dim(\mathcal{H}_k)=d_k=\binom{n}{k}$.

Section~\ref{sec:restrictedsubspace-algebras} shows that, when the dynamics are initialized in a fixed sector $\mathcal{H}_k$ with $k=\mathcal{O}(1)$, the relevant effective DLA is a subalgebra of $\mathfrak{u}(d_k)$. To enable efficient $\gsim$ preprocessing in this regime, \cref{ssec:HWalgebraBasisPrimitives} introduces the modified generalized Gell--Mann (MGGM) basis and closed-form commutation rules for computing structure constants.
This appendix collects the supporting material for Section~\ref{sec:restrictedsubspace-algebras}: (i) a derivation of the decomposition and dimension statement in Proposition~\ref{prop:HWalgebradecompdim}, (ii) a proof of the MGGM-basis representation of HW-preserving two-qubit generators in Proposition~\ref{prop:HWgeneratosinMGGMbasis}, and (iii) numerical runtime benchmarks for MGGM-based structure constant computation.

\subsection{Proof of \cref{prop:HWalgebradecompdim}}
\label{prf:HWalgebradecompdim}

\begin{proof}
Let $N$ be the excitation-number operator \eqref{eq:excitation_number}. Its spectral decomposition is
\begin{equation}
N=\sum_{k=0}^n k\,\Pi_k,
\end{equation}
where $\Pi_k$ is the orthogonal projector onto the fixed Hamming-weight sector $\mathcal{H}_k$ \eqref{eq:k-weight-subspace-dimension}, and $\sum_{k=0}^n \Pi_k = I$.

Let $A\in \g^{\mathrm{HW}}$, so $[A,N]=0$. Multiplying the commutation relation on the left by $\Pi_j$ and on the right by $\Pi_k$ yields
\begin{equation}
0=\Pi_j[A,N]\Pi_k=\Pi_j A N \Pi_k-\Pi_j N A \Pi_k.
\end{equation}
Using $N\Pi_k=k\Pi_k$ and $\Pi_j N=j\Pi_j$, we obtain
\begin{equation}
0 = k\,\Pi_j A \Pi_k - j\,\Pi_j A \Pi_k = (k-j)\,\Pi_j A \Pi_k.
\end{equation}
Hence $\Pi_j A \Pi_k=0$ for all $j\neq k$, i.e.\ $A$ is block diagonal with respect to $\mathcal{H}=\bigoplus_{k=0}^n \mathcal{H}_k$:
\begin{equation}
A=\sum_{k=0}^n \Pi_k A \Pi_k \;\equiv\; \bigoplus_{k=0}^n A^{(k)},
\qquad A^{(k)}:=\Pi_k A \Pi_k\big|_{\mathcal{H}_k}.
\end{equation}
Conversely, any block-diagonal operator $A=\bigoplus_k A^{(k)}$ with $A^{(k)}\in\mathfrak{u}(\mathcal{H}_k)\cong \mathfrak{u}(d_k)$ commutes with $N$, since $N$ acts as the scalar $k$ on $\mathcal{H}_k$. Therefore the commutant algebra is exactly the direct sum~\eqref{eq:HWalgebradecomp}
\begin{equation}
\g^{\mathrm{HW}}
\cong
\bigoplus_{k=0}^n \mathfrak{u}(d_k).
\end{equation}

For the dimension, we use $\dim(\mathfrak{u}(d_k))=d_k^2$ to obtain
\begin{equation}
\dim(\g^{\mathrm{HW}})
=
\sum_{k=0}^n d_k^2
=
\sum_{k=0}^n \binom{n}{k}^2
=
\binom{2n}{n},
\end{equation}
where the last equality is Vandermonde's identity (set $m=r=n$). Finally, Stirling's approximation $\binom{2n}{n}\sim \frac{4^n}{\sqrt{\pi n}}$ which implies the asymptotic estimate $\binom{2n}{n}=\mathcal{O}\!\left(\frac{4^n}{\sqrt{n}}\right)$~\eqref{eq:HWalgebradim}.
\end{proof}

\subsection{Proof of \cref{prop:HWgeneratosinMGGMbasis}}
\label{prf:HWgeneratosinMGGMbasis}

\begin{proof}
For any fixed pair of qubits $(i,j)$, the computational basis states in $\mathcal{H}_k$ can be partitioned into four disjoint sectors based on their local occupations at sites $i$ and $j$: the configurations $|00\rangle_{ij}$, $|11\rangle_{ij}$, $|01\rangle_{ij}$, and $|10\rangle_{ij}$.

By definition, the operators $\mathbf{E}_{ij}$, $\mathbf{S}_{ij}$, $\mathbf{R}_{ij}$, and $\mathbf{J}_{ij}$ all act trivially (evaluating to zero) on any state where the occupations are identical. Specifically, $Z_i Z_j = I$ on $|00\rangle$ and $|11\rangle$, so $\mathbf{E}_{ij} = \frac{1}{2}(I-Z_i Z_j)$ annihilates them. Similarly, $\mathbf{S}_{ij} \propto (Z_i - Z_j)$, $\mathbf{R}_{ij} \propto (X_i X_j + Y_i Y_j)$, and $\mathbf{J}_{ij} \propto (X_i Y_j - Y_i X_j)$ strictly annihilate the $|00\rangle$ and $|11\rangle$ sectors.

Consequently, the non-trivial dynamics are entirely restricted to a direct sum of isolated 2D subspaces spanned by the active pairs $C_{ij}^{(k)}$~\eqref{eq:HWactivebasepairs}. For each pair $(a, b) \in C_{ij}^{(k)}$, the states $|a\rangle$ and $|b\rangle$ are identical on the $n-2$ spectator qubits but differ exactly by the local patterns $0_i 1_j$ and $1_i 0_j$. Without loss of generality, let $|a\rangle$ correspond to the $0_i 1_j$ configuration and $|b\rangle$ to the $1_i 0_j$ configuration (consistent with the sign conventions of the chosen index ordering).

Evaluating the action of the anti-Hermitian generators on this 2D span yields exactly the localized MGGM basis matrices:
\begin{align}
    i\mathbf{E}_{ij}\big|_{\operatorname{span}(a,b)} &= i(|a\rangle\langle a| + |b\rangle\langle b|) = \Lambda_P^a + \Lambda_P^b, \\
    i\mathbf{S}_{ij}\big|_{\operatorname{span}(a,b)} &= i(|a\rangle\langle a| - |b\rangle\langle b|) = \Lambda_P^a - \Lambda_P^b, \\
    i\mathbf{R}_{ij}\big|_{\operatorname{span}(a,b)} &= i(|a\rangle\langle b| + |b\rangle\langle a|) = \Lambda_S^{ab}, \\
    i\mathbf{J}_{ij}\big|_{\operatorname{span}(a,b)} &= |a\rangle\langle b| - |b\rangle\langle a| = \Lambda_A^{ab}.
\end{align}
Summing these distinct block contributions over all active pairs $(a,b) \in C_{ij}^{(k)}$ directly proves the expansions in \cref{eq:HWgeneratortermsinMGGMbasis}.

Finally, multiplying the physical Hamiltonian $H_{ij} = e \mathbf{E}_{i j} + s \mathbf{S}_{i j} + r \mathbf{R}_{i j} + j \mathbf{J}_{i j}$ by the imaginary unit $i$ to form the Lie generator, and substituting the derived mappings, yields:
\begin{equation}
    iH_{ij}\big|_{\mathcal{H}_k} = \sum_{(a,b)\in C_{ij}^{(k)}} \left[ j \Lambda_A^{ab} + r \Lambda_S^{ab} + s(\Lambda_P^a - \Lambda_P^b) + e(\Lambda_P^a + \Lambda_P^b) \right].
\end{equation}
Grouping the diagonal terms by $\Lambda_P^a$ and $\Lambda_P^b$ produces $(s+e)\Lambda_P^a + (e-s)\Lambda_P^b$, concluding the proof.
\end{proof}

\subsection{Numerical benchmarks for structure constants in $\mathfrak{u}(d_k)$}
\label{appssec:BMstructureconstantsHW}
Efficient expectation value simulation via the adjoint representation requires access to the structure constants of the Lie algebra generated by the circuit. For HW-preserving circuits restricted to a fixed sector $\mathcal{H}_k$ of dimension $d_k=\binom{n}{k}$, the relevant ambient algebra is $\mathfrak{u}(d_k)$ (cf.\ \cref{sec:restrictedsubspace-algebras}). HW-preserving generators admit a direct decomposition into the modified generalized Gell--Mann (MGGM) basis $\mathcal{M}_{n}^{(k)}$ (cf.\ \cref{def:MGGMbasisforHWalgebras,prop:HWgeneratosinMGGMbasis}), so $\gsim$-preprocessing reduces to generating the structure constants of $\mathfrak{u}(d_k)$ in this basis. As shown in \cref{th:HWstructureconstants-dk3}, the complete set of nonzero structure constants can be enumerated in time $\mathcal{O}(d_k^3)$.

\cref{fig:BMstructureconstantsHW} reports wall-clock runtimes for constructing the full adjoint data in the MGGM basis for $k=1,2,3$, comparing a straightforward Python implementation to a Just-In-Time (JIT) compiled variant using Numba~\cite{Lam.2015-numba}. Since the MGGM commutator relations \eqref{eq:MGGM_comm_relations} reduce all commutators to sparse integer-index manipulations, the JIT version yields substantial constant-factor improvements and, in these benchmarks, also improved fitted scaling exponents. Overall, the observed runtimes indicate that $\gsim$ is practical for low-weight sectors (notably $k=1$ and $k=2$) despite the rapidly increasing worst-case scaling with $k$. 

Using the fitted power laws $t(n)=a\,n^{b}$ from \cref{fig:BMstructureconstantsHW} as a rough extrapolation beyond the benchmarked range (and keeping in mind that such extrapolations might ignore some large-scale effects), we obtain the following back-of-the-envelope estimates for our current, not fully optimized code: for $k=1$ the fit suggests $t<1\,\mathrm{s}$ up to approximately $n\approx 470$, while for $k=2$ it suggests $t\approx 2\,\mathrm{h}$ at $n=100$. These estimates are intended only as practical guidance; further constant-factor improvements are likely (see Appendix~\ref{appssec:PObenchmarksstructureconstants}).
 
\begin{figure}[htbp]
\includegraphics[width=1\linewidth]{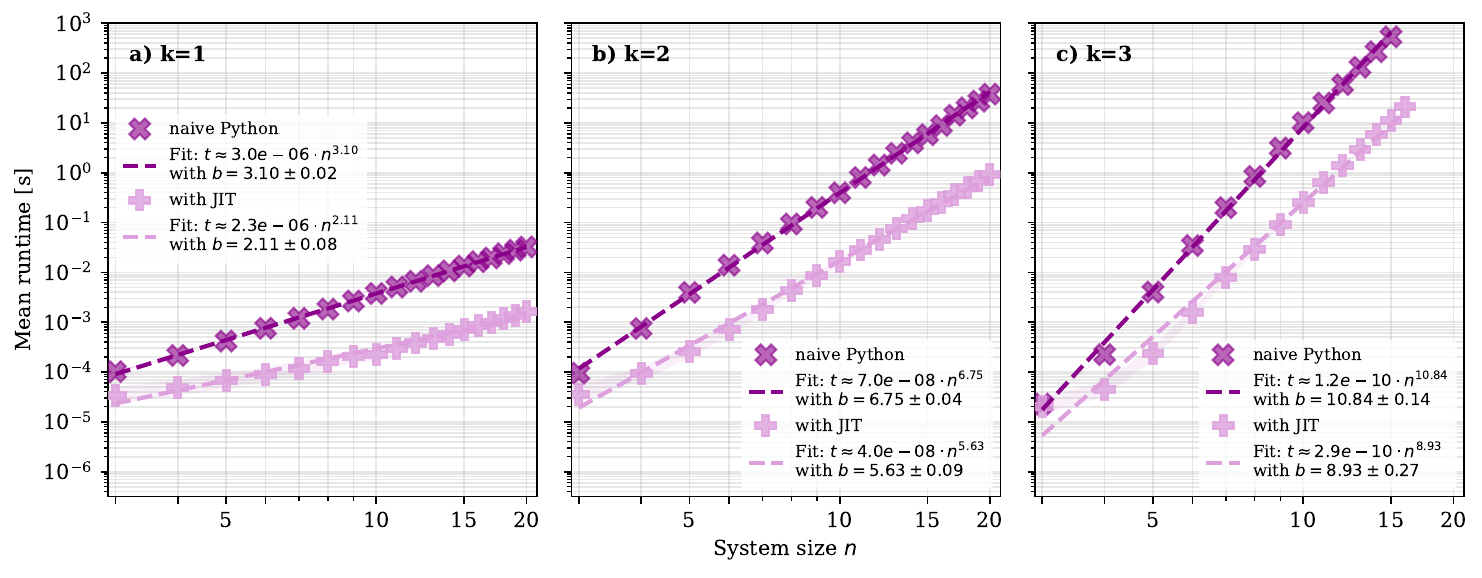}
\caption{\label{fig:BMstructureconstantsHW}
Runtime benchmarks for computing the adjoint data (via structure constant enumeration) of $\mathfrak{u}(d_k)$ in the MGGM basis $\mathcal{M}_{n}^{(k)}$ (cf.\ \cref{def:MGGMbasisforHWalgebras}). Panels (a)--(c) show fixed Hamming weight $k=1,2,3$, where $d_k=\binom{n}{k}$. Dark markers denote a naive Python implementation and light markers the JIT-compiled implementation (via Numba); dashed lines are least-squares fits of the form $t=a\,n^{b}$. The fitted exponents (naive / JIT) are: $b=3.10\pm0.02$ / $2.11\pm0.08$ for $k=1$, $b=6.75\pm0.04$ / $5.63\pm0.09$ for $k=2$, and $b=10.84\pm0.14$ / $8.93\pm0.27$ for $k=3$. Points show mean wall-clock times over $100$ independent runs (log--log scale). All benchmarks were run on a single CPU node with AMD EPYC 7713 processors (128 cores total) and 1\,TiB RAM per node; no GPUs were used.}
\end{figure}

\end{document}